\documentclass[10pt,letterpaper,oneside,reqno]{amsart}
\usepackage[top=1in, bottom=1in, right=1in, left=1.5in,includefoot]{geometry}
\usepackage{datetime}
\usepackage{relsize}
\usepackage{pdflscape}
\usepackage{microtype}
\usepackage{enumitem}
\usepackage{setspace}
\usepackage[utf8]{inputenc}
\usepackage[english]{babel}
\usepackage{bm}
\usepackage{pgfplots}
\usepackage{amsmath,amsthm}
\usepackage{amsfonts}
\usepackage{amssymb}
\usepackage{graphicx} 
\usepackage{lmodern} 
\usepackage{tikz-cd}
\usetikzlibrary{nfold}
\usepackage{mathtools}
\usepackage{xfrac}
\usepackage{bm}
\usepackage{xcolor}
\usepackage{soul}
\usepackage{mathrsfs}
\usepackage{hyperref}
\usepackage{xurl}
\hypersetup{hypertexnames=false, breaklinks=true}
\hypersetup{colorlinks=true, linkcolor=black,          
    citecolor=black,        
    filecolor=black,      
    urlcolor=gray           }
\usepackage{thmtools}
\usepackage[capitalise]{cleveref}
\usepackage{comment}
\usepackage{tcolorbox}
\usepackage{quiver}
\usepackage{mathdots}
\usepackage{calligra}
\usepackage{adjustbox}
\usepackage{mleftright}
\usepackage{wasysym}
\usepackage{etoolbox}
\usepackage{wrapfig}
\theoremstyle{theorem}
\newtheorem{theorem}{Theorem}[subsection]
\newtheorem{corollary}[theorem]{Corollary}
\newtheorem{lemma}[theorem]{Lemma}

\newtheorem{prop}[theorem]{Proposition}

\newlist{propenum}{enumerate}{1}
\setlist[propenum]{label=(\arabic*), ref=\thetheorem(\arabic*)}
\crefname{propenumi}{Proposition}{Propositions}

\theoremstyle{definition}

\newtheorem{question}[theorem]{Question}
\newtheorem{example}[theorem]{Example}
\newtheorem{definition}[theorem]{Definition}
\newtheorem{note}[theorem]{Note}
\newlist{exmpenum}{enumerate}{1}
\setlist[exmpenum]{label=(\arabic*), ref=\thetheorem(\arabic*)}
\crefname{exmpenumi}{Example}{Examples}

\theoremstyle{remark}
\newtheorem{remark}{Remark}
\newtheorem*{notation}{Notation}

\newlist{steps}{enumerate}{1}
\setlist[steps, 1]{label = Step \arabic*:}

\DeclareMathOperator{\F}{\mathsf{Fg}}
\newcommand{\A}{\mathcal A}
\newcommand{\C}{\mathcal C}
\newcommand{\D}{\mathcal D}

\newcommand{\U}{\mathcal{U}}
\newcommand{\V}{\mathcal{V}}

\newcommand{\K}{\mathcal K}
\renewcommand{\L}{\mathcal L}
\newcommand{\RI}{\mathcal R}

\renewcommand{\P}{\mathcal P}

\renewcommand{\S}{\mathcal S}

\newcommand{\Z}{\mathbb Z}
\newcommand{\B}{\mathcal{B}}

\newcommand{\modal}{{\ensuremath{\ocircle}}}

\newcommand{\N}{\mathbb N}

\DeclareMathOperator{\cogap}{\mathsf{cogap}}

\DeclareMathOperator{\const}{\mathsf{const}}

\DeclareMathOperator{\im}{im}
\DeclareMathOperator{\rid}{\mathsf{idr}}
\DeclareMathOperator{\lid}{\mathsf{idl}}

\DeclareMathOperator{\postcomp}{\mathsf{postcomp}}
\DeclareMathOperator{\pstc}{\mathsf{pstc}}

\DeclareMathOperator{\ob}{Ob}
\DeclareMathOperator{\PW}{\mathsf{PW}}
\DeclareMathOperator{\obb}{\mathsf{Ob}}

\newcommand{\fact}{\mathsf{fact}}

\newcommand{\refl}{\mathsf{refl}}
\newcommand{\ind}{\mathsf{ind}}
\newcommand{\inl}{\mathsf{inl}}
\newcommand{\inr}{\mathsf{inr}}

\newcommand{\0}{\mathbf{0}}
\newcommand{\1}{\mathbf{1}}
\newcommand{\2}{\mathbf{2}}

\DeclareMathOperator{\filll}{\mathsf{fill}}
\DeclareMathOperator{\cofiber}{\mathsf{cof}}

\DeclareMathOperator{\assoc}{\mathsf{assoc}}

\DeclareMathOperator{\tip}{\mathsf{tip}}

\DeclareMathOperator{\transport}{\mathsf{transp}}
\DeclareMathOperator{\ap}{\mathsf{ap}}
\DeclareMathOperator{\apd}{\mathsf{apd}}
\DeclareMathOperator{\rec}{\mathsf{rec}}

\DeclareMathOperator{\happly}{\mathsf{happly}}

\DeclareMathOperator{\isequiv}{\mathsf{is\_equiv}}

\DeclareMathOperator{\pr}{\mathsf{pr}}
\DeclareMathOperator{\glue}{\mathsf{glue}}

\DeclareMathOperator{\homm}{\mathsf{hom}}
\DeclareMathOperator{\op}{op}

\DeclareMathOperator{\colimm}{\mathsf{colim}}
\DeclareMathOperator{\limm}{\mathsf{lim}}
\DeclareMathOperator{\hocolimm}{hocolim}
\DeclareMathOperator{\holimm}{holim}

\newcommand{\lN}{\left\lVert}
\newcommand{\rN}{\right\rVert}

\DeclareMathOperator{\idd}{\mathsf{id}}
\DeclareMathOperator{\fib}{\mathsf{fib}}

\DeclareMathOperator{\ho}{Ho}
\DeclareMathOperator{\PI}{\mathsf{PI}}

\makeatletter
\newcommand{\changeoperator}[1]{%
  \csletcs{#1@saved}{#1@}%
  \csdef{#1@}{\changed@operator{#1}}%
}
\newcommand{\changed@operator}[1]{%
  \mathop{%
    \mathchoice{\textstyle\csuse{#1@saved}}
               {\csuse{#1@saved}}
               {\csuse{#1@saved}}
               {\csuse{#1@saved}}%
  }%
}
\makeatother

\changeoperator{sum}
\changeoperator{prod}
\changeoperator{bigvee}

\makeatletter
\newcommand{\hocolim}{\gen@colim{\hocolimm}}
\newcommand{\gen@hocolim}[1]{%
  \@ifnextchar_{\gen@@hocolim{#1}}{\mathbin{#1}}%
}
\def\gen@@hocolim#1_#2{%
  \mathpalette\gen@@@hocolim{{#1}{#2}}%
}
\newcommand\gen@@@hocolim[2]{\mathbin{\gen@@@@hocolim#1#2}}
\newcommand\gen@@@@hocolim[3]{%
  \ifx#1\displaystyle
    \mathop{#2}\limits_{#3}%
  \else
    {#2}_{#3}%
  \fi
}
\makeatother

\makeatletter
\newcommand{\holim}{\gen@colim{\holimm}}
\newcommand{\gen@holim}[1]{%
  \@ifnextchar_{\gen@@holim{#1}}{\mathbin{#1}}%
}
\def\gen@@holim#1_#2{%
  \mathpalette\gen@@@holim{{#1}{#2}}%
}
\newcommand\gen@@@holim[2]{\mathbin{\gen@@@@holim#1#2}}
\newcommand\gen@@@@holim[3]{%
  \ifx#1\displaystyle
    \mathop{#2}\limits_{#3}%
  \else
    {#2}_{#3}%
  \fi
}
\makeatother

\makeatletter
\newcommand{\subalign}[1]{%
  \vcenter{%
    \Let@ \restore@math@cr \default@tag
    \baselineskip\fontdimen10 \scriptfont\tw@
    \advance\baselineskip\fontdimen12 \scriptfont\tw@
    \lineskip\thr@@\fontdimen8 \scriptfont\thr@@
    \lineskiplimit\lineskip
    \ialign{\hfil$\m@th\scriptstyle##$&$\m@th\scriptstyle{}##$\hfil\crcr
      #1\crcr
    }%
  }%
}
\makeatother

\begin{document}

\title{Coslice Colimits in Homotopy Type Theory}

\author{Perry Hart}
\address{University of Minnesota, Twin Cities}
\email{hart1262@umn.edu}

\author{Kuen-Bang Hou (Favonia)}
\address{University of Minnesota, Twin Cities}
\email{kbh@umn.edu}

\date{\today}

\thanks{
This material is based upon work supported by the Air Force Office of Scientific Research under award number FA9550-21-1-0009. Any opinions, findings, and conclusions or recommendations expressed in this material are those of the author(s) and do not necessarily reflect the views of the US Air Force.
}

\maketitle

\begin{abstract}
  We contribute to the theory of (homotopy) colimits inside homotopy type theory. The heart of our work characterizes the connection between (graph-indexed) colimits in a type universe and colimits in coslices of the universe, called \emph{coslice colimits}. To derive this characterization, we give a construction of coslice colimits that is tailored to reveal the connection. We use the construction to prove that the forgetful functor from a coslice creates colimits over trees. We also use it to study how coslice colimits interact with orthogonal factorization systems and with cohomology theories. As a result of their interaction with orthogonal factorization systems, all colimits of pointed types preserve $n$-connectedness, which implies that higher groups---in the sense of Buchholtz, van Doorn, and Rijke---are closed under colimits. We have formalized major portions of this work in Agda (available \href{https://github.com/PHart3/colimits-agda/tree/v0.5.0}{here}), including our main construction of the coslice colimit functor.
\end{abstract}

\tableofcontents

\section{Introduction}\label{sec:intro}

Working in homotopy type theory (HoTT)~\cite{Uni13}, we study higher inductive types (HITs) arising as colimits (over graphs) in coslices of a universe, called \emph{coslice colimits}. Coslices of a universe are type-theoretic versions of coslice categories. Our study of coslice colimits is organized as follows.

\subsubsection*{The main connection (\cref{Constr})}

Suppose $\U$ is a universe and $A$ is a type in $\U$. We want to construct all colimits in $A/\U$, or \emph{$A$-colimits}. The wild category $A/\U$ has objects $\sum_{T : \U}A \to T$ with $X \to_A Y \coloneqq \sum_{k : \pr_1(X) \to \pr_1(Y)}k \circ \pr_2(X) \sim \pr_2(Y)$ as morphisms from $X$ to $Y$. HoTT has a general schema for HITs that would let us simply postulate $A$-colimits. We, however, explicitly construct $A$-colimits with just the machinery of Martin-L\"of type theory (MLTT) augmented with pushouts. We take this different approach to reveal the connection between $A$-colimits and their underlying colimits in $\U$. Therefore, we call the construction of \cref{Constr} the \emph{main connection}. In fact, our construction is \emph{not} a case of a general method to encode higher-dimensional HITs with pushouts (such as \cite[Section 5]{DRB}) but rather tailored to reveal this connection.

This connection sheds light on existing topics in synthetic homotopy theory, which we discuss now.

\subsubsection*{The universality of colimits (\cref{App1})}

The \emph{universality} of colimits is the defining feature of locally cartesian closed (LCC) $\infty$-categories, such as that of spaces. The main connection will establish a well-known classical result inside type theory: The forgetful functor $A/\U \to \U$ \emph{creates} (preserves and reflects) colimits of diagrams over contractible graphs (\cref{create}). We review such graphs, known as \emph{trees}, in \cref{trees}. With the forgetful functor creating colimits, we can transfer universality of ordinary colimits to coslice colimits in many cases (\cref{univcolim}). This result is notable as LCC $\infty$-categories are not closed under coslices.

\subsubsection*{Categories of higher groups are cocomplete (\cref{Presv})}

A striking feature of colimits is their interaction with (orthogonal) factorization systems. In \cref{Presv}, we use the main connection to show that colimits in $A/\U$ preserve left classes of maps of such systems on $\U$. It is significant that we consider factorization systems on $\U$ rather than $A/\U$. We could derive a similar preservation theorem for factorization systems on $A/\U$ directly from the universal property of an $A$-colimit. In practice, however, the factorization systems we tend to care about are on $\U$. Since the main connection relates the action of $A$-colimits on maps to the action of their underlying colimits on maps, we manage to deduce the preservation theorem for factorization systems on $\U$.

To prove this theorem, we find it useful to develop the theory of factorization systems in a more general setting than $\U$. In \cref{OFS}, we study such systems on \emph{wild categories}, which is the appropriate categorical framework for synthetic homotopy theory. We prove that if a functor $F$ of well-behaved wild categories with factorization systems has a right adjoint $G$, then---under a reasonable coherence condition on the adjunction---$F$ preserves the left class when $G$ preserves the right class (\cref{fspres}). We combine this result with the main connection to deduce the desired preservation property.

When we focus on the ($n$-connected, $n$-truncated) system on $\U$~\cite[Section 7.6]{Uni13} and take $A$ as the unit type, the main connection shows that the colimit of every diagram of pointed $n$-connected types is $n$-connected. One useful corollary of this is that the higher category ${\left(n,k\right)\mathsf{GType}}$ of $k$-tuply groupal $n$-groupoids considered by~\cite{Higher} is cocomplete (over graphs) for all truncation levels ${-2} \leq n \leq \infty$ and ${-1} \leq k < \infty$ (\cref{cocomp}). We also exploit the generality of the main connection to extend this cocompleteness result to categories of \emph{higher pointed abelian groups} (\cref{pthicol}).

\subsubsection*{Cohomology sends colimits to weak limits (\cref{contcohom})}

Finally, we examine how colimits interact with cohomology theories, which are important algebraic invariants of spaces. To do so, we consider \emph{weak limits}, which are key ingredients in the Brown representability theorem. A weak colimit in a category need not satisfy the uniqueness property required of a colimit. The Brown representability theorem specifies conditions for a presheaf on the homotopy category $\ho(\mathbf{Top}_{\ast,c})$ of pointed connected spaces to be representable. The known proof of this theorem requires the presheaf to send countable homotopy colimits in $\mathbf{Top}_{\ast,c}$ to weak limits in $\mathbf{Set}$. Eilenberg-Steenrod cohomology theories enjoy this property as set-valued functors.

In \cref{Contr2}, we use the main connection to establish a restricted, type-theoretic version of this property. From the main connection we derive another construction of $A$-colimits, as pushouts of coproducts (\cref{coeqcop}), which mirrors a well-known classical lemma. In \cref{cohomwklim}, we take $A$ as the unit type and combine the new construction with the Mayer-Vietoris sequence to find that cohomology takes finite colimits to weak limits assuming the internal axiom of choice.

\subsection*{Agda formalization}

Major portions of this work are mechanized in our Agda library on colimits and adjunctions~\cite{agda-colim-TR}, including but not limited to
\begin{itemize} 
\item the main construction of the coslice colimit functor as a left adjoint to the constant diagram functor (\cref{Constr})
\item the creation of colimits by the forgetful functor (\cref{create})
\item the fact that the coslice colimit functor preserves the left class of an orthogonal factorization system on a universe (\cref{Presv}).
\end{itemize} 
We will provide links to Agda code at relevant points in the paper.

\section{Type system}

We assume the reader is familiar with MLTT and HITs in the style of~\cite{Uni13}.
We primarily work in MLTT augmented with ordinary colimits, i.e., colimits in a universe, and denote this system by  $\mathsf{MLTT + Colim}$. (We review ordinary colimits in \cref{ordcol}.) In particular, all of \cref{Constr} takes place inside $\mathsf{MLTT + Colim}$. In fact, we need only augment MLTT with pushouts as they let us construct all nonrecursive $1$-HITs, including ordinary colimits, with all of their computational properties.
Notably, $\mathsf{MLTT + Colim}$ comes with strong function extensionality for free.

\subsection*{Remark on notation}

We point out two important conventions that we use throughout the paper.

\begin{itemize}
\item The symbol $=$ denotes the identity/path type. The symbol $\equiv$ denotes definitional equality. The symbol $\coloneqq$ denotes term definition.
\item For convenience, we may use the notation 
$\PI(p_1, \ldots, p_n) :  a = b$ to denote a path obtained by simultaneous or iterative path induction on paths $p_1$, \ldots, $p_n$. We only use this shorthand when the identity is constructed in an evident way.
\end{itemize}

\section{Categorical background} 

\subsection{Wild categories and functors}

\begin{definition}  Let $\U$ be a universe. A \textit{wild category} consists of a type $\obb$ of objects, a type family $\homm$ of morphisms twice indexed over $\obb$, and the following data:
\begin{itemize}
\item a composition operation $\circ  : \homm(Y,Z) \to \homm(X,Y) \to \homm(X,Z)$ for all objects $X,Y, Z$
\item identity morphisms  $\idd_X  : \homm(X,X)$ for every  object $X$
\item a path $\rid  : f \circ \idd_X = f$ for all morphisms $f : \homm(X ,Y)$
\item a path  $\lid : \idd_Y \circ f = f$ for all morphisms $f : \homm(X ,Y)$
\item a path $\assoc(h,g,f) : \left(h \circ g\right) \circ f = h \circ \left(g \circ f\right)$ for all composable morphisms $h$, $g$, and $f$.
\end{itemize}
\end{definition}

\begin{definition} \label{bistr}
A \textit{wild bicategory} is a wild category $\C$ equipped with 
\begin{itemize}
\item a \emph{triangle} identity 
\[\begin{tikzcd}[row sep = small]
	{\assoc(h, \idd, g ) \cdot \ap_{h \circ {-}}(\lid(g))} \\
	{\ap_{{-} \circ g}(\rid(h))}
	\arrow[Rightarrow, no head, from=1-1, to=2-1]
\end{tikzcd}\] for all composable morphisms $h$ and $g$
\item a \emph{pentagon} identity
\[\begin{tikzcd}[row sep = small]
	{\ap_{{-} \circ f}(\assoc(k,h, g)) \cdot \assoc(k, h \circ g,f) \cdot \ap_{k \circ {-}}(\assoc(h, g,f))} \\
	{ \assoc(k \circ h, g, f)  \cdot \assoc(k, h, g \circ f)  }
	\arrow[Rightarrow, no head, from=1-1, to=2-1]
\end{tikzcd}\]
for all composable morphisms $k$, $h$, $g$, and $f$.
\end{itemize}
Here, a wild bicategory is really a wild $\mleft(2,1\mright)$-category as its $2$-cells are paths, hence invertible.
\end{definition}

\begin{lemma} \label{kellycoh}
Let $\C$ be a bicategory. For all $A, B, C : \obb(\C)$, $f : \homm_{\C}(A,B)$, $g : \homm_{\C}(B,C)$,
\[
 \ap_{{-} \circ  f}(\lid(g)) \ = \ \assoc(\idd, g, f) \cdot \lid(g \circ f) 
\]
\end{lemma}
\begin{proof}
As $ \left(c = d \right) \xrightarrow{\ap_{\idd \circ {-}} } \left(\idd \circ c = \idd \circ d \right) $ is an equivalence for all morphisms $c$ and $d$, it suffices to prove 
$ \ap_{\idd \circ {-}} (\ap_{{-} \circ  f}(\lid(g))) = \ap_{\idd \circ {-}} (\assoc(\idd, g, f)) \cdot \ap_{\idd \circ {-}} (\lid(g \circ f))$. Consider the diagram
\[
\hspace*{-4mm}\begin{tikzcd}[column sep =60, row sep = 25]
	{\left(\left(\idd \circ \idd\right)\circ g \right) \circ f} & {\left( \idd \circ \left(\idd \circ g\right) \right) \circ f} & {\idd \circ \left( \left( \idd \circ g \right) \circ f\right)} & {\idd \circ \left(\idd \circ \left( g \circ f\right)\right)} \\
	& {\left( \idd \circ g\right) \circ f} & {\idd \circ \left(g \circ f\right)}
	\arrow["{\ap_{\idd \circ {-}}(\lid(g \circ f))}", curve={height=-6pt}, Rightarrow, no head, from=1-4, to=2-3]
	\arrow["{\ap_{\idd \circ {-}}(\ap_{{-} \circ f}(\lid(g)))}"{description}, Rightarrow, no head, from=1-3, to=2-3]
	\arrow["{\ap_{\idd \circ {-}}(\assoc(\idd, g, f))}", shift left, Rightarrow, no head, from=1-3, to=1-4]
	\arrow["{\assoc(\idd, \idd \circ g, f)}", shift left, Rightarrow, no head, from=1-2, to=1-3]
	\arrow["{\ap_{{-} \circ f}(\assoc(\idd, \idd, g))}", shift left, Rightarrow, no head, from=1-1, to=1-2]
	\arrow["{\ap_{{-} \circ f}(\ap_{{-} \circ g}(\rid(\idd)))}"', curve={height=6pt}, Rightarrow, no head, from=1-1, to=2-2]
	\arrow["{\ap_{{-} \circ f}(\ap_{\idd \circ {-}}(\lid(g)))}"{description}, Rightarrow, no head, from=1-2, to=2-2]
	\arrow["{\assoc(\idd, g, f)}"', Rightarrow, no head, from=2-2, to=2-3]
\end{tikzcd}\] Its left two subdiagrams commute, and we want to prove  the right one commutes. Hence it suffices to prove that this diagram's outer perimeter commutes. This follows from the commuting diagram
\[\begin{tikzcd}[column sep = 50, row sep = 25]
	& {\left( \idd \circ \left(\idd \circ g\right) \right) \circ f} & {\idd \circ \left( \left( \idd \circ g \right) \circ f\right)} \\
	{\left(\left(\idd \circ \idd\right)\circ g \right) \circ f} & {\left(\idd \circ \idd\right)\circ  \left( g  \circ f \right)} & {\idd \circ \left(\idd \circ \left( g \circ f\right)\right)} \\
	{\left( \idd \circ g\right) \circ f} & {\idd \circ \left(g \circ f\right)}
	\arrow["{\ap_{{-} \circ f}(\assoc(\idd, \idd, g))}", curve={height=-12pt}, Rightarrow, no head, from=2-1, to=1-2]
	\arrow["{\assoc(\idd, \idd \circ g, f)}", Rightarrow, no head, from=1-2, to=1-3]
	\arrow["{\ap_{\idd \circ {-}}(\assoc(\idd, g, f))}", Rightarrow, no head, from=1-3, to=2-3]
	\arrow["{\ap_{\idd \circ {-}}(\lid(g \circ f))}", curve={height=-12pt}, Rightarrow, no head, from=2-3, to=3-2]
	\arrow["{\assoc(\idd, g, f)}"', Rightarrow, no head, from=3-1, to=3-2]
	\arrow["{\ap_{{-} \circ f}(\ap_{{-} \circ g}(\rid(\idd)))}"', Rightarrow, no head, from=2-1, to=3-1]
	\arrow["{\assoc(\idd \circ \idd, g, f)}"', Rightarrow, no head, from=2-1, to=2-2]
	\arrow["{\ap_{{-} \circ \left( g \circ f\right)}(\rid(\idd))}"{description}, Rightarrow, no head, from=2-2, to=3-2]
	\arrow["{\assoc(\idd, \idd, g \circ f)}"', Rightarrow, no head, from=2-2, to=2-3]
\end{tikzcd}\]
\end{proof}

\begin{definition} \label{equivwild}
Let $\C$ be a wild category.
\begin{enumerate}
\item A morphism $f : \homm_{\C}(A,B)$ of $\C$ is an \textit{equivalence} if it is biinvertible, i.e.,
$\isequiv(f) \coloneqq \sum_{g,h : \homm_{\C}(B,A)}\left(g \circ f = \idd_A \right) \times \left(f \circ h = \idd_B\right)$. (Note that $\isequiv(f)$ is a proposition.)
We write $\simeq_{\C}$ for the type of equivalences.
\item We say that $\C$ is \textit{univalent} if for all $A, B : \obb(\C)$, the function
$\left(A =_{\obb(\C)} B\right) \to \left(A \simeq_{\C} B\right)$ sending $\refl_A$ to $\left(\idd_A, \idd_A, \idd_A, \lid(\idd_A), \lid(\idd_A)\right)$ is an equivalence.
\end{enumerate}
\end{definition}

\begin{example} 
The universe $\U$ is a bicategory and (given the univalence axiom) is univalent.
\end{example}

\begin{definition}
Let $\C$ and $\D$ be wild categories.
\begin{enumerate}
\item A \emph{wild functor} $F : \C \to \D$ from $\C$ to $\D$ is a tuple consisting of
\begin{align*}
F_0 & \ : \  \obb(\C) \to \obb(\D)
\\ F_1 & \ : \ \prod_{X,Y : \obb(\C)}\homm_{\C}(X,Y) \to \homm_{\D}(F_0(X), F_0(Y))
\\ F_{\circ} & \ : \ \prod_{X,Y,Z : \obb(\C)}\prod_{g : \homm_{\C}(Y,Z)}\prod_{f : \homm_{\C}(X,Y)} F_1(g \circ f) = F_1(g) \circ F_1(f)
\\ F_{\idd} & \ : \ \prod_{X : \obb(\C)}F_1(\idd_X) = \idd_{F_0(X)}
\end{align*}
We may refer to $F_0$ or $F_1$ by $F$. If the data $F_{\circ}$ and $F_{\idd}$ are omitted, then we call $F$ a \emph{$0$-functor}.
\item Let $F, G : \C \to D$ be $0$-functors. A \emph{natural transformation} $\tau : F \to G$ from $F$ to $G$ consists of 
\begin{align*}
\tau_0 & \ : \ \prod_{X : \obb(\C)}\homm_{\D}(F(X), G(X))
\\ \tau_1 & \ : \  \prod_{X, Y : \obb(\C)}\prod_{f : \homm_{\C}(X , Y)}G(f) \circ \tau_0(X) = \tau_0(Y) \circ F(f)
\end{align*}
We say that $\tau$ is a \emph{(natural) isomorphism} if $\tau_0(X)$ is an equivalence for each $X : \obb(C)$.  
\end{enumerate}
\end{definition}

Let $L : \C \to \D$ and $R: \D \to \C$ be $0$-functors of wild categories. 
\begin{definition}\label{adjdef}
An adjunction $L \dashv R$ consists of
\begin{align*}
\alpha & \ : \ \prod_{A : \obb(\C)}\prod_{ X : \obb(\D)}\homm_{\D}(L{A}, X) \simeq \homm_{\C}(A,R{X})
\\  V_1 & \ : \ \prod_{A :\obb(\C)}\prod_{ X, Y : \obb(\D)}\prod_{g : \homm_{\D}(X,Y)}\prod_{h : \homm_{\D}(L{A}, X)}R{g} \circ \alpha(h) = \alpha(g \circ h)
\\ V_2 & \ : \ \prod_{Y : \obb(\D)}\prod_{ A, B : \obb(\C)}\prod_{f : \homm_{\C}(A,B)}\prod_{h : \homm_{\D}(L{B}, Y)}\alpha(h) \circ f = \alpha(h \circ L{f}).
\end{align*} 
\end{definition}
\noindent Note that for each such triple, we also have naturality squares
\[\begin{tikzcd}[column sep = 27]
	{\homm_{\C}(A, R{X})} & {\homm_{\C}(A, R{Y})} & {\homm_{\C}(B, R{Y})} & {\homm_{\C}(A, R{Y})} \\
	{\homm_{\D}(L{A}, X)} & {\homm_{\D}(L{A}, Y)} & {\homm_{\D}(L{B}, Y)} & {\homm_{\D}(L{A},Y)}
	\arrow[""{name=0, anchor=center, inner sep=0}, "{{R{g} \circ {-}}}", from=1-1, to=1-2]
	\arrow["{{\alpha^{-1}}}"', from=1-1, to=2-1]
	\arrow["{{\alpha^{-1}}}", from=1-2, to=2-2]
	\arrow[""{name=1, anchor=center, inner sep=0}, "{{{-} \circ f}}", from=1-3, to=1-4]
	\arrow["{{\alpha^{-1}}}"', from=1-3, to=2-3]
	\arrow["{{\alpha^{-1}}}", from=1-4, to=2-4]
	\arrow[""{name=2, anchor=center, inner sep=0}, "{{g \circ {-}}}"', from=2-1, to=2-2]
	\arrow[""{name=3, anchor=center, inner sep=0}, "{{{-} \circ L{f}}}"', from=2-3, to=2-4]
	\arrow["{{\widetilde{V}_1(g)}}"{description}, draw=none, from=0, to=2]
	\arrow["{{\widetilde{V}_2(f)}}"{description}, draw=none, from=1, to=3]
\end{tikzcd}\]
Here, the homotopies witnessing these squares commute are defined by the commuting squares
\[
\begin{tikzcd}[column sep = large]
	{g \circ \alpha^{-1}(h)} && {\alpha^{-1}(R{g} \circ h)} \\
	{\alpha^{-1}(\alpha(g \circ \alpha^{-1}(h)))} && {\alpha^{-1}(Rg \circ \alpha(\alpha^{-1}(h)))} \\
	{\alpha^{-1}(h) \circ L{f}} && {\alpha^{-1}(h \circ f)} \\
	{\alpha^{-1}(\alpha(\alpha^{-1}(h) \circ L{f}))} && {\alpha^{-1}(\alpha (\alpha^{-1}(h)) \circ f)}
	\arrow["{\widetilde{V}_1(g,h)}", Rightarrow, no head, from=1-1, to=1-3]
	\arrow["{\eta_{\alpha}(g \circ \alpha^{-1}(h))}"', Rightarrow, no head, from=1-1, to=2-1]
	\arrow["{\ap_{\alpha^{-1}}(V_1(g, \alpha^{-1}(h)))}"', Rightarrow, no head, from=2-1, to=2-3]
	\arrow["{\ap_{\alpha^{-1}}(\ap_{R{g} \circ {-}}(\epsilon_{\alpha}(h)))}"', Rightarrow, no head, from=2-3, to=1-3]
	\arrow["{\widetilde{V}_2(f,h)}", Rightarrow, no head, from=3-1, to=3-3]
	\arrow["{\eta_{\alpha}(\alpha^{-1}(h) \circ L{f})}"', Rightarrow, no head, from=3-1, to=4-1]
	\arrow["{\ap_{\alpha^{-1}}(V_2(f, \alpha^{-1}(h)))}"', Rightarrow, no head, from=4-1, to=4-3]
	\arrow["{\ap_{\alpha^{-1}}(\ap_{{-} \circ f}(\epsilon_{\alpha}(h)))}"', Rightarrow, no head, from=4-3, to=3-3]
\end{tikzcd}
\]
where $\eta_{\alpha}$ and $\epsilon_{\alpha}$ refer to the unit and counit, respectively, of $\alpha$'s half-adjoint equivalence data. 

The \emph{counit} of an adjunction $ \left(\alpha, V_1, V_2\right)$ is the natural transformation $\epsilon : L \circ R \to \idd_{\D}$ defined component-wise by $\epsilon_X \coloneqq \alpha^{-1}(\idd_{R{X}})$. This family of morphisms is natural by the naturality of $\alpha^{-1}$ (i.e., $\widetilde{V}_1$ and $\widetilde{V}_2$): for every morphism $f : \homm_{\D}(X ,Y)$, 
\[
f \circ \epsilon_X = \alpha^{-1}(R(f) \circ \idd_{R{X}}) = \alpha^{-1}(\idd_{R{Y}} \circ R(f)) =  \epsilon_Y \circ L(R(f))
\] For each $X : \obb(\D)$, the \emph{zigzag identity (at $X$)} of the adjunction is the chain of paths
\[
\label{zigzag}
\begin{tikzcd}[column sep = 50]
	{R(\epsilon_X) \circ \alpha(\idd_{L(R(X))})} & {\alpha(\epsilon_X \circ \idd_{L(R(X))})} & {\alpha(\epsilon_X)} & {\idd_{R{X}}}
	\arrow["{{\textit{naturality of $\alpha$}}}", equals, from=1-1, to=1-2]
	\arrow["{{\textit{unit law}}}", equals, from=1-2, to=1-3]
	\arrow["{{\textit{def.\ of inverse}}}", equals, from=1-3, to=1-4]
\end{tikzcd}
\tag{$\mathtt{zz\text{-}counit}$}
\]

\subsection{Reflective subcategories}

\begin{definition} \label{reflsubdef}
Let $\C$ be a wild category. A \textit{pre-reflective subcategory} of $\C$ is a predicate $P : \obb(\C) \to \mathsf{Prop}$ together with functions
\[
\modal  \ : \ \obb(\C) \to \obb(\C) \qquad  \eta \ : \ \prod_{X : \obb(\C)}\homm_{\C}(X, \modal{X})
\] called the \emph{modal operator} and \emph{modal unit}, respectively,
such that 
\begin{itemize}
\item for each $X : \obb(\C)$, $P(\modal{X})$
\item for each $X, Y : \obb(\C)$ with $P(Y)$, the function
  $\left({-} \circ \eta_X \right)  : \homm_{\C}(\modal{X},Y) \to \homm_{\C}(X, Y)$ is an equivalence.\footnote{When $\C \equiv \U$, a pre-reflective subcategory is known as a \emph{reflective subuniverse}~\cite[Definition 1.3]{RS}.}
We denote the inverse of this map by $\rec_{\modal}$, which enjoys the $\beta$-law $\beta_{\eta} : \prod_{X, Y : \obb(\C)}\prod_{\underline{\hspace{1.5mm}} : P(Y)}\prod_{f : \homm_{\C}(X, Y)}\rec_{\modal}(f) \circ \eta_X = f$.
\end{itemize}
\begin{notation}
We define $\C_P \coloneqq \sum_{X : \obb(\C)}P(X)$. 
\end{notation}
\end{definition}

Suppose that $\C$ is a wild category. Let $\left(P, \modal, \eta\right)$ be a pre-reflective subcategory of $\C$.

\begin{prop} \label{nateqmod}
For each $X : \obb(\C)$, the following square commutes in $\C$:
\[
\begin{tikzcd}[column sep = huge]
	X & Y \\
	{\modal{X}} & {\modal{Y}}
	\arrow["f", from=1-1, to=1-2]
	\arrow["{\eta_X}"', from=1-1, to=2-1]
	\arrow["{\eta_Y}", from=1-2, to=2-2]
	\arrow["{\rec_{\modal}(\eta_Y \circ f)}"', from=2-1, to=2-2]
\end{tikzcd}
\] where the bottom arrow is called the \emph{action of $\modal$ on $f$}.
\end{prop}

\noindent It's easy to check that this action on maps makes $\modal$ into a wild functor.

\begin{lemma} \label{modeq}
Suppose that $\C$ is univalent. For each $X : \obb(\C)$, $P(X) \to \isequiv(\eta_X)$.
\end{lemma}
\begin{proof}
Let $X : \obb(\C)$. The type $T_{P,X}$ of tuples
\begin{align*}
Y & \ : \ \obb(\C)
\\ q & \ : \ P(Y)
\\ f & \ : \ \homm_{\C}(X,Y)
\\ I & \ : \ \prod_{Z : \obb(\C)}P(Z) \to \isequiv(\lambda(g : \homm_{\C}(Y,Z)).g \circ f)
\end{align*}
is a proposition. Suppose that $P(X)$. We have elements
$\left(X, \ldots, \idd_X, \ldots\right)$ and $\left(\modal{X}, \ldots, \eta_X, \ldots\right)$
of $T_{P,X}$, which must be equal. Therefore, we have a commuting triangle
\[\begin{tikzcd}
	& X \\
	X && {\modal{X}}
	\arrow["\simeq"', dashed, from=2-1, to=2-3]
	\arrow["{\eta_X}", from=1-2, to=2-3]
	\arrow["\idd"', from=1-2, to=2-1]
\end{tikzcd}\] in $\C$. This implies that $\eta_X$ is an equivalence.
\end{proof} 

Combined with \cref{nateqmod}, \cref{modeq} implies that when $\C$ is univalent, $\eta$ restricted to $\C_P$ is a \emph{natural} isomorphism $\idd_{\C_P} \xrightarrow{\cong} \modal \circ \mathcal{I}$ of wild functors where $\mathcal{I}$ denotes the inclusion of the full subcategory $\C_P$ into $\C$. Let $\epsilon$ denote the counit of the adjunction $\modal \dashv \mathcal{I}$ induced by $\eta$. By the zigzag identity \eqref{zigzag}, each component of $\epsilon$ is equal to an equivalence. Hence $\epsilon$ is also a natural isomorphism. 

Classically, given an adjunction $L \dashv R$, the condition that the counit is an isomorphism is well known to be equivalent to the condition that $R$ is fully faithful. In the wild setting, we impose an extra condition on $L$ to get a well-behaved notion of \emph{reflective subcategory}. The condition is called \emph{2-coherence} and is defined at \cref{2coher}.

\begin{definition}\label{reflin}
Let $\C$ and $\D$ be wild categories with a $0$-functor $\mathcal{I} : \D \to \C$. We say that $\D$ is a \emph{reflective subcategory of $\C$} if we have a $2$-coherent left adjoint $L : \C \to \D$ whose counit is an isomorphism. 
\end{definition} 

\subsection{Orthogonal factorization systems}\label{OFS}

\begin{definition} \label{OFSdefn}
Let $\C$ be a wild category. An \emph{orthogonal factorization system (OFS)} on $\C$ consists of predicates $\L, \RI :   \prod_{ A, B : \obb(\C)}\homm_{\C}(A,B)  \to \mathsf{Prop}$ such that
\begin{enumerate}
\item both $\L$ and $\RI$ are closed under composition and have all identity morphisms
\item for every $h : \homm_{\C}(A,B)$, the following type is contractible:
\[
\fact_{\L, \RI}(h) \ \coloneqq \ \sum_{D : \obb(\C)}\sum_{f : \homm_{\C}(A,D)}\sum_{g : \homm_{\C}(D, B)} \L(f) \times \RI(g) \times g \circ f = h
\]
\end{enumerate}
\end{definition}

\noindent In \cref{OFSdefn}, when $\C$ is univalent, both $\L$ and $\RI$ have all equivalences in $\C$ by $\simeq_{\C}$-induction.

\begin{example}
Rijke et al.\ use a particular indexed recursive $1$-HIT to show that every family $\prod_{a :A}F(a) \to G(a)$ of functions induces an OFS on $\U$~\cite[Section 2.4]{RS}.
\end{example}

\begin{definition}
Let $\C$ be a wild category. Let $l : \homm_{\C}(A, B)$ and $\mathcal{H}$ be a property of morphisms in $\C$. We say that $l$ has the \textit{(unique) left lifting property against $\mathcal{H}$} if for every commuting square
\[\begin{tikzcd}
	A & C \\
	B & D
	\arrow[""{name=0, anchor=center, inner sep=0}, "f", from=1-1, to=1-2]
	\arrow["l"', from=1-1, to=2-1]
	\arrow["r", from=1-2, to=2-2]
	\arrow[""{name=1, anchor=center, inner sep=0}, "g"', from=2-1, to=2-2]
	\arrow["S"{description}, draw=none, from=0, to=1]
\end{tikzcd}\] with $r \in \mathcal{H}$, the type of \emph{diagonal fillers $\filll(S)$} of $S$ is contractible where
\[
\filll(S) \  \coloneqq \  \sum_{d : \homm_{\C}(B, C)}\sum_{H_f : f = d \circ l}\sum_{H_g : g = r \circ d} \ap_{{-} \circ l}(H_g) \cdot \assoc(r,d,l)  =  S \cdot \ap_{r \circ {-}}(H_f) 
\] In this case, we write ${^{\perp}\mathcal{H}}(l)$. 
The predicate \emph{(unique) right lifting property} is defined similarly.
\end{definition}

Let $\C$ be a univalent wild bicategory and let $\left(\L, \RI\right)$ be an OFS on $\C$.

\begin{lemma}\label{fillid}
Let $h : \homm_{\C}(A, B)$ and $\left(U, s_U, t_U, p_U\right), \left(V, s_V, t_V, p_V\right) : \fact_{\L, \RI}(h)$. We have that 
\[\begin{tikzcd}[row sep = small]
	{\left(U, s_U, t_U, p_U\right) = \left(V, s_V, t_V, p_V\right)} \\
	{ \sum_{e : U \simeq_{\C} V}
	\sum_{H_{\L} : s_V  = e \circ s_U}\sum_{H_{\RI} : t_U = t_V \circ e} \ap_{{-} \circ s_U}(H_{\RI})  \cdot p_U = \assoc(t_V, e, s_U) \cdot \ap_{t_V \circ {-}}(H_{\L}) \cdot  p_V}
	\arrow["\simeq"{marking, allow upside down}, draw=none, from=1-1, to=2-1]
\end{tikzcd}\]
\end{lemma}
\begin{proof}
By \cref{SIP}.
\end{proof}

\begin{lemma} \label{ULP}
For each commuting square in $\C$
\[\begin{tikzcd}
	A & X \\
	B & Y
	\arrow[""{name=0, anchor=center, inner sep=0}, "f", from=1-1, to=1-2]
	\arrow["l"', from=1-1, to=2-1]
	\arrow["r", from=1-2, to=2-2]
	\arrow[""{name=1, anchor=center, inner sep=0}, "g"', from=2-1, to=2-2]
	\arrow["S"{description}, draw=none, from=0, to=1]
\end{tikzcd}\]  with $l \in \L$ and $r \in \RI$, the type $\filll(S)$ of diagonal fillers of $S$ is contractible.
\end{lemma}
\begin{proof}
The argument is a wild-categorical extension of the proof of \cite[Lemma 1.44]{RS} (which just deals with a type universe). We have the commuting diagram
\[\begin{tikzcd}[row sep = large, column sep = large]
	A & {\im(f)} & X \\
	B & {\im(g)} & Y
	\arrow["{{{s_f}}}"{description}, from=1-1, to=1-2]
	\arrow[""{name=0, anchor=center, inner sep=0}, "f", curve={height=-30pt}, from=1-1, to=1-3]
	\arrow["l"', from=1-1, to=2-1]
	\arrow["{{{t_f}}}"{description}, from=1-2, to=1-3]
	\arrow["r", from=1-3, to=2-3]
	\arrow["{{{s_g}}}"{description}, from=2-1, to=2-2]
	\arrow[""{name=1, anchor=center, inner sep=0}, "g"', curve={height=30pt}, from=2-1, to=2-3]
	\arrow["{{{t_g}}}"{description}, from=2-2, to=2-3]
	\arrow["{{{p_f}}}"{description}, draw=none, from=0, to=1-2]
	\arrow["{{{p_g}}}"{description}, draw=none, from=1, to=2-2]
\end{tikzcd}\] Since $\fact_{\L, \RI}(r \circ f)$ is contractible, so is its identity type
\[
\left(\im(f), s_f , r \circ t_f,  \assoc(r, t_f, s_f)  \cdot \ap_{r \circ {-}}(p_f) \right)  \ = \ \left(\im(g), s_g \circ l, t_g, \assoc(t_g, s_g, l)^{-1} \cdot \ap_{{-} \circ l}(p_g) \cdot S  \right)
\] By \cref{fillid}, the following type is also contractible:
\[ 
\mathcal{T} \coloneqq \subalign{& \sum_{e : \im(f) \simeq_{\C} \im(g)}\sum_{H_{\L} : s_g \circ l  = e \circ s_f} \\ & \quad \sum_{H_{\RI} : r \circ t_f = t_g \circ e} \ap_{{-} \circ s_f}(H_{\RI})  \cdot \assoc(r, t_f, s_f)  \cdot \ap_{r \circ {-}}(p_f) = \assoc(t_g, e, s_f) \cdot \ap_{t_g \circ {-}}(H_{\L}) \cdot \assoc(t_g, s_g, l)^{-1} \cdot \ap_{{-} \circ l}(p_g) \cdot S}
\] By the univalence of $\C$ along with its bicategorical structure, we can simulate the proof of \cite[Lemma 1.44]{RS} to show that $\mathcal{T} \simeq \filll(S)$. Hence $\filll(S)$ is contractible.
\end{proof}

\begin{corollary}\label{LRcls}
Let $f : \homm_{\C}(A,B)$. We have that $\L(f) \leftrightarrow {^{\perp}\RI(f)}$ and $\L^{\perp}(f) \leftrightarrow \RI(f)$.
\end{corollary}
\begin{proof}
We just prove that $\L \leftrightarrow {^{\perp}\RI}$ as the other case is formally dual.\footnote{The opposite wild category of a univalent wild bicategory is itself a univalent bicategory.} By \cref{ULP}, we know that $\L(f) \to {^{\perp}\RI}(f)$. To prove the reverse implication, suppose that ${^{\perp}\RI}(f)$. Factor $f$ as $\left(\im(f), s_f, t_f, p_f\right)$ and consider the commuting square
\[\begin{tikzcd}[column sep =large]
	A & {\im(f)} \\
	B & B
	\arrow["{{s_f}}", from=1-1, to=1-2]
	\arrow[""{name=0, anchor=center, inner sep=0}, "f"', from=1-1, to=2-1]
	\arrow[""{name=1, anchor=center, inner sep=0}, "{{t_f}}", from=1-2, to=2-2]
	\arrow["\idd"', from=2-1, to=2-2]
	\arrow["{{\lid(f) \cdot p_f^{-1}}}"{description}, draw=none, from=0, to=1]
\end{tikzcd}\] Since ${^{\perp}\RI}(f)$, $\filll\mleft(\lid(f) \cdot p_f^{-1}\mright)$ is contractible, whose center we denote by $\left( d, H_{s_f}, H_{\idd}, K\right)$. The square
\[\begin{tikzcd}[column sep =large]
	A & {\im(f)} \\
	{\im(f)} & B
	\arrow["{{s_f}}", from=1-1, to=1-2]
	\arrow[""{name=0, anchor=center, inner sep=0}, "{{s_f}}"', from=1-1, to=2-1]
	\arrow[""{name=1, anchor=center, inner sep=0}, "{{t_f}}", from=1-2, to=2-2]
	\arrow["{{t_f}}"', from=2-1, to=2-2]
	\arrow["{{\refl_{t_f \circ s_f}}}"{description}, draw=none, from=0, to=1]
\end{tikzcd}\] has the following two diagonal fillers, which are equal because ${^{\perp}\RI}(s_f)$ by \cref{ULP}:
\begin{align*}
 & \left( \idd, \lid(s_f),  \rid(t_f), \textit{$\C$'s triangle identity}\right) 
 \\ & \left( d \circ t_f, \assoc(d, t_f, s_f) \cdot \ap_{d \circ {-}}(p_f) \cdot H_{s_f}, \assoc(t_f, d, t_f)^{-1} \cdot \ap_{{-} \circ t_f}(H_{\idd}) \cdot \lid(t_f), \upsilon \right)
\end{align*}
where $\upsilon$ is the path obtained via $K$, $p_f$, and standard bicategorical laws, including \cref{kellycoh}.
This implies that $t_f$ is an equivalence with inverse $d$, so that $t_f \in \L$. Thus, $f$ is the composite of two maps in $\L$.
\end{proof}

\begin{lemma} \label{leftstable}
Let $\left(\L, \RI\right)$ be an OFS on the wild category of types. Consider the pushout square
\[\begin{tikzcd}[row sep = large]
	C & B \\
	A & {A \sqcup_C B}
	\arrow["g", from=1-1, to=1-2]
	\arrow[""{name=0, anchor=center, inner sep=0}, "f"', from=1-1, to=2-1]
	\arrow[""{name=1, anchor=center, inner sep=0}, "\inr", from=1-2, to=2-2]
	\arrow["\inl"', from=2-1, to=2-2]
	\arrow["\lrcorner"{anchor=center, pos=0.125, rotate=180}, draw=none, from=2-2, to=1-1]
	\arrow["\glue"{description}, draw=none, from=0, to=1]
\end{tikzcd}\] defined as a HIT in the standard way (where $\glue : \inl \circ f \sim \inr \circ g$). If $f$ is in $\L$, then so is $\inr$.
\end{lemma}
\begin{proof} 
 We want to prove that for every commuting square
\[
\begin{tikzcd}
	B & E \\
	{A \sqcup_C B} & H
	\arrow["t", from=1-1, to=1-2]
	\arrow[""{name=0, anchor=center, inner sep=0}, "{{\inr}}"', from=1-1, to=2-1]
	\arrow[""{name=1, anchor=center, inner sep=0}, "v", from=1-2, to=2-2]
	\arrow["b"', from=2-1, to=2-2]
	\arrow["S"{description}, draw=none, from=0, to=1]
\end{tikzcd}
\] with $v \in \RI$, the type $\filll(S)$ is contractible. Consider the composite diagram 
\[\begin{tikzcd}
	C & B & E \\
	A & {A \sqcup_C B} & H
	\arrow["g", from=1-1, to=1-2]
	\arrow["f"', from=1-1, to=2-1]
	\arrow["t", from=1-2, to=1-3]
	\arrow["v", from=1-3, to=2-3]
	\arrow["{\inl}"', from=2-1, to=2-2]
	\arrow["b"', from=2-2, to=2-3]
\end{tikzcd}\] which commutes via the path $\gamma(x) \coloneqq \ap_b(\glue(x)) \cdot S(g(x))$ for all $x : C$. The type of diagonal fillers
$\filll\mleft(\gamma\mright)$ is contractible because $f \in \L$. By the universal property of pushouts, letting $\rec_{\sqcup}$ denote the cogap map, we find that $\filll(S)$ is equivalent to the tuples consisting of
$k   : A \to E$, $K_1 :  k \circ f \sim t \circ g$, $K_2  :  b \sim v \circ \rec_{\sqcup}(k,t,K_1)$, and a path $ K_2(\inr(x)) = S(x)$ for each $x : B$.
By the same universal property, we can turn the last two fields into pairs consisting of $T :    b \circ \inl \sim v \circ k $ and a path $T(f(x))  \cdot \ap_v(K_1(x)) = \ap_b(\glue(x)) \cdot S(g(x))$ for each $x : C$.
The resulting type of tuples is clearly equivalent to $\filll\mleft(\gamma\mright)$, and thus $\filll(S)$ is contractible.
\end{proof}

Let $L : \C \to \D$ and $R: \D \to \C$ be $0$-functors of wild categories and let $\left(\alpha, V_1, V_2\right) : L \dashv R$. 
Suppose that for all $f : \homm_{\C}(A,B)$, $g : \homm_{\D}(X,Y)$, and $d : \homm_{\D}(L{B}, X)$, the following hexagon commutes:
\[\begin{tikzcd}[column sep = 60]
	{\alpha(\mleft(g \circ d\mright) \circ L{f})} & {\alpha(g \circ \mleft(d \circ L{f}\mright))} \\
	{\alpha(g \circ d) \circ f} & {R{g} \circ \alpha(d \circ L{f})} \\
	{\mleft(R{g} \circ \alpha(d)\mright) \circ f} & {R{g} \circ \mleft(\alpha(d) \circ f\mright)}
	\arrow["{\ap_{\alpha}(\assoc(g, d, L{f}))}", equals, from=1-1, to=1-2]
	\arrow["{V_2(f,g \circ d)}", equals, from=2-1, to=1-1]
	\arrow["{V_1(g, d \circ L{f})}"', equals, from=2-2, to=1-2]
	\arrow["{\ap_{{-} \circ f}(V_1(g , d))}", equals, from=3-1, to=2-1]
	\arrow["{\assoc(R{g},\alpha(d),f)}"', equals, from=3-1, to=3-2]
	\arrow["{\ap_{R{g} \circ {-}}(V_2(f , d))}"', equals, from=3-2, to=2-2]
\end{tikzcd}\label{natscohadj} \tag{$\texttt{$V_1$-$V_2$-hex}$} \]
(Note that this hexagon is different from that of \cref{2coher}.) Intuitively, this coherence condition expresses that the two evident ways (mediated by associativity) of combining $V_1$ and $V_2$ for a proof that $\alpha$ is natural in both variables simultaneously are equal. 
As we'll see in \cref{cocomp}, our chief example of an adjunction satisfying \eqref{natscohadj} is the ordinary colimit left adjoint (valued in $\U$).

\begin{lemma}\label{adjsq}
Let $f : \homm_{\C}(A,B)$ and $g : \homm_{\D}(X,Y)$. Consider a commuting square in $\D$:
\[
\begin{tikzcd}
	{L{A}} & X \\
	{L{B}} & Y
	\arrow[""{name=0, anchor=center, inner sep=0}, "u", from=1-1, to=1-2]
	\arrow["{L{f}}"', from=1-1, to=2-1]
	\arrow["g", from=1-2, to=2-2]
	\arrow[""{name=1, anchor=center, inner sep=0}, "v"', from=2-1, to=2-2]
	\arrow["S"{description}, draw=none, from=0, to=1]
\end{tikzcd}
\] 
The type $\filll(S)$ is equivalent to the type of diagonal fillers of the square
\[\begin{tikzcd}[column sep = large]
	A &&& {R{X}} \\
	B &&& {R{Y}}
	\arrow["{{\alpha(u)}}", from=1-1, to=1-4]
	\arrow[""{name=0, anchor=center, inner sep=0}, "f"', from=1-1, to=2-1]
	\arrow[""{name=1, anchor=center, inner sep=0}, "{{R{g}}}", from=1-4, to=2-4]
	\arrow["{{\alpha(v)}}"', from=2-1, to=2-4]
	\arrow["{{V_2(f,v) \cdot \ap_{\alpha}(S) \cdot V_1(g,u)^{-1}}}"{description}, draw=none, from=1, to=0]
\end{tikzcd}
\]
\end{lemma}
\begin{proof}
Letting $\zeta \coloneqq V_2(f,v) \cdot \ap_{\alpha}(S) \cdot V_1(g,u)^{-1}$, we have the chain of equivalences
\begin{align*}
 & \ \sum_{d : \homm_{\D}(L{B}, X)}\sum_{H_u : u = d \circ L{f}}\sum_{H_v : v = g \circ d} \ap_{{-} \circ L{f}}(H_v) \cdot \assoc(g,d,L{f})  =  S \cdot \ap_{g \circ {-}}(H_u)
\\ \simeq & \
\sum_{d : \homm_{\D}(L{B}, X)}\sum_{H_u : u = d \circ L{f}}\sum_{H_v : v = g \circ d} \ap_{{-} \circ f}(\ap_{\alpha}(H_v) \cdot V_1(g , d)) \cdot \assoc(R{g},\alpha(d),f)  =  \zeta \cdot \ap_{R{g} \circ {-}}( \ap_{\alpha}(H_u)  \cdot V_2(f , d)) 
\\ \simeq & \
\sum_{d : \homm_{\D}(L{B}, X)}\sum_{H_u : \alpha(u) = \alpha(d \circ L{f})}\sum_{H_v : \alpha(v) = \alpha (g \circ d)} \ap_{{-} \circ f}(H_v \cdot V_1(g , d)) \cdot \assoc(R{g},\alpha(d),f)  =  \zeta \cdot \ap_{R{g} \circ {-}}(H_u \cdot V_2(f , d))
\\ \simeq & \
\sum_{d : \homm_{\D}(L{B}, X)}\sum_{H_u : \alpha(u) = \alpha(d) \circ f}\sum_{H_v : \alpha(v) = R{g} \circ \alpha(d)} \ap_{{-} \circ f}(H_v) \cdot \assoc(R{g},\alpha(d),f)  =  \zeta \cdot \ap_{R{g} \circ {-}}(H_u)  
\\ \simeq & \
\sum_{d : \homm_{\C}(B, R{X})}\sum_{H_u : \alpha(u) = d \circ f}\sum_{H_v : \alpha(v) = R{g} \circ d} \ap_{{-} \circ f}(H_v) \cdot \assoc(R{g},d,f)  =  \zeta \cdot \ap_{R{g} \circ {-}}(H_u) 
\end{align*}
The final three equivalences in this chain are induced by equivalences in the base types, while the first one comes from a fiberwise equivalence: Let $d : \homm_{\D}(L{B}, X)$, $H_u : u = d \circ L{f}$, and $H_v : v = g \circ d$.  Since $\alpha$ is an equivalence (hence embedding), the type  $\mathcal{T}  \coloneqq  \ap_{{-} \circ L{f}}(H_v) \cdot \assoc(g,d,L{f})  =  S \cdot \ap_{g \circ {-}}(H_u)$ is equivalent  to its image $\mathcal{T}_{\alpha}$ under $\alpha$. Further, we can recast the two endpoints of $\mathcal{S} \coloneqq \ap_{{-} \circ f}(\ap_{\alpha}(H_v) \cdot V_1(g , d)) \cdot \assoc(R{g},\alpha(d),f)  =  \zeta \cdot \ap_{R{g} \circ {-}}( \ap_{\alpha}(H_u)  \cdot V_2(f , d))$ as follows:
\[\begin{tikzcd}
	{V_2(f, v)  \cdot \ap_{\alpha}(\ap_{{-} \circ L{f}}(H_v)) \cdot V_2(f,g \circ d)^{-1} \cdot \ap_{{-} \circ f}(V_1(g , d)) \cdot \assoc(R{g},\alpha(d),f) } \\
	{\ap_{\alpha({-}) \circ f}(H_v) \cdot \ap_{{-} \circ f}(V_1(g , d)) \cdot \assoc(R{g},\alpha(d),f) } \\
	{\zeta \cdot \ap_{R{g} \circ \alpha({-})}(H_u)  \cdot \ap_{R{g} \circ {-}}(V_2(f , d)) } \\
	{V_2(f,v) \cdot \ap_{\alpha}(S) \cdot  \ap_{\alpha}(\ap_{g \circ {-}}(H_u)) \cdot V_1(g,d\circ L{f})^{-1} \cdot \ap_{R{g} \circ {-}}(V_2(f , d)) }
	\arrow["{\textit{via homotopy naturality of $V_2$ at $H_v$}}", equals, from=1-1, to=2-1]
	\arrow[equals, from=2-1, to=3-1]
	\arrow["{\textit{via homotopy naturality of $V_1$ at $H_u$}}", equals, from=3-1, to=4-1]
\end{tikzcd}\] Thus, after some rearranging, we see that $\mathcal{S}$ is equivalent to $\mathcal{T}_{\alpha}$ by the coherence \eqref{natscohadj}.
\end{proof}

\begin{corollary}[{\cite[\href{https://github.com/PHart3/colimits-agda/blob/v0.5.0/HoTT-Agda/core/lib/wild-cats/Adj-OFS-wc.agda}{Adj-OFS-wc}]{agda-colim-TR}}]\label{fspres}
Let $\C$ and $\D$ be univalent wild bicategories endowed with OFS's $\left(\L_1, \RI_1\right)$ and $\left(\L_2, \RI_2\right)$, respectively. Then $R$ preserves $\RI$ if and only if $L$ preserves $\L$. 
\end{corollary}
\begin{proof}
Suppose that $R$ preserves $\RI$. Let $f : \homm_{\C}(A, B)$ and $f \in \L_1$. Consider a commuting square
\[
\begin{tikzcd}
	{L{A}} & X \\
	{L{B}} & Y
	\arrow[""{name=0, anchor=center, inner sep=0}, "u", from=1-1, to=1-2]
	\arrow["{L{f}}"', from=1-1, to=2-1]
	\arrow["g", from=1-2, to=2-2]
	\arrow[""{name=1, anchor=center, inner sep=0}, "v"', from=2-1, to=2-2]
	\arrow["S"{description}, draw=none, from=0, to=1]
\end{tikzcd}
\] where $g \in \RI_2$. By \cref{LRcls}, if $\filll(S)$ is contractible, then $L{f} \in \L_2$. By \cref{adjsq}, $\filll(S)$ is equivalent to the type of diagonal fillers of a square from $f$ to $R{g}$. By \cref{LRcls} again, this is contractible because $R{g} \in \RI_1$.

The converse is formally dual.
\end{proof}

\subsection{Coslices of a universe} \label{coscat}

Let $\U$ be a universe and $A$ be a type. Let $X, Y : A/\U \coloneqq \sum_{X : \U}\left(A \to X\right)$. Consider the type
\[
X \to_A Y \ \coloneqq \ \sum_{h :  \pr_1(X) \to \pr_1(Y)}h \circ \pr_2(X) \sim \pr_2(Y) 
\] of maps from $X$ to $Y$. (We sometimes call such maps \emph{$A$-maps}.) For example, 
$X \to_{\1} Y$ is equivalent to the type of pointed maps from $X$ to $Y$, i.e., $\left(\pr_1(X), \pr_2(X)(\ast)\right) \to_{\ast} \left(\pr_1(Y), \pr_2(Y)(\ast)\right)$.
Now, for all $g: X \to_A Y$ and $h : Y \to_A Z$, define their composite
\[
h \circ g \coloneqq \left(\pr_1(h) \circ \pr_1(g), \lambda{a}.\ap_{\pr_1(h)}(\pr_2(g)(a)) \cdot \pr_2(h)(a) \right) \ :  \ X \to_A Z
\] We also have an evident identity $A$-map $X \to_A X$ for each $X : A/\U$. The foregoing data gives us a  univalent wild bicategory, which we call the \textit{coslice of $\U$ under $A$}, written abusively as $A/\U$. 

\begin{example}
We call $\1/\U$ the wild category of \emph{pointed types}, sometimes denoted by $\U^{\ast}$.
\end{example}

\begin{prop}\label{Aiso}
For all $X,Y : A/\U$, we have an equivalence
\[
\left(X = Y\right)  \  \simeq   \  \sum_{k : \pr_1(X) \xrightarrow{\simeq} \pr_1(Y)}k \circ \pr_2(X) \sim \pr_2(Y)
\]
\end{prop}

\begin{definition} 
Let $f,g : X \to_A Y$.  An  \textit{$A$-homotopy $f \sim_A g$} between $f$ and $g$ is a homotopy
$H : \pr_1(f) \sim \pr_1(g)$ along with, for all $a : A$, a commuting triangle
\[\begin{tikzcd}
	{\pr_1(f)(\pr_2(X)(a))} && {\pr_1(g)(\pr_2(X)(a))} \\
	& {\pr_2(Y)(a)}
	\arrow["{H(\pr_2(X)(a))}", equals, from=1-1, to=1-3]
	\arrow["{\pr_2(f)(a)}"', equals, from=1-1, to=2-2]
	\arrow["{\pr_2(g)(a)}", equals, from=1-3, to=2-2]
\end{tikzcd}\]
\end{definition}

\begin{lemma}\label{ptdext}
 The canonical function $
 \happly_A  :  \left(f = g\right) \to \left(f \sim_A g\right)$ is an equivalence for all $f,g$.
\end{lemma}
\begin{proof}
By \cref{SIP}.
\end{proof}

\begin{notation}
Define $\left\langle{H,p}\right\rangle \coloneqq \happly_A^{-1}(H,p)$.
\end{notation}

\subsubsection*{Inherited pre-reflective subcategories}

Let $\left(P, \modal, \eta\right)$ be a pre-reflective subcategory of $\U$ (\cref{reflsubdef}).

\begin{lemma}
The data
\begin{align*}
P_A(X) & \ \coloneqq \ P(\pr_1(X))
\\ \modal_A(X) & \ \coloneqq \ \left(\modal(\pr_1(X)), \eta_{\pr_1(X)} \circ \pr_2(X) \right)
\\ \eta_A(X) & \ \coloneqq \ \left(\eta_{\pr_1(X)}, \refl_{\eta_{\pr_1(X)}(\pr_2(X)({-})) } \right)
\end{align*}
forms a pre-reflective subcategory of $A/\U$, which we denote by $\left(A/\U\right)_P$.
\end{lemma}
\begin{proof}
For all $Y$ with $P(\pr_1(Y))$, the following two functions are mutually inverse:
\begin{align*}
& \alpha_{P,A} \ : \ \left(\modal_A{X} \to_A Y\right)  \to  \left(X \to_A Y\right)
\\ & \alpha_{P,A}(f, f_p) \ \coloneqq \ \left(f \circ \eta_{\pr_1(X)}, f_p \right)
\\ & \rec_{P,A} \ : \ \left(X \to_A Y\right)   \to  \left(\modal_A{X} \to_A Y\right)
\\ &  \rec_{P,A}(g, g_p)   \ \coloneqq \ \left(\rec_{\modal}(g), \lambda{a}.\beta_{\eta}(g,\pr_2(X)(a)) \cdot g_p(a) \right)
\end{align*}
(where $\beta_{\eta}$ has type $\prod_{X, Y : \U}\prod_{\underline{\hspace{1.5mm}} : P(Y)}\prod_{g : X \to Y}\rec_{\modal}(g) \circ \eta_{\pr_1(X)} \sim g$).
\end{proof}

Now consider the prime example of a modal operator: the $n$-truncation $\lN{-}\rN_n$ for each $n \geq {-2}$~\cite[Example 1.6]{RS}. The wild functor $\lN{-}\rN_n : A/\U \to \left(A/\U\right)_{\leq n}$ is left adjoint to the forgetful functor. By \cref{LAPC}, we find that it preserves colimits on coslices of $\U$---a fact we record here.

\begin{prop} \label{modpres}
The left adjoint $\lN{-}\rN_n : A/\U \to \left(A/\U\right)_{\leq n}$ is $2$-coherent (\cref{2coher}), hence preserves colimits.
\end{prop}

\subsection{Diagrams in coslices}

Let $\Gamma$ be a graph. An  \emph{$A$-diagram over $\Gamma$} is a family $F : \Gamma_0 \to A/\U$ of objects in $A/\U$ along with a map $F_{i,j,g} : F_i \to_A F_j$ for all $i,j : \Gamma_0$ and $g : \Gamma_1(i,j)$. 

\begin{example}
  For each $D: A/\U$, the \emph{constant diagram $\const_{\Gamma}(D)$ at $D$} is defined by $\left(\const_{\Gamma}(D)\right)_0(i) \coloneqq D$ and $\left(\const_{\Gamma}(D)\right)_1(i,j,g) \coloneqq \idd_D$. We often write just $D$ for $\const_{\Gamma}(D)$.
\end{example}

\noindent Let $F$ be an $A$-diagram over $\Gamma$ and $C : A/\U$. A \emph{cocone under $F$ on $C$ / with tip $C$} is a  family of maps $r : \prod_i F_i \to_A C$ together with
\begin{enumerate}[label=(\roman*)]
\item for all $i,j : \Gamma_0$ and $g : \Gamma_1(i,j)$, a homotopy $h_{i,j,g} : \pr_1(r_j) \circ \pr_1(F_{i,j,g}) \sim \pr_1(r_i)  $
\item for each $a : A$, a path $h_{i,j,g}(\pr_2(F_i)(a))^{-1} \cdot \ap_{\pr_1(r_j)}(\pr_2(F_{i,j,g})(a)) \cdot \pr_2(r_j)(a) = \ap_{\pr_1(r_i)}(a)$
\end{enumerate}
By \cref{ptdext}, the second collection of data here is equivalent to a path $r_j \circ F_{i,j,g} = r_i$. We call $C$ the \emph{tip} of the cocone, denoted by $\tip$. We denote the type of cocones under $F$ on $C$ by $\mathsf{Cocone}_F(C)$.

\begin{lemma}\label{nattreq}
For all $\left(\alpha, p\right), \left(\rho, q\right) : \mathsf{Cocone}_F(C)$, 
$\left(\alpha, p\right) = \left(\rho, q\right)$ is equivalent to the type of tuples
\begin{align*}
& W \ : \ \prod_{i : \Gamma_0} \pr_1(\alpha_i) \sim \pr_1(\rho_i)
\\ &  u \ : \ \prod_{i : \Gamma_0}\prod_{a : A} W_i(\pr_2(F_i)(a))^{-1} \cdot \pr_2(\alpha_i)(a) = \pr_2(\rho_i)(a) 
\\ &   \textit{for all $i,j : \Gamma_0$ and $g: \Gamma_1(i,j)$,} 
\\ & \ S_1 \ : \ \prod_{x : \pr_1(F_i)}W_j(\pr_1(F_{i,j,g})(x))^{-1} \cdot \pr_1(p_{i,j,g})(x) \cdot W_i(x) = \pr_1(q_{i,j,g})(x) 
\\ & \ S_2 \ : \   \prod_{a :A}  \Xi(W, u, p_{i,j,g}, a)  = \ap_{{-}^{-1} \cdot \ap_{\pr_1(\rho_j)}(\pr_2(F_{i,j,g})(a))  \cdot \pr_2(\rho_j)(a)}(S_1(\pr_2(F_i)(a))) \cdot \pr_2(q_{i,j,g})(a) 
\end{align*} Here, $\Xi(W, u, p_{i,j,g}, a)$ is the chain of paths
\[
\hspace*{-4mm}
\begin{tikzcd}
	{\left( W_j(\pr_1(F_{i,j,g})(\pr_2(F_i)(a)))^{-1} \cdot \pr_1(p_{i,j,g})(\pr_2(F_i)(a)) \cdot W_i(\pr_2(F_i)(a)) \right)^{-1} \cdot \ap_{\pr_1(\rho_j)}(\pr_2(F_{i,j,g})(a))  \cdot \pr_2(\rho_j)(a)} \\
	{ W_i(\pr_2(F_i)(a))^{-1} \cdot \pr_1(p_{i,j,g})(\pr_2(F_i)(a))^{-1} \cdot \ap_{\pr_1(\alpha_j)}(\pr_2(F_{i,j,g})(a)) \cdot \pr_2(\alpha_j)(a)} \\
	{ W_i(\pr_2(F_i)(a))^{-1} \cdot \pr_2(\alpha_i)(a)} \\
	{\pr_2(\rho_i)(a)}
	\arrow["{{\textit{via $u_j(a)$}}}", equals, from=1-1, to=2-1]
	\arrow["{\textit{via $\pr_2(p_{i,j,g})(a)$}}", equals, from=2-1, to=3-1]
	\arrow["{{u_i(a)}}", equals, from=3-1, to=4-1]
\end{tikzcd}
\]
\end{lemma}
\begin{proof}
By \cref{SIP}.
\end{proof}

\begin{note} \label{Interch} Let $\Gamma$ be a graph. Let $F$ be an $A$-diagram over $\Gamma$ and let $C : A/\U$. We have the following two equivalent descriptions of $\mathsf{Cocone}_F(C)$.
\begin{enumerate}[label=(\arabic*)]
\item For all $A$-diagrams $G$ and $H$ over $\Gamma$, the type of \textit{natural transformations} from $G$ to $H$ is
\[
G \Rightarrow_A H \  \coloneqq \ \sum_{\alpha : \prod_{i : \Gamma_0}\pr_1(G_i) \to_A \pr_1(H_i)}\prod_{i,j : \Gamma_0}\prod_{g : \Gamma_1(i,j)} H_{i,j,g} \circ \alpha_i \sim_A \alpha_j \circ G_{i,j,g}
\]
We have an evident equivalence $\mathsf{Cocone}_F(C) \simeq \left(F \Rightarrow_A \const_{\Gamma}(C) \right)$. 
\item For every $\U$-valued diagram $G$ over $\Gamma$, recall the \emph{standard limit} of $G$~\cite[Definition 4.2.7]{Avi}:
\[
\limm(G) \  \coloneqq \  \sum_{\alpha : \prod_{i : \Gamma_0}G_i}\prod_{i,j : \Gamma_0}\prod_{g : \Gamma_1(i,j)} G_{i,j,g}(\alpha_i) = \alpha_j  
\] We remark that this is functorial in $G$: the action on maps sends $\left(k, K\right) : G \Rightarrow H$ to the function $\limm(k ,K) : \limm(G) \to \limm(H)$ defined by $\left(\alpha, D\right) \mapsto \left(\lambda{i}.k_i(\alpha_i), \lambda{i}\lambda{j}\lambda{g}.D_{i,j,g}(\alpha_i) \cdot \ap_{k_j}(K_{i,j,g}) \right)$. We have an evident equivalence $\mathsf{Cocone}_F(C) \simeq \limm_{i : \Gamma^{\op}}(F_i \to_A C)$.
\end{enumerate}
\end{note}

\section{Graphs and trees}\label{trees}

Let $\U$ and $\V$ be universes. A \textit{(directed) graph (relative to $\U$ and $\V$)} is a pair $\Gamma$ consisting of a type $\Gamma_0 : \U$ and a type family $\Gamma_1 : \Gamma_0 \to \Gamma_0 \to \V$.

\begin{definition} 
Let  $\Gamma$  be a graph. 
 The \emph{graph quotient $\left\lvert{\Gamma}\right\rvert$ of $\Gamma$} is the HIT  generated by 
 $\left\lvert{-}\right\rvert  :  \Gamma_0 \to \left\lvert{\Gamma}\right\rvert$ and 
 $ \glue :  \prod_{x,y : \Gamma_0}\Gamma_1(x,y) \to \left\lvert{x}\right\rvert = \left\lvert{y}\right\rvert$. We say that $\Gamma$ is a \textit{tree} if $\left\lvert{\Gamma}\right\rvert$ is contractible.
\end{definition}

\noindent In \cref{Constr}, we will see that $A$-colimits interact nicely with trees.

\begin{example} $ $
\begin{enumerate}[label=(\arabic*)]
\item Both $\N$ and $\Z$ are trees when viewed as graphs.
\item The span $l \leftarrow m \rightarrow r$ is a tree where $l, m, r$ denote the elements of $\mathsf{Fin}(3)$.
\end{enumerate}
\end{example}

Let $\Gamma$ be a graph. For all $i, j : \Gamma_0$, define the type $\mathcal{W}_{\Gamma}(i,j)$ of \textit{walks from $i$ to $j$} as the indexed inductive type with constructors  $\mathsf{nil} : \prod_{i : \Gamma_0}\mathcal{W}_{\Gamma}(i, i)$ and
$\mathsf{cons}  :  \prod_{i,j, k : \Gamma_0} \Gamma_1(i, j) \to \mathcal{W}_{\Gamma}(j, k)  \to  \mathcal{W}_{\Gamma}(i, k)$.

\begin{definition}
Let $j_0 : \Gamma_0$. We say that $\Gamma$ is a \textit{combinatorial tree (at $j_0$)} if
\begin{itemize}
\item for every $i : \Gamma_0$, we have an element $\nu(i, j_0) : \mathcal{W}_{\Gamma}(i, j_0)$
\item for all $i, j : \Gamma_0$ and $g : \Gamma_1(i, j)$, we have an element
$\sigma_g  :   \nu(i, j_0) = \mathsf{cons}(g, \nu(j, j_0))$.
\end{itemize}
\end{definition}
 
\begin{lemma}\label{stree-eq}
For all $ i, j : \Gamma_0$ and $z : \mathcal{W}_{\Gamma}(i, j)$, we have an element
$\tau(z) :   \left\lvert{i}\right\rvert = \left\lvert{j}\right\rvert$.
\end{lemma}
\begin{proof}
We proceed by induction on $\mathcal{W}_{\Gamma}$. 
For every $i : \Gamma_0$, define $\tau(\mathsf{nil}_i) \coloneqq \refl_{\left\lvert{i}\right\rvert}$. 
Next, let $i, j, k : \Gamma_0$, $g : \Gamma_1(i,j)$, and $z : \mathcal{W}_{\Gamma}(j, k)$. Suppose we've defined $\tau(z)$. Define $\tau(\mathsf{cons}(g, z)) \coloneqq  \glue(g) \cdot \tau(z)$.
\end{proof}

\begin{lemma}
Every combinatorial tree is a tree.
\end{lemma}
\begin{proof}
Let $\Gamma$ be a combinatorial tree. It suffices to prove that for every $x : \left\lvert{\Gamma}\right\rvert$, $ x = \left\lvert{j_0}\right\rvert $. We proceed by induction on graph quotients. For each $i : \Gamma_0$, we have
$\tau(\nu(i, j_0))  :   \left\lvert{i}\right\rvert = \left\lvert{j_0}\right\rvert$  by \cref{stree-eq}. Since $\Gamma$ is a combinatorial tree, we also see that for all $i, j : \Gamma_0$ and $g : \Gamma_1(i,j)$, 
\begin{align*}
& \ \transport^{x \mapsto x = \left\lvert{j_0}\right\rvert}(\glue(g), \tau(\nu(i, j_0)))  &
\\ = & \ \glue(g)^{-1} \cdot \tau(\nu(i, j_0))
\\  = & \ \glue(g)^{-1} \cdot \tau(\mathsf{cons}(g, \nu(j, j_0)))
\\  \equiv & \ \glue(g)^{-1} \cdot \glue(g) \cdot  \tau(\nu(j, j_0))
\\  = & \ \tau(\nu(j, j_0)) \qedhere
\end{align*}
\end{proof}

\begin{corollary}
Every \emph{directed tree}---in the sense of Rijke~\cite[\href{https://unimath.github.io/agda-unimath/trees.directed-trees.html}{Directed trees}]{agda-unimath}---is a tree.
\end{corollary}
\begin{proof}
Just notice that every directed tree is a combinatorial tree.
\end{proof}

\begin{example}
Trees are abundant in HoTT. Indeed, consider a coalgebra for a polynomial endofunctor $\P_{A,B}$ for a signature $\left(A, B\right)$:
\[
\mathcal{X} \ \coloneqq \ \left(X, \ \alpha : X \to \sum_{a : A}\left(B(a) \to X\right) \right)
\]
All elements of $X$ can be made into a directed tree~\cite[\href{https://unimath.github.io/agda-unimath/trees.underlying-trees-elements-coalgebras-polynomial-endofunctors.html}{The underlying trees of elements of coalgebras of polynomial endofunctors}]{agda-unimath}. Hence every element of the $\mathsf{W}$-type for $\left(A, B\right)$ is a tree as $\mathsf{W}(A, B)$ has a canonical coalgebra structure~\cite[\href{https://unimath.github.io/agda-unimath/trees.w-types.html\#w-types-as-coalgebras-for-a-polynomial-endofunctor}{W-types as coalgebras for a polynomial endofunctor}]{agda-unimath}. Also, every element of the coinductive type $\mathsf{M}(A, B)$, the terminal coalgebra for $\P_{A,B}$, is a tree.
\end{example}

\section{Colimits}

\subsection{Ordinary colimits}\label{ordcol}

Let $\U$ be a universe. For each graph $\Gamma$, the \emph{(ordinary) colimit $\colimm(F)$} of a $\Gamma$-shaped diagram $F$ in $\U$ is the HIT generated by the following constructors (for which we write the $\Pi$-type in Agda notation):
\[
  \begin{tikzcd}[row sep = -5pt, ampersand replacement=\&]
    \& {F_i} \&\& {F_j} \\
    {\begin{aligned} \iota  & \  : \ \left( i : \Gamma_0\right) \to F_i \to \colimm(F) \\  \kappa & \ : \ \left(i,j : \Gamma_0\right)  \left(g : \Gamma_1(i,j)\right) \to \iota_j \circ F_{i,j,g} \sim \iota_i \end{aligned}} \\
    \&\& {\colimm(F)}
    \arrow[""{name=0, anchor=center, inner sep=0}, "{{F_{i,j,g}}}", from=1-2, to=1-4]
    \arrow["{{\iota_i}}"', from=1-2, to=3-3]
    \arrow["{{\iota_j}}", from=1-4, to=3-3]
    \arrow["{{\kappa_{i,j,g}}}"{description}, draw=none, from=0, to=3-3]
  \end{tikzcd}
\]
The induction principle for $\colimm(F)$ states that for every type family $E$ over $\colimm(F)$ together with data 
\begin{equation*}
\begin{aligned}[c]
& r  \ : \  \prod_{i : \Gamma_0}\prod_{x : F_i}E(\iota_i(x))
\end{aligned}
\qquad
\begin{aligned}[c]
 & K  \ : \ \prod_{i,j : \Gamma_0}\prod_{g : \Gamma_1(i,j)}\prod_{x : F_i}\transport^E(\kappa_{i,j,g}(x), r(j,F_{i,j,g}(x))) = r(i,x)
\end{aligned}
\end{equation*}
we have a function $\ind(E, r, K) : \prod_{z : \colimm(F)}E(z)$ that satisfies $\ind(E, r, K)(\iota_i(x)) \equiv r(i,x)$ and is equipped with a path 
\[
\beta_{\ind(E, r, K)}(i,j,g,x) \ : \ \apd_{\ind(E, r, K)}(\kappa_{i,j,g}(x)) = K(i,j,g,x)
\]
In the non-dependent case, we derive from $\ind$ a recursion principle $\rec_{\colimm}(E, r, K) : \colimm(F) \to E$, known as the \emph{cogap map} for the cocone $\left(r,K\right)$ under $F$.

\begin{example} \label{colimexmps} $ $
\begin{exmpenum}
\item \label{colimexmps:p1}
If $\Gamma_0 \equiv \N$ and $\Gamma_1(i,j) \equiv  \mleft(i+1 = j\mright) $, then we refer to $\Gamma$ as $\omega$ since it defines the first infinite ordinal. (We often abuse notation by referring to $\omega$ as just $\N$.)  
For every type family $F: \N \to \U$, we have an equivalence  
\begin{align*}
& \sigma \ : \ \left(\prod_{n,m : \N}\left(n+1 = m\right) \to F_n \to F_m\right) \xrightarrow{\simeq} \left( \prod_{n : \N}F_n \to F_{n+1}\right)
\\ & \sigma(F, n) \  \coloneqq \   F_{n, n+1}(\refl_{n+1})
\end{align*}
Further, if $F$ is a diagram over $\omega$, we have an equivalence between $\colimm(F)$ and the sequential colimit of $\sigma(F)$~\cite[Section 3]{SDR}. Indeed, we have an inverse of $\sigma$ by sending each $f : \prod_{n : \N}F_n \to F_{n+1}$ to
\[
\lambda{n}\lambda{m}\lambda{g}.\transport^{k \mapsto F_n \to F_k}(g, f_n)  \ : \ \prod_{n,m : \N}\left(n+1 = m\right) \to F_n \to F_m
\]
\item \label{colimexmps:p2}
Suppose $\Gamma_0 \equiv \left\{l, r, m\right\}$ (the three-element type) and $\Gamma_1(m,l)  \equiv \1$, $\Gamma_1(m,r)  \equiv \1$, and $\Gamma_1(i,j)  \equiv \0$ otherwise.
Then we have an equivalence $\colimm(F) \simeq F(l) \sqcup_{F(m)} F(r)$ defined by colimit recursion in the forward direction.
\item \label{colimexmps:p3}
If $\Gamma_0$ is a type and $\Gamma_1(i,j) \equiv \0$ for all $i,j : \Gamma_0$, then $\Gamma$ is called the \textit{discrete graph on $\Gamma_0$}. In this case,  $\colimm(F)$ is equivalent to the coproduct $\sum_{i : \Gamma_0}F_i$.
\end{exmpenum}
\end{example}

\begin{prop} \label{treecolim}
For every graph $\Gamma$, $\colimm{\1} \simeq \left\lvert{\Gamma}\right\rvert$.
\end{prop}

\begin{corollary} \label{constcol}
Let $\Gamma$ be a tree and $A$ be a type. The function $\left[\idd_A\right]_{i : \Gamma_0} : \colimm{A} \to A$ is an equivalence. 
\end{corollary}
\begin{proof} 
Suppose that $\Gamma$ is a tree. We have a commuting diagram
\[\begin{tikzcd}
	{\colimm{A}} & {\colimm(A \times \1)} & {A \times \colimm{\1}} & A
	\arrow["\simeq", from=1-1, to=1-2]
	\arrow["\simeq", from=1-2, to=1-3]
	\arrow["\simeq", from=1-3, to=1-4]
	\arrow["{\left[\idd_A\right]}"', curve={height=24pt}, from=1-1, to=1-4]
\end{tikzcd}\]
\end{proof}

\begin{lemma} \label{colimmapeq}
Let $\Gamma$ be a graph. Suppose that $F$ is a diagram over $\Gamma$. Let $Z$ be a type and $h_1, h_2 :  \colimm(F) \to Z$. If we have a homotopy $p_i(x) : h_1 \circ \iota_i \sim h_2 \circ \iota_i$ for all $i : \Gamma_0$  along with a commuting square
\[\begin{tikzcd}[row sep = large, column sep = large]
	{h_1(\iota_j(F_{i,j,g}(x)))} && {h_1(\iota_i(x))} \\
	{h_2(\iota_j(F_{i,j,g}(x)))} && {h_2(\iota_i(x))}
	\arrow["{\ap_{h_1}(\kappa_{i,j,g}(x))}", Rightarrow, no head, from=1-1, to=1-3]
	\arrow["{p_j(F_{i,j,g}(x))}"', Rightarrow, no head, from=1-1, to=2-1]
	\arrow["{p_i(x)}", Rightarrow, no head, from=1-3, to=2-3]
	\arrow["{\ap_{h_2}(\kappa_{i,j,g}(x))}"', Rightarrow, no head, from=2-1, to=2-3]
\end{tikzcd}\]
for all $i, j : \Gamma_0$, $g : \Gamma_1(i,j)$, and $x : F_i$, then $h_1 \sim h_2$.
\end{lemma}
\begin{proof}
By colimit induction.
\end{proof}

\begin{lemma}\label{funcpo}
Consider a pushout square
\[\begin{tikzcd}[row sep = large, column sep = large]
	C & B \\
	A & {A \sqcup_C B}
	\arrow[""{name=0, anchor=center, inner sep=0}, "f"', from=1-1, to=2-1]
	\arrow["g", from=1-1, to=1-2]
	\arrow[""{name=1, anchor=center, inner sep=0}, "\inr", from=1-2, to=2-2]
	\arrow["\inl"', from=2-1, to=2-2]
	\arrow["\lrcorner"{anchor=center, pos=0.125, rotate=180}, draw=none, from=2-2, to=1-1]
	\arrow["{\glue}"{description}, draw=none, from=0, to=1]
\end{tikzcd}\] Let $Z$ be a type and $h_1, h_2 : {A \sqcup_C B} \to Z$. If we have homotopies
$p_1 :  h_1 \circ \inl \sim h_2 \circ \inl$ and $p_2   :  h_1 \circ \inr \sim h_2 \circ \inr$ along with a commuting square
\[\begin{tikzcd}[row sep = large, column sep = large]
	{h_1(\inl(f(c)))} && {h_1(\inr(g(c)))} \\
	{h_2(\inl(f(c)))} && {h_2(\inr(g(c)))}
	\arrow["{p_2(g(c))}", Rightarrow, no head, from=1-3, to=2-3]
	\arrow["{p_1(f(c))}"', Rightarrow, no head, from=1-1, to=2-1]
	\arrow["{\ap_{h_1}(\glue(c))}", Rightarrow, no head, from=1-1, to=1-3]
	\arrow["{\ap_{h_2}(\glue(c))}"', Rightarrow, no head, from=2-1, to=2-3]
\end{tikzcd}\]
of paths in $Z$ for every $c :C$, then $h_1 \sim h_2$.
\end{lemma}
\begin{proof}
By pushout induction.
\end{proof}

Note that $\colimm$ is a wild functor from the category of diagrams over $\Gamma$ to $\U$. In particular, for each $\left(\alpha, p\right) : F \Rightarrow G$, the function $\colimm(\alpha, p) : \colimm(F) \to \colimm(G)$ is the canonical map induced by the following cocone under $F$:
\[
\begin{tikzcd}[ampersand replacement=\&]
	{F_i} \&\& {F_j} \\
	{G_i} \&\& {G_j} \\
	\& {\colimm(G)}
	\arrow["{\alpha_i}"', from=1-1, to=2-1]
	\arrow["{\alpha_j}", from=1-3, to=2-3]
	\arrow["{F_{i,j,g}}", from=1-1, to=1-3]
	\arrow["{\iota_i}"', from=2-1, to=3-2]
	\arrow["{\iota_j}", from=2-3, to=3-2]
	\arrow["{G_{i,j,g}}", dotted, from=2-1, to=2-3]
\end{tikzcd}  \tag{$\lambda{x}.\ap_{\iota_j}(p_{i,j,g}(x))^{-1} \cdot \kappa^G_{i,j,g}(\alpha_i(x))$}
\] 
Likewise, the pushout HIT is a wild functor on spans. For each map $\left(\psi, S\right)$ of spans
\[\begin{tikzcd}[column sep = large, row sep = 20]
	{A_1} & {C_1} & {B_1} \\
	{A_2} & {C_2} & {B_2}
	\arrow["{\psi_1}"', from=1-1, to=2-1]
	\arrow["{\psi_2}"{description}, from=1-2, to=2-2]
	\arrow["{\psi_3}", from=1-3, to=2-3]
	\arrow[""{name=0, anchor=center, inner sep=0}, "{g_1}", from=1-2, to=1-3]
	\arrow[""{name=1, anchor=center, inner sep=0}, "{g_2}"', from=2-2, to=2-3]
	\arrow[""{name=2, anchor=center, inner sep=0}, "{f_2}", from=2-2, to=2-1]
	\arrow[""{name=3, anchor=center, inner sep=0}, "{f_1}"', from=1-2, to=1-1]
	\arrow["{S_2}"{description}, draw=none, from=0, to=1]
	\arrow["{S_1}"{description}, draw=none, from=3, to=2]
\end{tikzcd}\] the function $\mathsf{po}(\psi, S) : A_1 \sqcup_{C_1} B_1 \to A_2 \sqcup_{C_2} B_2$ is the canonical map induced by
\[\begin{tikzcd}[row sep = small]
	{C_1} && {B_1} \\
	\\
	{A_1} && {A_2 \sqcup_{C_2}B_2}
	\arrow[""{name=0, anchor=center, inner sep=0}, "{f_1}"', from=1-1, to=3-1]
	\arrow["{g_1}", from=1-1, to=1-3]
	\arrow[""{name=1, anchor=center, inner sep=0}, "{\inr \circ \psi_3}", from=1-3, to=3-3]
	\arrow["{\inl \circ \psi_1}"', from=3-1, to=3-3]
\end{tikzcd}
\tag{$\lambda{x}.\ap_{\inl}(S_1(x))^{-1} \cdot \glue_2(\psi_2(x)) \cdot \ap_{\inr}(S_2(x))$}
\]

By \cref{colimmapeq}, for every map of spans $\delta : F \Rightarrow G$, the equivalence from \cref{colimexmps:p2} fits into a commuting square
\[
\label{pocolcsq}
\begin{tikzcd}
	{\colimm(F)} & {\colimm(G)} \\
	{F(l) \sqcup_{F(m)} F(r)} & {G(l) \sqcup_{G(m)} G(r)}
	\arrow["{{\colimm( \delta)}}", from=1-1, to=1-2]
	\arrow["{{\rec_{\colimm}}}"', from=1-1, to=2-1]
	\arrow["\simeq", from=1-1, to=2-1]
	\arrow["{{\rec_{\colimm}}}", from=1-2, to=2-2]
	\arrow["\simeq"', from=1-2, to=2-2]
	\arrow["{{\mathsf{po}(\delta)}}"', from=2-1, to=2-2]
\end{tikzcd}
\tag{$\texttt{po-colim}$}
\] 

\subsection{Coslice colimits}

Let $A$ be a type and $\Gamma$ be a graph. Let $F$ be an $A$-diagram over $\Gamma$. A cocone $\mathcal{A} \coloneqq \left(C, r, K\right)$ under $F$ is \emph{colimiting} if the following function is an equivalence for every $T : A/\U$:
\begin{align*}
& \postcomp_{\A}(T) \ : \ \left(C \to_A T\right)  \to \mathsf{Cocone}_F(T)
\\ & \postcomp_{\A}(T, \left(f, f_p\right))  \ \coloneqq \ \left(c_1 , c_2\right)
\\ & \quad c_1(i) \ \coloneqq \ \left(f \circ \pr_1(r_i), \lambda{a}.\ap_f(\pr_2(r_i)(a)) \cdot f_p(a)\right)  
\\ & \quad c_2(i,j,g)  \ \coloneqq \ \left(\lambda{x}.\ap_f(\pr_1(K_{i,j,g})(x)), \lambda{a}.\Theta_{\pr_1(K_{i,j,g})}(f^{\ast},a) \cdot \ap_{\ap_f({-}) \cdot f_p(a)}(\pr_2(K_{i,j,g})(a))  \right)
\end{align*} where $\Theta_{\pr_1(K_{i,j,g})}(f^{\ast},a) $ is the evident path of type
\[\begin{tikzcd}
	{\ap_f(\pr_1(K_{i,j,g})(\pr_2(F_i)(a)))^{-1} \cdot \ap_{f \circ \pr_1(r_j)}(\pr_2(F_{i,j,g})(a)) \cdot \ap_f(\pr_2(r_j)(a)) \cdot f_p(a)} \\
	{\ap_f(\pr_1(K_{i,j,g})(\pr_2(F_i)(a))^{-1} \cdot \ap_{\pr_1(r_j)}(\pr_2(F_{i,j,g})(a)) \cdot \pr_2(r_j)(a)) \cdot f_p(a)}
	\arrow[Rightarrow, no head, from=1-1, to=2-1]
\end{tikzcd}\]

For each $m : C \to_A T$, the first component of $\postcomp_{\A}(T)(m)$ is exactly $m$ composed with $r_i$. We also can put its second component into a polished form. Indeed, it equals the right whiskering of $K_{i,j,g}$ by $m$ adjusted by associativity of $A$-maps. This description matches the definition of a colimiting cocone in an arbitrary wild category (\cref{colimwildc}).

If $\A$ is colimiting, then for each $\B : \mathsf{Cocone}_F(T)$, the function $\postcomp_{\A}^{-1}(\B) : C \to_A T$ is called the \emph{cogap map of $\B$}.

\begin{note}[Forgetful functor]
 The (wild) forgetful functor $\F : A/\U \to \U$ induces a functor  $\F : \mathsf{Diag}(\Gamma, A/\U) \to \mathsf{Diag}(\Gamma, \U)$ from the wild category of diagrams in $A/\U$ to that of diagrams in $\U$. It also induces a functor $\F : \mathsf{Cocone}(F) \to \mathsf{Cocone}(\pr_1 \circ F)$ between categories of cocones for each diagram $F : \Gamma \to A/\U$. In this case, $\F$ maps a cocone $\left(C, r, K\right)$ under $F$ to the following cocone under $\F(F)$:
\[\begin{tikzcd}[column sep = large, row sep = small]
	{\pr_1(F_i)} && {\pr_1(F_j)} \\
	\\
	& {\pr_1(C)}
	\arrow[""{name=0, anchor=center, inner sep=0}, "{\pr_1(F_{i,j,g})}", from=1-1, to=1-3]
	\arrow["{\pr_1(r_i)}"', from=1-1, to=3-2]
	\arrow["{\pr_1(r_j)}", from=1-3, to=3-2]
	\arrow["{\pr_1(K_{i,j,g})}"{description}, draw=none, from=0, to=3-2]
\end{tikzcd}\]
\end{note}

Let $\A$ and $\B$ be cocones under $F$. A \emph{morphism $\A \to \B$} is a map $\varphi :  \tip(A) \to_A \tip(B)$ with a path $\postcomp_{\A}(h) = \B$. The cogap map of a cocone under $F$ has an evident cocone morphism structure.

\begin{definition}
A morphism $\varphi : \A \to \B$ of cocones under $F$ is an \emph{isomorphism} if the map of types $\pr_1(\varphi) : \pr_1(\tip(\A)) \to \pr_1(\tip(B))$ is an equivalence.
\end{definition}

\noindent The following proposition can be proved purely in the language of wild bicategories.

\begin{prop} \label{coluniq}
Let $\A$ and $\B$ be cocones under $F$.
\begin{propenum}
\item  If they are both colimiting, then there is a unique cocone morphism $\A \to \B$, and this is an isomorphism. \label{coluniq:p1}
\item  If we have a cocone isomorphism between them, then one is colimiting if and only if the other is colimiting. \label{coluniq:p2}
\end{propenum}
\end{prop}

\subsubsection*{Intuition for colimit in $A/\U$}

For all $i, j : \Gamma_0$ and $g : \Gamma_1(i,j)$, the commuting triangle of a cocone
\[\begin{tikzcd}
	{F_i} && {F_j} \\
	& C
	\arrow["{F_{i,j,g}}", from=1-1, to=1-3]
	\arrow["{r_i}"', from=1-1, to=2-2]
	\arrow["{r_j}", from=1-3, to=2-2]
\end{tikzcd}\] under $F$ is equivalent to a homotopy $\eta_{i,j,g} : \pr_1(r_j) \circ \pr_1(F_{i,j,g}) \sim \pr_1(r_i)$
equipped with a commuting square
\[
\begin{tikzcd}[column sep = 35]
	{\pr_1(r_j)(\pr_1(F_{i,j,g})(\pr_2(F_i)(a)))} && {\pr_1(r_i)(\pr_2(F_i)(a))} \\
	{\pr_1(r_j)(\pr_2(F_j)(a))} && {\pr_2(Y)(a)}
	\arrow["{\eta_{i,j,g}(\pr_2(F_i)(a))}", Rightarrow, no head, from=1-1, to=1-3]
	\arrow["{\ap_{\pr_1(r_j)}(\pr_2(F_{i,j,g})(a))}"', Rightarrow, no head, from=1-1, to=2-1]
	\arrow["{\pr_2(r_i)(a)}", Rightarrow, no head, from=1-3, to=2-3]
	\arrow["{\pr_2(r_j)(a) }"', Rightarrow, no head, from=2-1, to=2-3]
\end{tikzcd}
\label{etaeq} \tag{$\texttt{$2$-c}$}  \] of paths for each $a :A$.
It is this family of $2$-cells which distinguishes the colimit of $F$, in $A/\U$, from $\colimm(\F(F))$. The $2$-cells affect $\colimm(\F(F))$ by collapsing its nontrivial loops formed by paths of the form $\eta(\pr_2(F_i)(a))$. We call such loops \emph{distinguished loops} in $\colimm(\F(F))$. For example, if $i \equiv j$ and $F_{i,j,g} \equiv \idd_{F_i}$, then \eqref{etaeq} is equivalent to $ \eta(\pr_2(F_i)(a)) = \refl_{\pr_1(r_i)(\pr_2(F_i)(a))}$. In this case, it fills the loop $\eta(\pr_2(F_i)(a))$.

\subsection{Coslice coproducts} \label{coprodcos}

Coslice coproducts admit a special construction: as wedge sums. Let $A$ be a type. Let $\Delta$ be a graph and $G$ be an $A$-diagram over $\Delta$.
If $\Delta$ is discrete, then  the pushout
  \[\begin{tikzcd}
      {\Delta_0 \times A} & {\sum_{i : \Delta_0}\pr_1(G_i)} \\
      A & D
      \arrow["{t}", from=1-1, to=1-2]
      \arrow["{\pr_2}"', from=1-1, to=2-1]
      \arrow["\inr", from=1-2, to=2-2]
      \arrow["\inl"', from=2-1, to=2-2]
      \arrow["\lrcorner"{anchor=center, pos=0.125, rotate=180}, draw=none, from=2-2, to=1-1]
    \end{tikzcd} \tag{$t(i,a) \coloneqq  \left(i, \pr_2(G_i)(a)\right)$}\] together with $\inl$
is the coproduct of the $G_i$ in $A/\U$, whose cocone structure is the family of $A$-maps
\[\begin{tikzcd}[row sep = 30, column sep = 30]
	& A \\
	{\pr_1(G_i)} & {\sum_{i : \Delta_0}\pr_1(G_i)} & D
	\arrow["{{\pr_2(G_i)}}"', from=1-2, to=2-1]
	\arrow["{{\glue_D(i,a)^{-1}}}"{description, pos=0.6}, draw=none, from=1-2, to=2-2]
	\arrow["\inl", from=1-2, to=2-3]
	\arrow["{{\left(i, {-}\right)}}"', from=2-1, to=2-2]
	\arrow["\inr"', from=2-2, to=2-3]
\end{tikzcd}
 \] 
\begin{notation}
We write $\bigvee_{\Delta}G$ for $D$.
\end{notation} 
\noindent Indeed, for each $X: A/\U$, the canonical map
\begin{align*}
& \postcomp \ : \ \left( \left(D , \inl\right) \to_A X\right) \to \left(\prod_{i : \Delta_0}G_i \to_A X\right)
\\ & \postcomp(f, i) \ \coloneqq \  f \circ \left( \inr(i, {-}) , \glue_D(i,{-})^{-1}\right) 
\end{align*} has inverse taking $m : \prod_{i : \Delta_0}G_i \to_A X$ to the  $A$-map $\left(h_m , \refl_{\pr_2(X)({-})}\right)$ where $h_m : \pr_1(D) \to \pr_1(X)$ is defined by pushout recursion via the cocone
\[\begin{tikzcd}
	{\Delta_0 \times A} & {\sum_{i : \Delta_0}\pr_1(G_i)} \\
	A & X
	\arrow[from=1-1, to=1-2]
	\arrow[""{name=0, anchor=center, inner sep=0}, from=1-1, to=2-1]
	\arrow[""{name=1, anchor=center, inner sep=0}, "{\left(i,x\right) \mapsto \pr_1(m_i)(x)}", from=1-2, to=2-2]
	\arrow["{\pr_2(X)}"', from=2-1, to=2-2]
	\arrow["{\pr_2(m_i)(a)}"{description}, draw=none, from=0, to=1]
\end{tikzcd}\] A direct proof of this equivalence is not hard, but we omit it as the equivalence will follow from our general construction of coslice colimits: \cref{mainequiv} (below).

We also can describe the coproduct in $A/\U$ as the colimit of an augmented diagram in $\U$: If $\Delta$ is any graph, we define a graph $I(\Delta)$ along with a diagram $\zeta(\Delta, G)$ over it: 
\begin{align*}
\begin{split}
I(\Delta)_0 & \ \coloneqq \ \Delta_0 + \1
\\ I(\Delta)_1(\inl(i), \inl(j)) &  \ \coloneqq \ \Delta_1(i,j)
\\ I(\Delta)_1(\inl(i), \inr(\ast)) & \ \coloneqq \ \0
 \\ I(\Delta)_1(\inr(\ast), \inl(i)) & \ \coloneqq \ \1
  \\ I(\Delta)_1(\inr(\ast), \inr(\ast)) & \ \coloneqq \ \0
\end{split}
\begin{split}
 \\ \zeta(\Delta, G)_{\inl(i)} & \ \coloneqq \ \pr_1(G_i)
\\ \zeta(\Delta, G)_{\inr(\ast)} & \ \coloneqq \ A
 \\ \zeta(\Delta, G)_{\inl(i), \inl(j), g} & \ \coloneqq \ \pr_1(G_{i,j,g})
  \\ \zeta(\Delta, G)_{\inr(\ast), \inl(i), \ast} & \ \coloneqq \ \pr_2(G_i)
\end{split}
\end{align*}

\begin{lemma}[{\cite[\href{https://github.com/PHart3/colimits-agda/blob/v0.5.0/Colimit-coslice/Coproduct/CosWedge.agda\#L101}{cos-wedge-colim-iso}]{agda-colim-TR}}] \label{coscoprodeqv}
Suppose that $\Delta$ is discrete. The coproduct $\bigvee_{\Delta}{G}$ fits into a cocone isomorphism in $A/\U$:
\[ 
\begin{tikzcd}
	& {G_i} \\
	{\mleft(\bigvee_{\Delta}{G}, \inl\mright)} && {\mleft(\colimm(\zeta(\Delta, G)), \iota_{\inr(\ast)}\mright)}
	\arrow["{\simeq}"', from=2-1, to=2-3]
	\arrow["{\left(\inr(i, {-}), \lambda{a}.\glue(i,a)^{-1}\right)}"', from=1-2, to=2-1]
	\arrow["{\left(\iota_{\inl(i)}, \kappa_{\zeta(\Delta,G)}(\inr(\ast), \inl(i), \ast )\right)}", from=1-2, to=2-3]
\end{tikzcd} \label{eq:trico} \tag{$\texttt{tri-$\bigvee$}$}\]
\end{lemma}
\begin{proof}
Define
\begin{align*}
& \left(\varphi, \alpha\right)  \ : \ \left(\bigvee_{\Delta}{G}\right) \to_A \colimm(\zeta(\Delta, G)) 
\\ & \varphi(\inl(a))  \ \coloneqq \ \iota_{\inr(\ast)}(a)
\\ & \varphi(\inr(i, x))  \ \coloneqq \ \iota_{\inl(i)}(x)
\\ & \ap_{\varphi}(\glue(i,a)) \ = \  \kappa_{\inr(\ast), \inl(i), \ast}(a)^{-1} 
\\ & \alpha  \ \coloneqq \ \refl_{\iota_{\inr(\ast)}({-})}
\end{align*}
Conversely, define
\begin{align*}
& \nu  \ : \ \colimm(\zeta(\Delta, G)) \to \bigvee_{\Delta}{G}
\\ & \nu(\iota_{\inr(\ast)}(a))  \ \coloneqq \ \inl(a)
\\ & \nu(\iota_{\inl(i)}(x))  \ \coloneqq \ \inr(i,x)
\\ & \ap_{\nu}(\kappa_{\inr(\ast), \inl(i), \ast}(a))  \ = \  \glue(i, a)^{-1}
\end{align*}
It is easy to prove that $\varphi$ and $\nu$ are mutual inverses as ordinary functions, as is checking that the triangle \eqref{eq:trico} commutes in $A/\U$.
\end{proof}

\begin{remark}
It is \emph{not} the case that $\colimm^A(G)$ and $\colimm(\zeta(\Delta, G))$ are equivalent for general graphs $\Delta$. For example, the pointed colimit (i.e., colimit in the wild category of pointed types) of $\1 \xrightarrow{\idd} \1$ is contractible, but the colimit of the augmented diagram
\[\begin{tikzcd}
	& \1 \\
	\1 && \1
	\arrow["\idd"', from=1-2, to=2-1]
	\arrow["\idd"', from=2-1, to=2-3]
	\arrow["\idd", from=1-2, to=2-3]
\end{tikzcd}\]
equals $S^1$. This situation may seem different from classical category theory, wherein colimits in coslice categories can be computed as colimits of augmented diagrams in the underlying category. Note, however, that the internal augmented diagram may add ``composites'' that are \emph{not} treated as such by the free category generated by the graph. Rather, it treats them as unrelated arrows.

\end{remark}

\subsection{Quotient construction of coslice colimits}\label{Constr}

This section describes the main connection between ordinary colimits and coslice colimits. We build the colimit of $F$ in a way that never produces an augmented diagram. We start with the ordinary colimit $\colimm(\F(F))$ which ignores the coslice structure of $F$. Then, we glue onto this colimit the $2$-cells required by the coslice colimit. We do this via a quotient of $\colimm(\F(F))$ that fills its distinguished loops. For convenience, we recall the following two standard results of HoTT.

\begin{lemma} \label{transfunc}
Let $X$ be a type and $P : X \to \U$. Let $f, g : \prod_{x : X}P(x)$. For all $x, y : X$, $p : x = y$, and $H : f \sim g$, we have the commuting square
\[\begin{tikzcd}
	{\transport^P(p, f(x))} && {\transport^P(p, g(x))} \\
	{f(y)} && {g(y)}
	\arrow["{{\ap_{p_{\ast}}(H(x))}}", Rightarrow, no head, from=1-1, to=1-3]
	\arrow["{{\apd_f(p)}}"', Rightarrow, no head, from=1-1, to=2-1]
	\arrow["{{\apd_g(p)}}", Rightarrow, no head, from=1-3, to=2-3]
	\arrow["{{\transport^{z \mapsto f(z) = g(z)}(p, H(y))}}"', Rightarrow, no head, from=2-1, to=2-3]
\end{tikzcd}
\] If $P$ is constant, then this becomes the commuting square
\[\begin{tikzcd}[column sep = 50]
	{f(x)} && {g(x)} \\
	{f(y)} && {g(y)}
	\arrow["{H(x)}", Rightarrow, no head, from=1-1, to=1-3]
	\arrow["{\ap_f(p)}"', Rightarrow, no head, from=1-1, to=2-1]
	\arrow["{\ap_g(p)}", Rightarrow, no head, from=1-3, to=2-3]
	\arrow["{\transport^{z \mapsto f(z) = g(z)}(p, H(y))}"', Rightarrow, no head, from=2-1, to=2-3]
\end{tikzcd}
\] 
\end{lemma}

\begin{corollary}\label{apdnat}
Let $X$ be a type and $P : X \to \U$. Let $f, g : \prod_{x : X}P(x)$. For all $x,y : X$, $p : x = y$, and $H : f \sim g$, we have the commuting square
\[ \begin{tikzcd}[column sep = large]
	{\transport^P(p, f(x))} & {\transport^P(p, g(x))} \\
	{f(y)} & {g(y)}
	\arrow["{{{\ap_{p_{\ast}}(H(x))}}}", equals, from=1-1, to=1-2]
	\arrow["{{{\apd_f(p)}}}"', equals, from=1-1, to=2-1]
	\arrow["{{{\apd_g(p)}}}", equals, from=1-2, to=2-2]
	\arrow["{{{H(y)}}}"', equals, from=2-1, to=2-2]
\end{tikzcd} \]
\end{corollary}

Let $A$ be a type. Consider a graph $\Gamma$ and a diagram $F : \Gamma \to A/\U$ over $\Gamma$.
Define $\psi : \colimm{A} \to \colimm(\F(F))$ as the function induced by the cocone
\[\begin{tikzcd}
	A && A \\
	& {\colimm(\F(F))}
	\arrow["{{{{\idd_A}}}}", from=1-1, to=1-3]
	\arrow["{{{{\iota_i \circ \pr_2(F_i) }}}}"', from=1-1, to=2-2]
	\arrow["{{{{\iota_j \circ \pr_2(F_j)}}}}", from=1-3, to=2-2]
\end{tikzcd}
\tag{$\lambda{a}.\ap_{\iota_j}(\pr_2(F_{i,j,g})(a))^{-1} \cdot \kappa_{i,j,g}(\pr_2(F_i)(a))$} \] 
under the constant diagram at $A$. Then form the following pushout square (which we think of as a quotient of $\colimm(\F(F))$):
\[
\begin{tikzcd}[column sep =large]
	{\colimm{A}} & {\colimm(\F(F))} \\
	A & \P_A(F)
	\arrow["{\left[\idd_A\right]_{i : \Gamma_0}}"', from=1-1, to=2-1]
	\arrow["\psi", from=1-1, to=1-2]
	\arrow["\inr", from=1-2, to=2-2]
	\arrow["\inl"', from=2-1, to=2-2]
	\arrow["\lrcorner"{anchor=center, pos=0.125, rotate=180}, draw=none, from=2-2, to=1-1]
\end{tikzcd}
\]
We build an cocone, which we call $\K(\P_A(F))$, on $\left(\P_A(F), \inl\right)$ under $F$ as follows:
\[
\begin{tikzcd}[column sep = large]
	{F_i} && {F_j} \\
	& {\mleft(\P_A(F), \inl\mright)}
	\arrow["{\left(\inr \circ \iota_i, \tau_i\right) }"', from=1-1, to=2-2]
	\arrow["{\left(\inr \circ \iota_j, \tau_j\right)}", from=1-3, to=2-2]
	\arrow[""{name=0, anchor=center, inner sep=0}, "{F_{i,j,g}}", from=1-1, to=1-3]
	\arrow["{\left\langle{ \delta_{i,j,g}  , \epsilon_{i,j,g}}\right\rangle}"{description}, draw=none, from=0, to=2-2]
\end{tikzcd} \tag{$ \tau_i(a) \coloneqq \glue_{\P_A(F)}(\iota_i(a))^{-1}$}
\] Here, we define $
\delta_{i,j,g} \coloneqq \lambda{x}.\ap_{\inr}(\kappa_{i,j,g}(x))  :  \inr \circ \iota_j  \circ F_{i,j,g} \sim \inr \circ \iota_i$, and we define, for each $a : A$, $\epsilon_{i,j,g}(a)$ as the chain 
\[  \begin{tikzcd}
	{\ap_{\inr}(\kappa_{i,j,g}(\pr_2(F_i)(a)))^{-1} \cdot \ap_{\inr \circ \iota_j}(\pr_2(F_{i,j,g})(a))   \cdot \tau_j(a)} \\
	{\ap_{\inr}(\ap_{\iota_j}(\pr_2(F_{i,j,g})(a))^{-1} \cdot \kappa_{i,j,g}(\pr_2(F_i)(a)))^{-1} \cdot \tau_j(a) \cdot \refl_{\inl(a)}} \\
	{\ap_{\inr}(\ap_{\iota_j}(\pr_2(F_{i,j,g})(a))^{-1} \cdot \kappa_{i,j,g}(\pr_2(F_i)(a)))^{-1} \cdot \tau_j(a) \cdot \ap_{\inl}(\ap_{\left[\idd_A\right]}(\kappa_{i,j,g}(a)))} \\
	{\ap_{\inr}(\ap_{\psi}(\kappa_{i,j,g}(a)))^{-1} \cdot \tau_j(a) \cdot \ap_{\inl}(\ap_{\left[\idd_A\right]}(\kappa_{i,j,g}(a)))} \\
	{\left(\kappa_{i,j,g}(a)\right)_{\ast}(\tau_j(a))} \\
	{ \tau_i(a)}
	\arrow[Rightarrow, no head, from=1-1, to=2-1]
	\arrow["{\textit{via $\beta_{\left[\idd_A \right]}(i,j,g,a)$}}", Rightarrow, no head, from=2-1, to=3-1]
	\arrow["{\textit{via $\beta_{\psi}(i,j,g,a)$}}", Rightarrow, no head, from=3-1, to=4-1]
	\arrow["{{{{{\textit{via \cref{transfunc}}}}}}}", Rightarrow, no head, from=4-1, to=5-1]
	\arrow["{{{{{\apd_{\glue({-})^{-1}}(\kappa_{i,j,g}(a))}}}}}", Rightarrow, no head, from=5-1, to=6-1]
\end{tikzcd}   \] For \cref{mainequiv}, it will be convenient to decompose $\epsilon_{i,j,g}(a)$ into the following chains of paths:
\begin{enumerate}
\item $E_1(i,j,g,a)$, the first path of $\epsilon_{i,j,g}(a)$
\item $E_2(i,j,g,a)$, the second path of $\epsilon_{i,j,g}(a)$
\item $E_3(i,j,g,a)$, the final three paths of $\epsilon_{i,j,g}(a)$.
\end{enumerate} 

\begin{theorem}[{\cite[\href{https://github.com/PHart3/colimits-agda/blob/v0.5.0/Colimit-coslice/Main-Theorem/CosColim-Iso.agda}{CosColim-Iso}]{agda-colim-TR}}] \label{mainequiv} 
Let $\left(T, f_T\right) : A/\U$. The $\postcomp$ function
\begin{align*}
& \pstc_{F,T} \ : \ \left(\left(\P_A(F), \inl\right) \to_A \left(T, f_T\right)\right) \to \mathsf{Cocone}_F(T, f_T)
\\ & \pstc_{F,T}(f, f_p) \ \coloneqq \  
\\ & \ \left(\lambda{i}. \left(f \circ \inr \circ \iota_i, \lambda{a}.\ap_f(\tau_i(a)) \cdot f_p(a) \right), \lambda{i}\lambda{j}\lambda{g}.\left(\lambda{x}.\ap_f(\delta_{i,j,g}(x)), \lambda{a}.\Theta_{\delta_{i,j,g}}(f^{\ast},a) \cdot \ap_{\ap_f({-}) \cdot f_p(a)}(\epsilon_{i,j,g}(a))  \right) \right)
 \end{align*}
is an equivalence, i.e., the cocone $\K(\P_A(F))$ is colimiting in $A/\U$.
\end{theorem}
\begin{proof}
We define an inverse of $\pstc_{F,T}$ as follows. Consider a cocone $\left(r, K\right) : \mathsf{Cocone}_F(T,f_T)$ under $F$ on $\left(T, f_T\right)$. For all $i : \Gamma_0$ and $a : A$, $\pr_2(r_i)(a)^{-1}$ witnesses that $f_T(a) = \rec_{\colimm}(\F(r, K))(\pr_2(F_i(a)))$. Also, for each edge $g : \Gamma_1(i,j)$ and $a :A$, we have a chain $\eta_{i,j,g}(a)$ of paths
\[
\begin{tikzcd}
	{\transport^{x \mapsto  f_T(\left[\idd_A\right](x)) = \rec_{\colimm}(\F(r, K))(\psi(x))}\mleft(\kappa_{i,j,g}(a), \pr_2(r_j)(a)^{-1}\mright)} \\
	{\ap_{f_T}(\ap_{\left[\idd_A\right]}(\kappa_{i,j,g}(a))) ^{-1} \cdot \pr_2(r_j)(a)^{-1} \cdot  \ap_{\rec_{\colimm}(\F(r, K))}(\ap_{\psi}(\kappa_{i,j,g}(a)) ) } \\
	{\ap_{f_T}(\ap_{\left[\idd_A\right]}(\kappa_{i,j,g}(a))) ^{-1} \cdot \pr_2(r_j)(a)^{-1} \cdot  \ap_{\rec_{\colimm}(\F(r, K))}(\ap_{\iota_j}(\pr_2(F_{i,j,g})(a))^{-1} \cdot \kappa_{i,j,g}(\pr_2(F_i)(a)) ) } \\
	{\ap_{f_T}(\ap_{\left[\idd_A\right]}(\kappa_{i,j,g}(a)))^{-1} \cdot \pr_2(r_j)(a)^{-1} \cdot  \ap_{\pr_1(r_j)}(\pr_2(F_{i,j,g})(a))^{-1} \cdot \pr_1(K_{i,j,g})(\pr_2(F_i)(a))} \\
	{\left(\pr_1(K_{i,j,g})(\pr_2(F_i)(a))^{-1}  \cdot  \ap_{\pr_1(r_j)}(\pr_2(F_{i,j,g})(a)) \cdot \pr_2(r_j)(a)  \right)^{-1}} \\
	{\pr_2(r_i)(a)^{-1}}
	\arrow["{{{{\textit{via \cref{transfunc}}}}}}", equals, from=1-1, to=2-1]
	\arrow["{{{{\textit{via $\beta_{\psi}(i,j,g,a)$}}}}}", equals, from=2-1, to=3-1]
	\arrow["{{{{\textit{via $\beta_{\rec_{\colimm}(\F(r,K))}(i,j,g,\pr_2(F_i)(a))$}}}}}", equals, from=3-1, to=4-1]
	\arrow["{{\textit{via $\beta_{\left[\idd_A\right]}(i,j,g,a)$}}}", equals, from=4-1, to=5-1]
	\arrow["{{\ap_{{-}^{-1}}(\pr_2(K_{i,j,g})(a))}}", equals, from=5-1, to=6-1]
\end{tikzcd} \]
This gives us a function 
\[ \label{sigmafun}
\sigma \ : \ \prod_{x : \colimm{A}}f_T(\left[\idd_A\right](x)) = \rec_{\colimm}(\F(r,K))(\psi(x))
\tag{$\texttt{coc-forg}$}
\] and thus the cogap map $h_{r,K}: \P_A(F) \to T$:
\[
\begin{tikzcd}[row sep = large,  column sep = 40]
	{\colimm{A}} & {\colimm(\F(F))} \\
	A & {\P_A(F)} \\
	&& T
	\arrow[from=1-1, to=1-2]
	\arrow[from=1-1, to=2-1]
	\arrow[from=1-2, to=2-2]
	\arrow["{{\rec_{\colimm}(\F(r, K))}}", curve={height=-12pt}, from=1-2, to=3-3]
	\arrow[from=2-1, to=2-2]
	\arrow["{{f_T}}"', curve={height=12pt}, from=2-1, to=3-3]
	\arrow["{{h_{r,K}}}"{pos=0.4}, dashed, from=2-2, to=3-3]
\end{tikzcd}
\]
Since $h(\inl(a)) \equiv f_T(a)$, we can form the map
$\cogap_A(r,K) \coloneqq \left(h_{r,K}, \refl_{f_T({-})}\right)  :  \P_A(F) \to_A T$.  
\begin{remark}
The cogap map for the $A$-colimit is quite tractable. On $\inr$ it is definitionally equal to the cogap map of $\colimm$ in $\U$, and on $\inl$ to the given function $A \to T$.
\end{remark}

\noindent \textbf{The homotopy $\pstc_{F,T} \circ \cogap_A(r,K) \sim \idd_{\mathsf{Cocone}_F(T, f_T)}$~\cite[\href{https://github.com/PHart3/colimits-agda/tree/v0.5.0/Colimit-coslice/R-L-R}{R-L-R}]{agda-colim-TR}:}

\medskip

We have the definitional equalities
\[
\hspace*{-2.1cm}
\begin{tikzcd}[row sep = small]
	{\pstc_{F,T}(h_{r,K}, \refl_{f_T({-})})} \\
	{\left( \lambda{i}.\left(h_{r,K} \circ \inr \circ \iota_i, \ap_{h_{r,K}}(\tau_i(a)) \cdot \refl_{f_T(a)} \right) ,  \lambda{i}\lambda{j}\lambda{g}.\left(\lambda{x}.\ap_{h_{r,K}}(\delta_{i,j,g}(x)), \lambda{a}. \Theta_{\delta_{i,j,g}}(h_{r,K}^{\ast},a) \cdot  \ap_{\ap_{h_{r,K}}({-}) \cdot \refl_{f_T(a)}}(\epsilon_{i,j,g}(a)) \right) \right)}
	\arrow[Rightarrow, scaling nfold=3, no head, from=1-1, to=2-1]
\end{tikzcd}\]
and $h_{r,K} \circ \inr \circ \iota_i \equiv \pr_1(r_i)$. For each $i: \Gamma_0$ and $a : A$, we have a chain $P_i(a)$ of paths
{\allowdisplaybreaks
\begin{align*}
& \ \ap_{h_{r,K}}(\tau_i(a)) \cdot \refl_{f_T(a)} 
\\ = & \ \ap_{h_{r,K}}(\glue_{\P_A(F)}(\iota_i(a)))^{-1} \tag{$\Delta_i(a) \coloneqq \PI(\glue_{\P_A(F)}(\iota_i(a)))$}
\\ = & \ \left(\pr_2(r_i)(a)^{-1}\right)^{-1} \tag{$\ap_{{-}^{-1}}(\beta_{h_{r,K}}(\iota_i(a)))$}
\\ = & \ \pr_2(r_i)(a)
\end{align*}}
Moreover, for all $g : \Gamma_1(i,j)$ and $x : \pr_1(F_i)$, we have a chain $Q_{i,j,g}(x)$
{\allowdisplaybreaks
\begin{align*}
& \ \ap_{h_{r,K}}(\delta_{i,j,g}(x)) \cdot \refl_{h_{r, K}(\inr(\iota_i(x)))} 
\\ \equiv & \  \ap_{h_{r,K}}(\ap_{\inr}(\kappa_{i,j,g}(x))) \cdot \refl_{h_{r, K}(\inr(\iota_i(x)))}
\\ = & \ \ap_{\rec_{\colimm}(\F(r,K))}(\kappa_{i,j,g}(x))
\\  = & \ \pr_1(K_{i,j,g})(x) \tag{$ \beta_{\rec_{\colimm}(\F(r,K))}(i,j,g,x) $}
\end{align*}}
By \cref{nattreq}, we want to prove that for all $g : \Gamma_1(i,j)$ and $a : A$,
\[ 
\begin{tikzcd}[ampersand replacement = \&]
	{\subalign{ & \ap_{{-}^{-1} \cdot \ap_{\pr_1(r_j)}(\pr_2(F_{i,j,g})(a)) \cdot \pr_2(r_j)(a) }(Q_{i,j,g}(\pr_2(F_i)(a)))^{-1} \cdot \\ &  \Xi(P, \left(\ap_{h_{r,K}}(\delta_{i,j,g}(x)), \Theta_{\delta_{i,j,g}}(h_{r,K}^{\ast},a) \cdot  \ap_{\ap_{h_{r,K}}({-}) \cdot \refl_{f_T(a)}}(\epsilon_{i,j,g}(a))\right), a)}} \\
	{\pr_2(K_{i,j,g})}
	\arrow[from=1-1, to=2-1, equal]
\end{tikzcd}
\]
To this end, recalling the function \eqref{sigmafun}, note that
\[
\begin{tikzcd}[ampersand replacement = \&]
    {\Delta_i(a) \cdot \ap_{{-}^{-1}}(\beta_{h_{r,K}}(\iota_i(a)))} \\
	{\transport^{x \mapsto \ap_{h_{r,K}}\mleft(\glue_{\P_A(F)}(x)^{-1}\mright) \cdot \refl_{f_T(\left[\idd_A\right](x))} = \sigma(x)^{-1}}(\kappa_{i,j,g}(a), \Delta_j(a) \cdot \ap_{{-}^{-1}}(\beta_{h_{r,K}}(\iota_j(a)) ))} \\
	{\subalign{& \apd_{\ap_{h_{r,K}}\mleft(\glue_{\P_A(F)}({-})^{-1}\mright) \cdot \refl_{f_T(\left[\idd_A\right]({-}))}}( \kappa_{i,j,g}(a))^{-1} \cdot \\ &  \ap_{\transport^{x \mapsto \rec_{\colimm}(\F(r,K))(\psi(x)) = f_T(\left[\idd_A\right](x))}(\kappa_{i,j,g}(a) , {-} )}(\Delta_j(a) \cdot \ap_{{-}^{-1}}(\beta_{h_{r,K}}(\iota_j(a)) )) \cdot \apd_{\sigma({-})^{-1} }( \kappa_{i,j,g}(a))}}
	\arrow[Rightarrow, no head, from=2-1, to=1-1]
	\arrow[Rightarrow, no head, from=3-1, to=2-1]
\end{tikzcd}
\] and that the triangle
\[
\begin{tikzcd}[column sep = large, row sep = large]
	{\transport^{x \mapsto \rec_{\colimm}(\F(r,K))(\psi(x)) = f_T(\left[\idd_A\right](x))}\mleft(\kappa_{i,j,g}(a), \left(\pr_2(r_j)(a)^{-1}\right)^{-1}\mright)} && {\left(\pr_2(r_i)(a)^{-1}\right)^{-1}} \\
	{\transport^{x \mapsto f_T(\left[\idd_A\right](x))= \rec_{\colimm}(\F(r,K))(\psi(x))}\mleft(\kappa_{i,j,g}(a), \pr_2(r_j)(a)^{-1}\mright)^{-1}}
	\arrow["{{\apd_{\sigma({-})^{-1} }( \kappa_{i,j,g}(a))}}", Rightarrow, no head, from=1-1, to=1-3]
	\arrow["{{\PI(\kappa_{i,j,g}(a))}}"', Rightarrow, no head, from=1-1, to=2-1]
	\arrow["{{\ap_{{-}^{-1}}(\apd_{\sigma}(\kappa_{i,j,g}(a)))}}"', Rightarrow, no head, from=2-1, to=1-3]
\end{tikzcd}
\] commutes, where $\apd_{\sigma}(\kappa_{i,j,g}(a)) = \eta_{i,j,g}(a)$ by the induction principle for $\sigma$.
Therefore, after unfolding $\Xi$, we want to show that for each $a : A$, $\pr_2(K_{i,j,g})(a)$ equals the chain $C_{\Xi}(a)$, displayed by \cref{bigcoher2}. We can reduce $C_{\Xi}(a)$ to $\pr_2(K_{i,j,g})(a)$, which appears in $\eta_{i,j,g}(a)$, in a bottom-up fashion. This process iteratively removes the $\beta$-rules appearing in $C_{\Xi}(a)$. We refer the reader to the Agda formalization for the full reduction.
\begin{figure}
\centering
\begin{tikzcd}
	{\pr_1(K_{i,j,g})(\pr_2(F_i)(a))^{-1} \cdot \ap_{\pr_1(r_j)}(\pr_2(F_{i,j,g})(a)) \cdot \pr_2(r_j)(a} \\
	{\ap_{h_{r,K}}(\delta_{i,j,g}(\pr_2(F_i)(a)))^{-1} \cdot \ap_{\pr_1(r_j)}(\pr_2(F_{i,j,g})(a)) \cdot \pr_2(r_j)(a)} \\
	{\ap_{h_{r,K}}(\delta_{i,j,g}(\pr_2(F_i)(a)))^{-1} \cdot \ap_{\pr_1(r_j)}(\pr_2(F_{i,j,g})(a)) \cdot \ap_{h_{r,K}}(\tau_j(a)) \cdot \refl_{f_T(a)}} \\
	{\ap_{h_{r,K}}(\tau_i(a)) \cdot \refl_{f_T(a)}} \\
	{\transport^{x \mapsto \rec_{\colimm}(\F(r,K))(\psi(x)) = f_T(\left[\idd_A\right](x))}(\kappa_{i,j,g}(a) , \ap_{h_{r,K}}(\tau_j(a)) \cdot \refl_{f_T(a)} )} \\
	{\left(\kappa_{i,j,g}(a)\right)_{\ast}\mleft( \left(\pr_2(r_j)(a)^{-1}\right)^{-1}\mright)} \\
	{\left(\kappa_{i,j,g}(a)\right)_{\ast}\mleft(\pr_2(r_j)(a)^{-1}\mright)^{-1}} \\
	{\pr_2(r_i)(a)}
	\arrow["{{{{\textit{via $Q_{i,j,g}(\pr_2(F_i)(a))$}}}}}", equals, from=1-1, to=2-1]
	\arrow[""{name=0, anchor=center, inner sep=0}, "{{{\textit{via $P_j(a)$}}}}", equals, from=2-1, to=3-1]
	\arrow["{{{\textit{via $\epsilon_{i,j,g}(a)$}}}}", equals, from=3-1, to=4-1]
	\arrow["{{{\apd_{\ap_{h_{r,K}}\mleft(\glue_{\P_A(F)}({-})^{-1}\mright) \cdot \refl_{f_T(\left[\idd_A\right]({-}))}}( \kappa_{i,j,g}(a))}}}", equals, from=4-1, to=5-1]
	\arrow["{{{\ap_{\left(\kappa_{i,j,g}(a)\right)_{\ast}}(\Delta_j(a) \cdot \ap_{{-}^{-1}}(\beta_{h_{r,K}}(\iota_j(a)) ))}}}", equals, from=5-1, to=6-1]
	\arrow["{{{\PI(\kappa_{i,j,g}(a))}}}", equals, from=6-1, to=7-1]
	\arrow["{\textit{via $\eta_{i,j,g}(a)$}}", equals, from=7-1, to=8-1]
	\arrow[equals, from=2-1, to=0]
\end{tikzcd}
\caption{$C_{\Xi}(a)$}
\label{bigcoher2}
\end{figure}

\medskip

\noindent \textbf{The homotopy $\cogap_A(r,K) \circ \pstc_{F,T} \sim  \idd_{\P_A(F) \to_A T}$~\cite[\href{https://github.com/PHart3/colimits-agda/tree/v0.5.0/Colimit-coslice/L-R-L}{L-R-L}]{agda-colim-TR}:}

\medskip

Suppose that $\left(f, f_p\right) : \P_A(F) \to_A T$ and let $\zeta_1 \coloneqq \pr_1(\pstc_{F,T}(f,f_p))$ and $\zeta_2 \coloneqq \pr_2(\pstc_{F,T}(f,f_p))$. Letting $\tilde{h} \coloneqq h_{\zeta_1, \zeta_2}$, we want to construct functions
$\alpha  :  \prod_{x : \P_A(F)}f(x) = \tilde{h}(x)$ and $\widehat{\alpha}  : \prod_{a : A}\alpha(\inl(a))^{-1} \cdot f_p(a) = \refl_{f_T(a)}$.
To construct $\alpha$, we use \cref{funcpo}. For each $a : A$, $f_p(a)$ witnesses that $f(\inl(a)) = \tilde{h}(\inl(a))$.
Already, we see that once $\alpha$ is constructed, it is easy to derive $\widehat{\alpha}$ from it.

Continuing with $\alpha$, we see
$f(\inr(\iota_i(x))) \equiv \tilde{h}(\inr(\iota_i(x)))$. We also have a chain $V_{i,j,g}(x)$ of paths:
\[\begin{tikzcd}
	{\transport^{y \mapsto f(\inr(y)) = \tilde{h}(\inr(y))  }(\kappa_{i,j,g}(x), \refl_{f(\inr(\iota_j(F_{i,j,g}(x))))  } )} \\
	{\ap_f(\ap_{\inr}(\kappa_{i,j,g}(x)))^{-1} \cdot \ap_{\rec_{\colimm}(\F(\zeta_1, \zeta_2 )) }(\kappa_{i,j,g}(x))} \\
	{\refl_{f(\inr(\iota_i(x)))}}
	\arrow["{{{\textit{via \cref{transfunc}}}}}", equals, from=1-1, to=2-1]
	\arrow["{{{\textit{via $\beta_{\rec_{\colimm}(\F(\zeta_1,\zeta_2)}$}(i,j,g,x)}}}", equals, from=2-1, to=3-1]
\end{tikzcd} \]
By induction on $\colimm(\F(F))$, this gives us a term
$\gamma : \prod_{x : \colimm(\F(F))}f(\inr(x)) = \tilde{h}(\inr(x))$.
For all $i : \Gamma_0$ and $a :A$, we have a chain $R_i(a)$ of paths:
{\allowdisplaybreaks
\begin{align*}
&  \ \left(\ap_f(\glue_{\P_A(F)}(\iota_i(a)))^{-1} \cdot f_p(a) \right) \cdot \ap_{\tilde{h}}(\glue_{\P_A(F)}(\iota_i(a)))  
\\ =  & \ \left( \ap_f(\glue_{\P_A(F)}(\iota_i(a)))^{-1} \cdot f_p(a) \right ) \cdot \left(\ap_f(\tau_i(a)) \cdot f_p(a)\right)^{-1} 
\tag{\textit{via $\beta_{\tilde{h}}(\iota_i(a)) $}}
\\ =  & \ \refl_{f(\inr(\iota_i(\pr_2(F_i)(a))))} \tag{$M_i(a) \coloneqq \PI(\glue_{\P_A(F)}(\iota_i(a)), f_p(a) )$}
\\ \equiv  & \  \gamma(\psi(\iota_i(a) ))
\end{align*}}

Further, by \cref{transfunc} again, for all $g : \Gamma_1(i,j)$ and $a : A$,
\[
\begin{tikzcd}[ampersand replacement =\&]
	{\transport^{x \mapsto \left(\ap_f(\glue_{\P_A(F)}(x))^{-1} \cdot f_p(\left[\idd_A\right](x)) \right) \cdot \ap_{\tilde{h}}(\glue_{\P_A(F)}(x)) =_{f(\inr(\psi(x))) = \tilde{h}(\inr(\psi(x)))} \gamma(\psi(x))}(\kappa_{i,j,g}(a), R_j(a) )} \\
	{\subalign{& \apd_{\left(\ap_f(\glue_{\P_A(F)}({-}))^{-1} \cdot  f_p(\left[\idd_A\right]({-})) \right) \cdot \ap_{\tilde{h}}(\glue_{\P_A(F)}({-}))}(\kappa_{i,j,g}(a) )^{-1} \cdot \\ & \ap_{\transport^{x \mapsto f(\inr(\psi(x))) = \tilde{h}(\inr(\psi(x)))}(\kappa_{i,j,g}(a) ,{-})}(R_j(a) ) \cdot \apd_{\gamma(\psi({-}))}(\kappa_{i,j,g}(a))}}
	\arrow[from=1-1, to=2-1, equal]
\end{tikzcd} 
\] We want to prove the bottom path equals $R_i(a)$. By \cref{apdnat}, we have the commuting square
\[\adjustbox{scale=.88}
{
\hspace*{-2.5cm}
\begin{tikzcd}[row sep=large, column sep = 23]
	{\left(\kappa_{i,j,g}(a)\right)_{\ast}\mleft(\left(\ap_f(\glue_{\P_A(F)}(\iota_j(a)))^{-1} \cdot f_p(a) \right) \cdot \ap_{\tilde{h}}(\glue_{\P_A(F)}(\iota_j(a)))\mright)} && {\left(\ap_f(\glue_{\P_A(F)}(\iota_i(a)))^{-1} \cdot f_p(a) \right) \cdot \ap_{\tilde{h}}(\glue_{\P_A(F)}(\iota_i(a)))} \\
	\\
	{\left(\kappa_{i,j,g}(a)\right)_{\ast}\mleft(\left(\ap_f(\glue_{\P_A(F)}(\iota_j(a)))^{-1} \cdot f_p(a) \right) \cdot \left(\ap_f(\tau_j(a)) \cdot f_p(a)\right)^{-1}\mright)} && {\left(\ap_f(\glue_{\P_A(F)}(\iota_i(a)))^{-1} \cdot f_p(a) \right) \cdot \left(\ap_f(\tau_i(a)) \cdot f_p(a)\right)^{-1}}
	\arrow["{\apd_{\left(\ap_f(\glue_{\P_A(F)}({-}))^{-1} \cdot f_p(\left[\idd_A\right]({-})) \right) \cdot \ap_{\tilde{h}}(\glue_{\P_A(F)}({-}))}(\kappa_{i,j,g}(a) )}", shift left=2, Rightarrow, no head, from=1-1, to=1-3]
	\arrow["{\textit{via $\beta_{\tilde{h}}(\iota_j(a))$}}"', Rightarrow, no head, from=1-1, to=3-1]
	\arrow["{\apd_{\left(\ap_f(\glue_{\P_A(F)}({-}))^{-1} \cdot f_p(\left[\idd_A\right]({-})) \right) \cdot \sigma({-})}(\kappa_{i,j,g}(a) )}"', shift right=2, Rightarrow, no head, from=3-1, to=3-3]
	\arrow["{\textit{via $\beta_{\tilde{h}}(\iota_i(a))$}}"', Rightarrow, no head, from=3-3, to=1-3]
\end{tikzcd}  } 
\]
Therefore, it suffices to show that
\[\begin{tikzcd}
	{\ap_{\transport^{x \mapsto f(\inr(\psi(x))) = \tilde{h}(\inr(\psi(x)))}(\kappa_{i,j,g}(a) ,{-})}(M_j(a) ) \cdot \apd_{\gamma(\psi({-}))}(\kappa_{i,j,g}(a))} \\
	{\apd_{\left(\ap_f( \glue_{\P_A(F)}({-}))^{-1} \cdot  f_p(\left[\idd_A\right]({-})) \right) \cdot \sigma({-})}(\kappa_{i,j,g}(a)) \cdot M_i(a)}
	\arrow[equals, from=1-1, to=2-1]
\end{tikzcd} \tag{$\texttt{$M$-coher}$} \label{Mcoher} \]
We begin with the two $\apd$ terms appearing in \eqref{Mcoher}. We have the following two commuting triangles:
\[ 
\hspace*{-2.5cm}
\adjustbox{scale=.88} {
\begin{tikzcd}[row sep = 35,  /tikz/row 2/.append style={row sep=1pt}, column sep = {140, between origins}]
	& {\left(\kappa_{i,j,g}(a)\right)_{\ast}\mleft(\mleft(\ap_f(\glue_{\P_A(F)}(\iota_j(a)))^{-1} \cdot f_p(a) \mright) \cdot \sigma(\iota_j(a))\mright)} \\
	{\left(\ap_f(\glue_{\P_A(F)}(\iota_i(a)))^{-1} \cdot f_p(a) \right) \cdot \left(\kappa_{i,j,g}(a)\right)_{\ast}( \sigma(\iota_j(a)))} && {\left(\ap_f(\glue_{\P_A(F)}(\iota_i(a)))^{-1} \cdot f_p(a) \right) \cdot \sigma(\iota_i(a))} \\
	\\
	& {\left(\kappa_{i,j,g}(a) \right)_{\ast}(\gamma(\psi(\iota_j(a))))} \\
	{\ap_{\psi}(\kappa_{i,j,g}(a))_{\ast}(\gamma(\psi(\iota_j(a))))} && {\gamma(\psi(\iota_i(a)))}
	\arrow["{{{\PI(\kappa_{i,j,g}(a))}}}"', Rightarrow, no head, from=1-2, to=2-1]
	\arrow["{{{\apd_{\left(\ap_f( \glue_{\P_A(F)}({-}))^{-1} \cdot f_p(\left[\idd_A\right]({-})) \right) \cdot \sigma({-})}(\kappa_{i,j,g}(a))}}}", Rightarrow, no head, from=1-2, to=2-3]
	\arrow["{{{\ap_{\left(\ap_f(\glue_{\P_A(F)}(\iota_i(a)))^{-1} \cdot f_p(a) \right) \cdot {-}}( \apd_{\sigma}(\kappa_{i,j,g}(a)))}}}"', Rightarrow, no head, from=2-1, to=2-3]
	\arrow["{{{\PI(\kappa_{i,j,g}(a))}}}"', Rightarrow, no head, from=4-2, to=5-1]
	\arrow["{{{\apd_{\gamma(\psi({-}))}(\kappa_{i,j,g}(a))}}}", Rightarrow, no head, from=4-2, to=5-3]
	\arrow["{{{\apd_{\gamma}(\ap_{\psi}(\kappa_{i,j,g}(a)))}}}"', Rightarrow, no head, from=5-1, to=5-3]
\end{tikzcd}}
\]
 Note that $\apd_{\sigma}(\kappa_{i,j,g}(a)) =  \eta_{i,j,g}(a)$ by the induction principle for $\colimm{A}$. In addition, 
\[
\begin{tikzcd}
	{\apd_{\gamma}(\ap_{\psi}(\kappa_{i,j,g}(a)))} \\
	{\ap_{{-}_{\ast}(\gamma(\psi(\iota_j(a))))}(\beta_{\psi}(i,j,g,a)) \cdot \apd_{\gamma}(\ap_{\iota_j}(\pr_2(F_{i,j,g})(a))^{-1} \cdot \kappa_{i,j,g}(\pr_2(F_i)(a)))} \\
	{\ap_{{-}_{\ast}(\refl_{f(\inr(\iota_j(\pr_2(F_j)(a))))})}(\beta_{\psi}(i,j,g,a)) \cdot \PI(\pr_2(F_{i,j,g})(a)) \cdot \apd_{\gamma}(\kappa_{i,j,g}(\pr_2(F_i)(a)))}
	\arrow[Rightarrow, no head, from=1-1, to=2-1]
	\arrow[Rightarrow, no head, from=2-1, to=3-1]
\end{tikzcd}
\] where $\PI(\pr_2(F_{i,j,g})(a))$ has type
\[\begin{tikzcd}
	{\left(\ap_{\iota_j}(\pr_2(F_{i,j,g})(a))^{-1} \cdot \kappa_{i,j,g}(\pr_2(F_i)(a))\right)_{\ast}(\refl_{f(\inr(\iota_j(\pr_2(F_j)(a))))})} \\
	{ \left(\kappa_{i,j,g}(\pr_2(F_i)(a)) \right)_{\ast}(\gamma(\iota_j(F_{i,j,g}(\pr_2(F_i)(a)))) )}
	\arrow[Rightarrow, no head, from=1-1, to=2-1]
\end{tikzcd}\]
Note that $\apd_{\gamma}(\kappa_{i,j,g}(\pr_2(F_i)(a))) = V_{i,j,g}(\pr_2(F_i)(a))$ by the induction principle for $\colimm(\F(F))$.

Now, let $Y_{i,j,g}(a)  \coloneqq \Theta_{\ap_{\inr}(\kappa_{i,j,g}({-}))}(f^{\ast},a) \cdot \ap_{\ap_f({-}) \cdot f_p(a)}(\epsilon_{i,j,g}(a))$ and consider the following chain $\chi(s)$ of paths for each $s : f(\inr(\psi(\iota_j(a)))) = \tilde{h}(\inr(\psi(\iota_j(a))))$:
\[ 
\hspace*{-3cm}\begin{tikzcd}[ampersand replacement = \&]
	{\transport^{x \mapsto f(\inr(\psi(x))) = \tilde{h}(\inr(\psi(x)))}(\kappa_{i,j,g}(a) ,s)} \\
	{\ap_f(\ap_{\inr}(\ap_{\psi}(\kappa_{i,j,g}(a))))^{-1} \cdot s \cdot \ap_{\rec_{\colimm}(\F(\zeta_1, \zeta_2))}(\ap_{\psi}(\kappa_{i,j,g}(a)))} \\
	{\ap_f(\ap_{\inr}(\ap_{\psi}(\kappa_{i,j,g}(a))))^{-1} \cdot s \cdot \ap_{\rec_{\colimm}(\F(\zeta_1, \zeta_2))}(\ap_{\iota_j}(\pr_2(F_{i,j,g})(a))^{-1} \cdot \kappa_{i,j,g}(\pr_2(F_i)(a)))} \\
	{ \ap_f(\ap_{\inr}(\ap_{\psi}(\kappa_{i,j,g}(a))))^{-1} \cdot s \cdot \left( \ap_f(\tau_j(a)) \cdot f_p(a) \right) \cdot \left(\ap_f(\ap_{\inr}(\kappa_{i,j,g}(\pr_2(F_i)(a))))^{-1} \cdot \ap_{f \circ \inr \circ \iota_j}(\pr_2(F_{i,j,g})(a)) \cdot \ap_f(\tau_j(a)) \cdot f_p(a)  \right)^{-1}} \\
	{\ap_f(\ap_{\inr}(\ap_{\psi}(\kappa_{i,j,g}(a))))^{-1} \cdot s \cdot \left(\ap_f(\tau_j(a)) \cdot f_p(a)\right) \cdot \left(\ap_f(\tau_i(a)) \cdot f_p(a)\right)^{-1} } \\
	{\ap_f(\ap_{\inr}(\ap_{\psi}(\kappa_{i,j,g}(a))))^{-1} \cdot s \cdot \left(\ap_f(\ap_{\inr}(\ap_{\psi}(\kappa_{i,j,g}(a)))) \cdot \left(\ap_f(\tau_i(a)) \cdot f_p(a)\right) \cdot \refl_{f_T(a)}\right) \cdot \left(\ap_f(\tau_i(a)) \cdot f_p(a)\right)^{-1} } \\
	{
					  \ap_f(\ap_{\inr}(\ap_{\psi}(\kappa_{i,j,g}(a))))^{-1} \cdot s \cdot \ap_f(\ap_{\inr}(\ap_{\psi}(\kappa_{i,j,g}(a))))
					} \\
	{\transport^{x \mapsto f(\inr(\psi(x))) = f(\inr(\psi(x)))}(\kappa_{i,j,g}(a) , s)}
	\arrow["{\textit{via $\cref{transfunc}$}}", Rightarrow, no head, from=1-1, to=2-1]
	\arrow["{\textit{via $\beta_{\psi}(i,j,g, a)$}} ", Rightarrow, no head, from=2-1, to=3-1]
	\arrow["{\textit{via $\mu_1(i,j,g,a)$} }", Rightarrow, no head, from=3-1, to=4-1]
	\arrow["{\textit{via $Y_{i,j,g}(a)$}}", Rightarrow, no head, from=4-1, to=5-1]
	\arrow["{\textit{via $\mu_2(i,j,g,a)$}}", Rightarrow, no head, from=5-1, to=6-1]
	\arrow["{{{{\PI(\tau_i(a),f_p(a), \kappa_{i,j,g}(a))}}}}", Rightarrow, no head, from=6-1, to=7-1]
	\arrow["{\textit{via $\cref{transfunc}$}}", Rightarrow, no head, from=7-1, to=8-1]
\end{tikzcd}
\] where $\mu_1(i,j,g,a)$ and $\mu_2(i,j,g,a)$ denote, respectively, the chains of paths
\[
\begin{tikzcd}[/tikz/row 4/.append style={row sep=5pt}, /tikz/row 5/.append style={row sep=5pt}]
	{\ap_{\rec_{\colimm}(\F(\zeta_1, \zeta_2))}(\ap_{\iota_j}(\pr_2(F_{i,j,g})(a))^{-1} \cdot \kappa_{i,j,g}(\pr_2(F_i)(a)))} \\
	{\ap_{\rec_{\colimm}(\F(\zeta_1,\zeta_2)) \circ \iota_j}(\pr_2(F_{i,j,g})(a))^{-1} \cdot \ap_{\rec_{\colimm}(\F(\zeta_1, \zeta_2))}( \kappa_{i,j,g}(\pr_2(F_i)(a)))} \\
	{\ap_{f \circ \inr \circ \iota_j}(\pr_2(F_{i,j,g})(a))^{-1} \cdot   \ap_f(\ap_{\inr}(\kappa_{i,j,g}(\pr_2(F_i)(a))))} \\
	{\left( \ap_f(\tau_j(a)) \cdot f_p(a) \right) \cdot \left(\ap_f(\ap_{\inr}(\kappa_{i,j,g}(\pr_2(F_i)(a))))^{-1} \cdot \ap_{f \circ \inr \circ \iota_j}(\pr_2(F_{i,j,g})(a)) \cdot \ap_f(\tau_j(a)) \cdot f_p(a)  \right)^{-1}} \\
	\\
	{\ap_f(\tau_j(a)) \cdot f_p(a)} \\
	{\transport^{y \mapsto f(\inr(\psi(y))) = f_T(\left[\idd_A\right](y))}(\kappa_{i,j,g}(a)^{-1}, \ap_f(\tau_i) \cdot f_p(a))} \\
	{\ap_f(\ap_{\inr}(\ap_{\psi}(\kappa_{i,j,g}(a)))) \cdot \left(\ap_f(\tau_i) \cdot f_p(a) \right) \cdot \ap_{f_T}(\ap_{\left[\idd_A\right]}(\kappa_{i,j,g}(a)))^{-1}} \\
	{\ap_f(\ap_{\inr}(\ap_{\psi}(\kappa_{i,j,g}(a)))) \cdot \left(\ap_f(\tau_i(a)) \cdot f_p(a)\right) \cdot \refl_{f_T(a)}}
	\arrow[equals, from=1-1, to=2-1]
	\arrow["{{\textit{via $\beta_{\rec_{\colimm}(\F(\zeta_1, \zeta_2))}(i,j,g, \pr_2(F_i)(a))$}}}", equals, from=2-1, to=3-1]
	\arrow[equals, from=3-1, to=4-1]
	\arrow["{{{{\apd_{\ap_f\mleft(\glue_{\P_A(F)}({-})^{-1}\mright) \cdot f_p(\left[\idd_A\right]({-}))}\mleft(\kappa_{i,j,g}(a)^{-1}\mright)}}}}", equals, from=6-1, to=7-1]
	\arrow["{{\textit{via $\cref{transfunc}$}}}", equals, from=7-1, to=8-1]
	\arrow["{{\textit{via $\beta_{\left[\idd_A\right]}(i,j,g,a)$}}}", equals, from=8-1, to=9-1]
\end{tikzcd}  \]
By homotopy naturality, we have the commuting square
\[ 
\hspace*{-.5cm}
\begin{tikzcd}[column sep = 35, row sep = 60]
	{\transport^{x \mapsto f(\inr(\psi(x))) = \tilde{h}(\inr(\psi(x)))}(\kappa_{i,j,g}(a) ,s_1)} && {\transport^{x \mapsto f(\inr(\psi(x))) = \tilde{h}(\inr(\psi(x)))}(\kappa_{i,j,g}(a) ,s_2)} \\
	{\transport^{x \mapsto f(\inr(\psi(x))) = f(\inr(\psi(x)))}(\kappa_{i,j,g}(a) ,s_1)} && {\transport^{x \mapsto f(\inr(\psi(x))) = f(\inr(\psi(x)))}(\kappa_{i,j,g}(a) ,s_2)}
	\arrow["{\textit{via $M_j(a)$}}", equals, from=1-1, to=1-3]
	\arrow["{ \chi\mleft(\left( \ap_f( \glue_{\P_A(F)}(\iota_j(a)))^{-1} \cdot f_p(a) \right ) \cdot \left(\ap_f(\tau_j(a)) \cdot f_p(a)\right)^{-1}\mright)}"{description}, equals, from=1-1, to=2-1]
	\arrow["{\chi(\refl_{f(\inr(\iota_j(\pr_2(F_j)(a))))})}"{description}, equals, from=1-3, to=2-3]
	\arrow["{\textit{via $M_j(a)$}}"', equals, from=2-1, to=2-3]
\end{tikzcd}
\]
with $s_1 \coloneqq \left( \ap_f( \glue_{\P_A(F)}(\iota_j(a)))^{-1} \cdot f_p(a) \right ) \cdot \left(\ap_f(\tau_j(a)) \cdot f_p(a)\right)^{-1}$ and $s_2 \coloneqq \refl_{f(\inr(\iota_j(\pr_2(F_j)(a))))}$.
One can also prove that the following square commutes:
\[ 
\hspace*{-2.5cm}
\begin{tikzcd}[row sep = 60, ampersand replacement=\&]
	{\left(\ap_f(\glue_{\P_A(F)}(\iota_i(a)))^{-1} \cdot f_p(a)  \right) \cdot \left(\ap_f(\tau_i(a)) \cdot f_p(a) \right)^{-1}} \&[9 mm] {\subalign{& \transport^{x \mapsto f(\inr(\psi(x))) = f(\inr(\psi(x)))} \\ & \quad \mleft(\kappa_{i,j,g}(a) , \left( \ap_f( \glue_{\P_A(F)}(\iota_j(a)))^{-1} \cdot f_p(a) \right ) \cdot \left(\ap_f(\tau_j(a)) \cdot f_p(a)\right)^{-1}\mright)}} \\
	{\subalign{ & \left(\ap_f(\glue_{\P_A(F)}(\iota_i(a)))^{-1} \cdot f_p(a)  \right) \cdot \\ & \transport^{x \mapsto f_T(\left[\idd_A\right](x)) =  \tilde{h}(\inr(\psi(x))) }\mleft(\kappa_{i,j,g}(a), \left(\ap_f(\tau_j(a))\cdot f_p(a)\right)^{-1}\mright)}} \&[9 mm] {\subalign{& \transport^{x \mapsto f(\inr(\psi(x))) = \tilde{h}(\inr(\psi(x)))} \\ & \quad \mleft(\kappa_{i,j,g}(a) , \left( \ap_f( \glue_{\P_A(F)}(\iota_j(a)))^{-1} \cdot f_p(a) \right ) \cdot \left(\ap_f(\tau_j(a)) \cdot f_p(a)\right)^{-1}\mright)}}
	\arrow["{{{\PI(\kappa_{i,j,g}(a))}}}", Rightarrow,  no head, from=1-1, to=1-2]
	\arrow["{\textit{via $\eta_{i,j,g}(a)$}}"{description}, Rightarrow, no head, from=1-1, to=2-1]
	\arrow["{{{\PI(\kappa_{i,j,g}(a))}}}"', Rightarrow, no head, from=2-1, to=2-2]
	\arrow["{{{\chi\mleft(\left( \ap_f( \glue_{\P_A(F)}(\iota_j(a)))^{-1} \cdot f_p(a) \right ) \cdot \left(\ap_f(\tau_j(a)) \cdot f_p(a)\right)^{-1}\mright)}}}"{description}, Rightarrow, no head, from=2-2, to=1-2]
\end{tikzcd}
\]
At this point, we have put the path
\[
\hspace*{-.2cm} \apd_{\left(\ap_f( \glue_{\P_A(F)}({-}))^{-1} \cdot f_p(\left[\idd_A\right]({-})) \right) \cdot \sigma({-})}(\kappa_{i,j,g}(a))^{-1} \cdot \chi\mleft(\left( \ap_f( \glue_{\P_A(F)}(\iota_j(a)))^{-1} \cdot f_p(a) \right ) \cdot \left(\ap_f(\tau_j(a)) \cdot f_p(a)\right)^{-1}\mright)
\] into a form that will be useful. 
We want to do the same for the path 
\[
\chi\mleft(\refl_{f(\inr(\iota_j(\pr_2(F_j)(a))))}\mright)^{-1} \cdot \apd_{\gamma(\psi({-}))}(\kappa_{i,j,g}(a))
\] To this end, consider the following three chains of paths:

\[
\hspace*{-2.9cm}
\begin{tikzcd}[/tikz/row 3/.append style={row sep=7pt},
/tikz/row 4/.append style={row sep=7pt},
/tikz/row 7/.append style={row sep=7pt},
/tikz/row 8/.append style={row sep=7pt}]
	{\ap_f(\ap_{\inr}(\ap_{\psi}(\kappa_{i,j,g}(a))))^{-1}  \cdot \left(\ap_f(\ap_{\inr}(\ap_{\psi}(\kappa_{i,j,g}(a)))) \cdot \left(\ap_f(\tau_i(a)) \cdot f_p(a)\right) \cdot \refl_{f_T(a)}\right) \cdot \left(\ap_f(\tau_i(a)) \cdot f_p(a)\right)^{-1} } \\
	{\ap_f(\ap_{\inr}(\ap_{\psi}(\kappa_{i,j,g}(a))))^{-1}  \cdot \left(\ap_f(\tau_j(a)) \cdot f_p(a)\right) \cdot \left(\ap_f(\tau_i(a)) \cdot f_p(a)\right)^{-1} } \\
	{\ap_f(\ap_{\inr}(\ap_{\psi}(\kappa_{i,j,g}(a))))^{-1}  \cdot \left( \ap_f(\tau_j(a)) \cdot f_p(a) \right) \cdot \left(\ap_f(\ap_{\inr}(\ap_{\psi}(\kappa_{i,j,g}(a)))^{-1} \cdot \tau_j(a) \cdot \refl_{\inl(a)}) \cdot f_p(a)  \right)^{-1}} \\
	\\
	{\ap_f(\ap_{\inr}(\ap_{\psi}(\kappa_{i,j,g}(a))))^{-1}  \cdot \left( \ap_f(\tau_j(a)) \cdot f_p(a) \right) \cdot \left(\ap_f(\ap_{\inr}(\ap_{\psi}(\kappa_{i,j,g}(a)))^{-1} \cdot \tau_j(a) \cdot \refl_{\inl(a)}) \cdot f_p(a)  \right)^{-1}} \\
	{\ap_f(\ap_{\inr}(\ap_{\psi}(\kappa_{i,j,g}(a))))^{-1}  \cdot \left( \ap_f(\tau_j(a)) \cdot f_p(a) \right) \cdot \left(\ap_f(\ap_{\inr}(\ap_{\iota_j}(\pr_2(F_{i,j,g})(a))^{-1} \cdot \kappa_{i,j,g}(\pr_2(F_i)(a)))^{-1} \cdot \tau_j(a)) \cdot f_p(a)  \right)^{-1}} \\
	{\ap_f(\ap_{\inr}(\ap_{\psi}(\kappa_{i,j,g}(a))))^{-1} \cdot \left( \ap_f(\tau_j(a)) \cdot f_p(a) \right) \cdot \left(\ap_f(\ap_{\inr}(\kappa_{i,j,g}(\pr_2(F_i)(a))))^{-1} \cdot \ap_{f \circ \inr \circ \iota_j}(\pr_2(F_{i,j,g})(a)) \cdot \ap_f(\tau_j(a)) \cdot f_p(a)  \right)^{-1}} \\
	\\
	{\ap_f(\ap_{\inr}(\ap_{\psi}(\kappa_{i,j,g}(a))))^{-1}  \cdot \left( \ap_f(\tau_j(a)) \cdot f_p(a) \right) \cdot \left(\ap_f(\ap_{\inr}(\kappa_{i,j,g}(\pr_2(F_i)(a))))^{-1} \cdot \ap_{f \circ \inr \circ \iota_j}(\pr_2(F_{i,j,g})(a)) \cdot \ap_f(\tau_j(a)) \cdot f_p(a)  \right)^{-1}} \\
	{\ap_f(\ap_{\inr}(\ap_{\psi}(\kappa_{i,j,g}(a))))^{-1}  \cdot \ap_{\rec_{\colimm}(\F(\zeta_1, \zeta_2))}(\ap_{\iota_j}(\pr_2(F_{i,j,g})(a))^{-1} \cdot \kappa_{i,j,g}(\pr_2(F_i)(a)))} \\
	{\ap_f(\ap_{\inr}(\ap_{\psi}(\kappa_{i,j,g}(a))))^{-1}  \cdot \ap_{\rec_{\colimm}(\F(\zeta_1, \zeta_2))}(\ap_{\psi}(\kappa_{i,j,g}(a)))} \\
	{\transport^{x \mapsto f(\inr(\psi(x))) = \tilde{h}(\inr(\psi(x)))}(\kappa_{i,j,g}(a) ,\refl_{f(\inr(\iota_j(\pr_2(F_j)(a))))})} \\
	{\ap_{\psi}(\kappa_{i,j,g}(a))_{\ast}(\gamma(\psi(\iota_j(a))))} \\
	{\left(\ap_{\iota_j}(\pr_2(F_{i,j,g})(a))^{-1} \cdot \kappa_{i,j,g}(\pr_2(F_i)(a))\right)_{\ast}(\refl_{f(\inr(\iota_j(\pr_2(F_j)(a))))})} \\
	{\left(\kappa_{i,j,g}(\pr_2(F_i)(a)) \right)_{\ast}(\gamma(\iota_j(F_{i,j,g}(\pr_2(F_i)(a)))) )} \\
	{\ap_f(\ap_{\inr}(\kappa_{i,j,g}(\pr_2(F_i)(a))))^{-1} \cdot \ap_{\rec_{\colimm}(\F(\zeta_1, \zeta_2 ))}(\kappa_{i,j,g}(\pr_2(F_i)(a)))} \\
	{\ap_f(\ap_{\inr}(\kappa_{i,j,g}(\pr_2(F_i)(a))))^{-1} \cdot \ap_f(\ap_{\inr}(\kappa_{i,j,g}(\pr_2(F_i)(a))))}
	\arrow["{{{{{{{{\textit{via $\mu_2(i,j,g,a)$}}}}}}}}}", equals, from=1-1, to=2-1]
	\arrow["{{{\textit{via $E_3(i,j,g,a)$}}}}", equals, from=2-1, to=3-1]
	\arrow["{{{\textit{via $E_2(i,j,g,a)$}}}}", equals, from=5-1, to=6-1]
	\arrow["{{{\textit{via $E_1(i,j,g,a)$}}}}", equals, from=6-1, to=7-1]
	\arrow["{{\textit{via $\mu_1(i,j,g,a)$} }}", equals, from=9-1, to=10-1]
	\arrow["{{\textit{via $\beta_{\psi}(i,j,g, a)$  }}}", equals, from=10-1, to=11-1]
	\arrow["{{\textit{via \cref{transfunc}}}}", equals, from=11-1, to=12-1]
	\arrow["{{{{{{{\PI(\kappa_{i,j,g}(a))}}}}}}}", equals, from=12-1, to=13-1]
	\arrow["{{{\ap_{{-}_{\ast}(\refl_{f(\inr(\iota_j(\pr_2(F_j)(a))))})}(\beta_{\psi}(i,j,g,a))}}}", equals, from=13-1, to=14-1]
	\arrow["{{{\PI(\pr_2(F_{i,j,g})(a))}}}", equals, from=14-1, to=15-1]
	\arrow["{{\textit{via \cref{transfunc}}}}", equals, from=15-1, to=16-1]
	\arrow["{{\textit{via $\beta_{\rec_{\colimm}(\F(\zeta_1,\zeta_2))}(i,j,g,\pr_2(F_i)(a)) $}}}", equals, from=16-1, to=17-1]
\end{tikzcd} \]

\noindent We denote these chains by $P_1(i,j,g,a)$, $P_2(i,j,g,a)$, and $P_3(i,j,g,a)$, respectively. We can show that all three are equal to canonical paths $\PI$. As a consequence, we have the commuting diagram displayed by \cref{bigpent}. It's not hard to check that the bottom string of paths in \cref{bigpent} equals $\chi\mleft(\refl_{f(\inr(\iota_j(\pr_2(F_j)(a))))}\mright)^{-1} \cdot \apd_{\gamma(\psi({-}))}(\kappa_{i,j,g}(a))$, which therefore equals the top path: $\PI(\kappa_{i,j,g}(a))$.
Thus, we have produced a chain of paths
\[
\begin{tikzcd}[ampersand replacement=\&] 
	{\subalign{& \apd_{\left(\ap_f( \glue_{\P_A(F)}({-}))^{-1} \cdot f_p(\left[\idd_A\right]({-})) \right) \cdot \sigma({-})}(\kappa_{i,j,g}(a))^{-1}  \cdot \\ & \ap_{\transport^{x \mapsto f(\inr(\psi(x))) = \tilde{h}(\inr(\psi(x)))}(\kappa_{i,j,g}(a) ,{-})}(M_j(a)) \cdot \apd_{\gamma(\psi({-}))}(\kappa_{i,j,g}(a))}} \\
	{\PI(\kappa_{i,j,g}(a)) \cdot \ap_{\transport^{x \mapsto f(\inr(\psi(x))) = f(\inr(\psi(x)))}(\kappa_{i,j,g}(a) ,{-})}(M_j(a))  \cdot \PI(\kappa_{i,j,g}(a))} \\
	{M_i(a)}
	\arrow[Rightarrow, no head, from=1-1, to=2-1]
	\arrow["{\PI(\kappa_{i,j,g}(a)) }", Rightarrow, no head, from=2-1, to=3-1]
\end{tikzcd} \] which fulfills our goal: \eqref{Mcoher}.
\begin{figure}
\centering
 \adjustbox{scale=.88} {
 \hspace*{-2.5cm}
\begin{tikzcd}[column sep = -65, row sep = large, ampersand replacement=\&]  
	{\transport^{x \mapsto f(\inr(\psi(x))) = f(\inr(\psi(x)))}(\kappa_{i,j,g}(a) , \refl_{f(\inr(\iota_j(\pr_2(F_j)(a))))})} \&\& {\refl_{f(\inr(\iota_i(\pr_2(F_i)(a))))}} \\
	{ 					  \ap_f(\ap_{\inr}(\ap_{\psi}(\kappa_{i,j,g}(a))))^{-1} \cdot  \ap_f(\ap_{\inr}(\ap_{\psi}(\kappa_{i,j,g}(a)))) } \&\& {\ap_f(\ap_{\inr}(\kappa_{i,j,g}(\pr_2(F_i)(a))))^{-1} \cdot \ap_f(\ap_{\inr}(\kappa_{i,j,g}(\pr_2(F_i)(a))))} \\
	{\subalign{& \ap_f(\ap_{\inr}(\ap_{\psi}(\kappa_{i,j,g}(a))))^{-1} \cdot \\ & \left(\ap_f(\ap_{\inr}(\ap_{\psi}(\kappa_{i,j,g}(a)))) \cdot \left(\ap_f(\tau_i(a)) \cdot f_p(a)\right) \cdot \refl_{f_T(a)}\right) \cdot \left(\ap_f(\tau_i(a)) \cdot f_p(a)\right)^{-1}}} \&\& {\subalign{& \ap_f(\ap_{\inr}(\ap_{\psi}(\kappa_{i,j,g}(a))))^{-1} \cdot \left( \ap_f(\tau_j(a)) \cdot f_p(a) \right) \cdot \\ &  \left(\ap_f(\ap_{\inr}(\kappa_{i,j,g}(\pr_2(F_i)(a))))^{-1} \cdot \ap_{f \circ \inr \circ \iota_j}(\pr_2(F_{i,j,g})(a)) \cdot \ap_f(\tau_j(a)) \cdot f_p(a)  \right)^{-1}}} \\
	\&  {\subalign{& \ap_f(\ap_{\inr}(\ap_{\psi}(\kappa_{i,j,g}(a))))^{-1}  \cdot   \left( \ap_f(\tau_j(a)) \cdot f_p(a) \right) \cdot \\ & \left(\ap_f(\ap_{\inr}(\ap_{\psi}(\kappa_{i,j,g}(a)))^{-1} \cdot \tau_j(a) \cdot \refl_{\inl(a)}) \cdot f_p(a)  \right)^{-1}}}
	\arrow["{{{{{{\PI(\kappa_{i,j,g}(a))}}}}}}", Rightarrow, from=1-1, to=1-3]
	\arrow["{{\PI(\kappa_{i,j,g}(a))}}"', Rightarrow, from=1-1, to=2-1]
	\arrow["{{\PI(\tau_i(a),f_p(a), \kappa_{i,j,g}(a))}}"', Rightarrow, from=2-1, to=3-1]
	\arrow["{{{{{{\PI(\kappa_{i,j,g}(\pr_2(F_i)(a)))}}}}}}"', Rightarrow, from=2-3, to=1-3]
	\arrow["{{{P_1(i,j,g,a)}}}"', Rightarrow, from=3-1, to=4-2]
	\arrow["{{{P_3(i,j,g,a)}}}"', Rightarrow, from=3-3, to=2-3]
	\arrow["{{{P_2(i,j,g,a)}}}"', Rightarrow, from=4-2, to=3-3]
\end{tikzcd}}
\caption{reduction to canonical $\PI$ term}
\label{bigpent}  
 \end{figure}
\end{proof}
For our first application of \cref{mainequiv}, recall that a type is \emph{acyclic} if its suspension is contractible~\cite{acyc}.

\begin{corollary}
Pointed acyclic types are closed under pointed colimits $\colimm^{\ast}$.
\end{corollary}
\begin{proof}
Since $\Sigma : \U^{\ast} \to \U^{\ast}$ preserves colimits over graphs (\cref{suspcol}), $\Sigma(\colimm^{\ast}(F)) \simeq \colimm^{\ast}(\Sigma \circ F)$. If each $F_i$ is acyclic, then the second colimit is the colimit of the constant pointed diagram at $\1$, which is contractible as the cofiber of the identity function on $\colimm{\1}$.
\end{proof}

We should expect $\colimm^A$ to be left adjoint to the constant diagram functor. Before building the additional machinery to prove this, we record the easiest ingredient of the proof: naturality in the codomain.

\begin{lemma}[{\cite[\href{https://github.com/PHart3/colimits-agda/blob/v0.5.0/Colimit-coslice/Map-Nat/CosColimitPstCmp.agda}{CosColimitPstCmp}]{agda-colim-TR}}]\label{pstcompnatsq}
Let $F$ be an $A$-diagram over $\Gamma$. For every map $h^{\ast} : T \to_A U$, the square
\[
\begin{tikzcd}[column sep = small]
	{\left(\colimm(F) \to_A T \right)} &&&& {\left(\colimm(F) \to_A U\right)} \\
	{\mathsf{Cocone}_F(T)} &&&& { \mathsf{Cocone}_F(U)}
	\arrow["{{h^{\ast} \circ {-}}}", from=1-1, to=1-5]
	\arrow["{{\pstc_{F,T}}}"', from=1-1, to=2-1]
	\arrow["{{\pstc_{F,U}}}", from=1-5, to=2-5]
	\arrow["{{\mathsf{Cocone}_F(h^{\ast} \circ {-})}}"', shift right, from=2-1, to=2-5]
\end{tikzcd} 
\] commutes where $\mathsf{Cocone}_F(h^{\ast} \circ {-})$ is defined by right whiskering a given cocone by $h^{\ast}$.
\end{lemma}

\subsubsection*{Action on maps}

We now describe the action of $\colimm^A({-}) \coloneqq \left(\P_A({-}), \inl\right)$ on morphisms.
Suppose that $F$ and $G$ are $A$-diagrams over $\Gamma$. Consider a morphism $\delta \coloneqq \left(d, \left\langle{\xi, \tilde{\xi}}\right\rangle\right) : F \Rightarrow_A G$ of $A$-diagrams:
\[\begin{tikzcd}[column sep =huge]
	{F_i} & {F_j} \\
	{G_i} & {G_j}
	\arrow[""{name=0, anchor=center, inner sep=0}, "{F_{i,j,g}}", from=1-1, to=1-2]
	\arrow["{d_j}", from=1-2, to=2-2]
	\arrow["{d_i}"', from=1-1, to=2-1]
	\arrow[""{name=1, anchor=center, inner sep=0}, "{G_{i,j,g}}"', from=2-1, to=2-2]
	\arrow["{\left\langle{\xi_{i,j,g}, \tilde{\xi}_{i,j,g}}\right\rangle}"{description}, draw=none, from=0, to=1]
\end{tikzcd}\]
 The action on $\delta$ fits into a commuting square like so:
\[\begin{tikzcd}
	{F_i} && {G_i} \\
	{\colimm^A(F)} && {\colimm^A(G)}
	\arrow["{\colimm^A(\delta)}"', dashed, from=2-1, to=2-3]
	\arrow["{d_i}", from=1-1, to=1-3]
	\arrow["{\iota_i^G}", from=1-3, to=2-3]
	\arrow["{\iota_i^F}"', from=1-1, to=2-1]
\end{tikzcd}\]
Indeed, we have a function $\hat{\delta} : \colimm(\F(F)) \to \colimm(\F(G))$ induced by the diagram map over $\Gamma$
\[\begin{tikzcd}[column sep =huge]
	{\pr_1(F_i)} & {\pr_1(F_j)} \\
	{\pr_1(G_i)} & {\pr_1(G_j)}
	\arrow["{\pr_1(d_i)}"', from=1-1, to=2-1]
	\arrow[""{name=0, anchor=center, inner sep=0}, "{\pr_1(F_{i,j,g})}", from=1-1, to=1-2]
	\arrow["{\pr_1(d_j)}", from=1-2, to=2-2]
	\arrow[""{name=1, anchor=center, inner sep=0}, "{\pr_1(G_{i,j,g})}"', from=2-1, to=2-2]
	\arrow["{\xi_{i,j,g}}"{description}, draw=none, from=0, to=1]
\end{tikzcd}\] Note that for each $a : A$,
$\tilde{\xi}_{i,j,g}(a) :  \xi_{i,j,g}(\pr_2(F_i)(a))^{-1} \cdot \ap_{\pr_1(G_{i,j,g})}(\pr_2(d_i)(a)) \cdot \pr_2(G_{i,j,g})(a) = \ap_{\pr_1(d_j)}(\pr_2(F_{i,j,g})(a)) \cdot \pr_2(d_j)(a)$.
Letting $E_{i,j,g}(a) \coloneqq \ap_{\pr_1(G_{i,j,g})}(\pr_2(d_i)(a)) \cdot \pr_2(G_{i,j,g})(a) \cdot \pr_2(d_j)(a)^{-1} \cdot \ap_{\pr_1(d_j)}(\pr_2(F_{i,j,g})(a))^{-1}$, we may assume that $\tilde{\xi}_{i,j,g}(a)$ instead has the equivalent type
$\xi_{i,j,g}(\pr_2(F_i)(a))  = E_{i,j,g}(a)$.

We have a commuting triangle
\[ \label{tripsi}
\begin{tikzcd}
	& {\colimm{A}} & \\
	{\colimm(\F(F))} && {\colimm(\F(G))}
	\arrow["{{\psi_F}}"', from=1-2, to=2-1]
	\arrow["{{\psi_G}}", from=1-2, to=2-3]
	\arrow[""{name=0, anchor=center, inner sep=0}, "{{\hat{\delta}}}"', from=2-1, to=2-3]
	\arrow["C"{description}, draw=none, from=1-2, to=0]
\end{tikzcd} \tag{$\texttt{tri-$\psi$}$}\]
by induction on $\colimm{A}$. Indeed, for all $i : \Gamma_0$ and $a : A$, we have 
\[\begin{tikzcd}[column sep = 17]
	{\hat{\delta}(\psi_F(\iota_i(a)))} & {\hat{\delta}(\iota_i(\pr_2(F_i)(a)))} & {\iota_i(\pr_1(d_i)(\pr_2(F_i)(a)))} & {\iota_i(\pr_2(G_i)(a))} & {\psi_G(\iota_i(a))}
	\arrow[Rightarrow, scaling nfold=3, no head, from=1-1, to=1-2]
	\arrow[Rightarrow, scaling nfold=3, no head, from=1-2, to=1-3]
	\arrow["{C_i(a)}", Rightarrow, no head, from=1-3, to=1-4]
	\arrow[Rightarrow, scaling nfold=3, no head, from=1-4, to=1-5]
\end{tikzcd}\] where $C_i(a)$ is defined as $\ap_{\iota_i}(\pr_2(d_i)(a))$. By homotopy naturality, we have a path
$S_{i,j,g}(a)  :  \kappa_{i,j,g}(\pr_1(d_i)(\pr_2(F_i)(a)))^{-1} \cdot \ap_{\iota_j \circ \pr_1(G_{i,j,g})}(\pr_2(d_i)(a))  \cdot \kappa_{i,j,g}(\pr_2(G_i)(a)) = C_i(a)$. Hence we have a chain $\gamma_{i,j,g}(a)$ of paths
\[ 
\hspace*{-10mm}\begin{tikzcd}
	{\left(\kappa_{i,j,g}(a)\right)_{\ast}(C_j(a))} \\
	{\ap_{\hat{\delta}}(\ap_{\psi_F}(\kappa_{i,j,g}(a)))^{-1} \cdot C_j(a) \cdot \ap_{\psi_G}(\kappa_{i,j,g}(a))} \\
	{\ap_{\hat{\delta}}(\kappa_{i,j,g}(\pr_2(F_i)(a)))^{-1} \cdot \ap_{\iota_j \circ \pr_1(d_j)}(\pr_2(F_{i,j,g})(a)) \cdot C_j(a) \cdot \ap_{\psi_G}(\kappa_{i,j,g}(a))} \\
	{\left( \ap_{\iota_j}(\xi_{i,j,g}(\pr_2(F_i)(a)))^{-1} \cdot \kappa_{i,j,g}(\pr_1(d_i)(\pr_2(F_i)(a)))\right)^{-1} \cdot \ap_{\iota_j \circ \pr_1(d_j)}(\pr_2(F_{i,j,g})(a)) \cdot C_j(a) \cdot \ap_{\psi_G}(\kappa_{i,j,g}(a))} \\
	{\left( \ap_{\iota_j}(E_{i,j,g}(\pr_2(F_i)(a)))^{-1} \cdot \kappa_{i,j,g}(\pr_1(d_i)(\pr_2(F_i)(a)))\right)^{-1} \cdot \ap_{\iota_j \circ \pr_1(d_j)}(\pr_2(F_{i,j,g})(a)) \cdot C_j(a) \cdot \ap_{\psi_G}(\kappa_{i,j,g}(a))} \\
	{\kappa_{i,j,g}(\pr_1(d_i)(\pr_2(F_i)(a)))^{-1} \cdot \ap_{\iota_j \circ \pr_1(G_{i,j,g})}(\pr_2(d_i)(a))  \cdot \kappa_{i,j,g}(\pr_2(G_i)(a))} \\
	{C_i(a)}
	\arrow["{{\textit{via $\cref{transfunc}$}}}", equals, from=1-1, to=2-1]
	\arrow["{{{\textit{via $\beta_{\psi_F}(i,j,g,a)$}}}}", equals, from=2-1, to=3-1]
	\arrow["{{{\textit{via $\beta_{\hat{\delta}}(i,j,g,\pr_2(F_i)(a))$}}}}", equals, from=3-1, to=4-1]
	\arrow["{{{\textit{via $\tilde{\xi}_{i,j,g}(a)$}}}}", equals, from=4-1, to=5-1]
	\arrow["{{{\textit{via $\beta_{\psi_G}(i,j,g,a)$}}}}", equals, from=5-1, to=6-1]
	\arrow["{\textit{via $S_{i,j,g}(a)$}}", equals, from=6-1, to=7-1]
\end{tikzcd}
\] for all $g : \Gamma_1(i,j)$ and $a :A$. Thus, \eqref{tripsi} commutes, and we get a map of spans
\[\begin{tikzcd}[column sep = large, row sep = large]
    A &&[-2mm] {\colimm{A}} & {\colimm(\F(F))} \\
    A &&[-2mm] {\colimm{A}} & {\colimm(\F(G))}
    \arrow["\idd"', from=1-1, to=2-1]
    \arrow[""{name=0, anchor=center, inner sep=0}, from=1-3, to=1-1]
    \arrow[""{name=1, anchor=center, inner sep=0}, from=1-3, to=1-4]
    \arrow["\idd", from=1-3, to=2-3]
    \arrow["{{{{\hat{\delta}}}}}", from=1-4, to=2-4]
    \arrow[""{name=2, anchor=center, inner sep=0}, from=2-3, to=2-1]
    \arrow[""{name=3, anchor=center, inner sep=0}, from=2-3, to=2-4]
    \arrow["{\refl}"{description}, draw=none, from=0, to=2]
    \arrow["{{{C^{-1}}}}", draw=none, from=1, to=3]
  \end{tikzcd}\]
   This gives us
\begin{align*}
\label{colaction} \colimm^A(\delta) & \ \coloneqq \ \left( \Psi_{\delta}, \refl_{\inl({-})} \right) : \P_A(F) \to_A \P_A(G) 
\tag{$\texttt{col-act}$}
\\ \Psi_{\delta}(\inr(\iota_i(x))) & \ \equiv \ \inr(\iota_i(\pr_1(d_i)(x)))
\\ \beta_{\Psi_{\delta}}(x) & \ : \ \ap_{\Psi_{\delta}}(\glue_{\P_A(F)}(x))  = \glue_{\P_A(G)}(x) \cdot \ap_{\inr}(C^{-1}(x)) 
\end{align*}

Now, let $F : \C \to \D$ be a functor---either wild or classical. We say that $F$ \emph{creates colimits} if it both preserves and reflects co\
limiting cocones. By \emph{reflects colimits}, we mean that for any cocone $K$, if $F(K)$ is colimiting in $\D$, then $K$ is colimiting i\
n $\C$.

\begin{corollary}[{\cite[\href{https://github.com/PHart3/colimits-agda/tree/v0.5.0/Colimit-coslice/Create}{Create}]{agda-colim-TR}}] \label{create}
The forgetful functor $A/\U \to \U$ creates colimits over trees.
\end{corollary}
\begin{proof}
  Suppose that $\Gamma$ is a tree and let $F$ be an $A$-diagram over $\Gamma$. By \cref{constcol}, the function $\left[\idd_A\right] : \colimm{A} \to A$ is an equivalence. One can check that
  \[\begin{tikzcd}[column sep = large]
      {\colimm{A}} & {\colimm(\F(F))} \\
      A & {\colimm(\F(F))}
      \arrow["\psi", from=1-1, to=1-2]
      \arrow["{\left[\idd_A\right]}"', from=1-1, to=2-1]
      \arrow["{\psi \circ \left[\idd_A\right]^{-1}}"', from=2-1, to=2-2]
      \arrow["\idd", from=1-2, to=2-2]
    \end{tikzcd}\] is a pushout square in $\U$. By uniqueness of pushouts, this gives us an equivalence $\gamma : \P_A(F) \xrightarrow{\simeq} \colimm(\F(F))$ such that  $\gamma(\inr(\iota_i(x))) \equiv \iota_i(x)$ for all $i : \Gamma_0$ and $x : \pr_1(F_i)$.
  Further,
  \[
    \ap_{\gamma}(\ap_{\inr}(\kappa_{i,j,g}(x))) \ = \ \ap_{\gamma \circ \inr}(\kappa_{i,j,g}(x)) \ \equiv \ \ap_{\idd}(\kappa_{i,j,g}(x)) \ = \ \kappa_{i,j,g}(x)
  \] for all $g : \Gamma_1(i,j)$ and $x : \pr_1(F_i)$.  This means that $\gamma$ is a morphism of cocones under $\F(F)$. It follows from \cref{coluniq:p2} that the forgetful functor preserves colimits over $\Gamma$.

  It remains to prove that the forgetful functor reflects colimits over $\Gamma$. Consider an $F$-cocone $\K$:
  \[\begin{tikzcd}
      {F_i} && {F_j} \\
      & C
      \arrow[""{name=0, anchor=center, inner sep=0}, "{F_{i,j,g}}", from=1-1, to=1-3]
      \arrow["{r_i}"', from=1-1, to=2-2]
      \arrow["{r_j}", from=1-3, to=2-2]
      \arrow["{\left\langle{H,K}\right\rangle}"{description}, draw=none, from=0, to=2-2]
    \end{tikzcd}\]
 as well as the cocone $\F(\K) \coloneqq \left(\pr_1(C), \pr_1 \circ r, H\right)$ under $\F(F)$ obtained by applying the forgetful functor to $\K$. Suppose that $\F(\K)$ is colimiting in $\U$. By the universal property of colimits in $A/\U$, we have a morphism $\tau : \left(\P_A(F), \inl\right) \to C$ of cocones, which induces a morphism $\F(\tau) : \P_A(F) \to \pr_1(C)$ of cocones in $\U$. By \cref{coluniq:p1}, $\F(\tau)$ is an isomorphism. Thus, $\tau$ is a cocone isomorphism, so that $\K$ is colimiting by \cref{coluniq:p2}.
\end{proof}

\noindent We pose the converse to \cref{create} as the following question.

\begin{question}
Let $\Delta$ be a graph and $G$ be an $A$-diagram over $\Delta$. If the canonical function $\colimm(\F(G)) \to \pr_1(\colimm^A(G))$ is an equivalence, then is $\Delta$ a tree?
\end{question}

\begin{corollary}
If $\Gamma$ is a tree, then for each $X : A/\U$, the colimit $\colimm^A$ of the constant diagram at $X$ is the canonical cocone on $X$.
\end{corollary}
\begin{proof} 
By \cref{create,constcol}.
\end{proof}

\begin{note}
By \cref{leftstable}, we can refine \cref{create} as follows. If $ \left\lvert{\Gamma}\right\rvert$ is $n$-connected, then so is the function $\colimm(\F(F)) \xrightarrow{\inr} \P_A(F)$ by virtue of the commuting triangle
\[\begin{tikzcd}
        {\colimm{A}} && {A \times \left\lvert{\Gamma}\right\rvert } \\
	& A
        \arrow["\simeq", from=1-1, to=1-3]
        \arrow["{\mleft[\idd_A\mright]}"', from=1-1, to=2-2]
        \arrow["{\pr_1}", from=1-3, to=2-2]
\end{tikzcd}\] In this way, the degree to which $\F$ approximates $\colimm^A(F)$ increases linearly with how close $\Gamma$ is to a tree. 
\end{note} 

\subsubsection*{Adjunction with the constant diagram functor}

Next, we verify that our action on maps \eqref{colaction} is correct by showing that the resulting $0$-functor $\colimm^A$ is left adjoint (in the sense of \cref{adjdef}) to the constant diagram functor.  Consider again a morphism $\delta \coloneqq \left(d, \left\langle{\xi, \tilde{\xi}}\right\rangle\right) : F \Rightarrow_A G$ of $A$-diagrams.

\begin{note} \label{cancocmap}
We have the cocone $\K(\delta) \coloneqq \left(c_1, c_2\right)$ under $F$ with tip $\P_A(G)$ where
 \begin{align*}
& c_1(i)  \ \coloneqq \ \left(\inr \circ \iota_i \circ \pr_1(d_i) , \lambda{a}.\ap_{\inr \circ \iota_i}(\pr_2(d_i)(a)) \cdot \tau^G_i(a)\right)
\\ & c_2(i,j,g)  \ \coloneqq \ \left(\lambda{x}. \ap_{\inr \circ \iota_j}(\xi_{i,j,g}(x))^{-1} \cdot \ap_{\inr}(\kappa_{i,j,g}^G(\pr_1(d_i)(x))) , \lambda{a}.\Theta(\epsilon_{i,j,g}(a), \tilde{\xi}_{i,j,g}(a)) \right) 
\end{align*} where $\Theta(\epsilon_{i,j,g}(a), \tilde{\xi}_{i,j,g}(a))$ is the chain of paths
\[
\hspace*{-2.1cm}
\begin{tikzcd}
	{\left(\ap_{\inr \circ \iota_j}(\xi_{i,j,g}(\pr_2(F_i)(a)))^{-1} \cdot \ap_{\inr}( \kappa_{i,j,g}^G(\pr_1(d_i)(\pr_2(F_i)(a))))\right)^{-1} \cdot \ap_{\inr \circ \iota_j \circ \pr_1(d_j)}(\pr_2(F_{i,j,g})(a)) \cdot \ap_{\inr \circ \iota_j}(\pr_2(d_j)(a)) \cdot \tau^G_j(a)} \\
	{\left(\ap_{\inr \circ \iota_j}(E_{i,j,g}(a))^{-1} \cdot \ap_{\inr}( \kappa_{i,j,g}^G(\pr_1(d_i)(\pr_2(F_i)(a))))\right)^{-1} \cdot \ap_{\inr \circ \iota_j \circ \pr_1(d_j)}(\pr_2(F_{i,j,g})(a)) \cdot \ap_{\inr \circ \iota_j}(\pr_2(d_j)(a)) \cdot \tau^G_j(a)} \\
	{\ap_{\inr \circ \iota_i}(\pr_2(d_i)(a)) \cdot \ap_{\inr}(\kappa^G_{i,j,g}(\pr_2(G_i)(a)))^{-1} \cdot \ap_{\inr \circ \iota_j}(\pr_2(G_{i,j,g})(a)) \cdot \tau^G_j(a)} \\
	{\ap_{\inr \circ \iota_i}(\pr_2(d_i)(a)) \cdot \tau_i^G(a)}
	\arrow["{{{\textit{via $\tilde{\xi}_{i,j,g}(a)$}}}}", equals, from=1-1, to=2-1]
	\arrow["{{{{\textit{via homotopy naturality of $\kappa_{i,j,g}^G$ at $\pr_2(d_i)(a)$}}}}}", equals, from=2-1, to=3-1]
	\arrow["{{\textit{via $\epsilon_{i,j,g}(a)$}}}", equals, from=3-1, to=4-1]
\end{tikzcd} 
 \] 
We claim that our action on morphisms equals the cogap map of $\K(\delta)$, i.e.,
\[
\label{delteq}  \colimm(\delta) = \pstc^{-1}_{F, \colimm(G)}(\K(\delta))
\tag{$\texttt{map-eq}$} \] This goal amounts to showing that $\colimm(\delta)$ belongs to the fiber of $\pstc_{F, \colimm(G)}$ over $\K(\delta)$. The proof closely resembles the first half of the proof of \cref{mainequiv}. We again leave it to the Agda formalization~\cite[\href{https://github.com/PHart3/colimits-agda/blob/v0.5.0/Colimit-coslice/Map-Nat/CosColimitMap16.agda\#L155}{fib-inhab}]{agda-colim-TR}.
\end{note} 

\noindent The cocone $K(\delta)$ from \cref{cancocmap} is an instance of the following more general operation.

\begin{definition}
For each $T : A/\U$, define $\mathsf{Cocone}^T({-} \circ \delta) : \mathsf{Cocone}_G(T) \to \mathsf{Cocone}_F(T)$ by pre-composing $\delta$ with a given cocone $K$ under $G$ like so:
\[\begin{tikzcd}
	{F_i} && {F_j} \\
	{G_i} && {G_j} \\
	& T
	\arrow[""{name=0, anchor=center, inner sep=0}, from=1-1, to=1-3]
	\arrow[from=1-1, to=2-1]
	\arrow[from=1-3, to=2-3]
	\arrow[""{name=1, anchor=center, inner sep=0}, dashed, from=2-1, to=2-3]
	\arrow[from=2-1, to=3-2]
	\arrow[from=2-3, to=3-2]
	\arrow["\delta"{description}, Rightarrow, from=0, to=1, shorten =2mm]
	\arrow["K"{description}, draw=none, from=1, to=3-2]
\end{tikzcd}\]
\end{definition}

\begin{lemma}[{\cite[\href{https://github.com/PHart3/colimits-agda/blob/v0.5.0/Colimit-coslice/Map-Nat/CosColimitPreCmp.agda}{CosColimitPreCmp}]{agda-colim-TR}}] \label{precompnatsq}
The following square commutes:
\[\begin{tikzcd}
	{\left(\colimm(G) \to_A T \right)} && {\left(\colimm(F) \to_A T\right)} \\
	{ \mathsf{Cocone}_G(T)} && { \mathsf{Cocone}_F(T)}
	\arrow["{{\pstc_{G,T}}}"', from=1-1, to=2-1]
	\arrow["{\mathsf{Cocone}^T({-} \circ \delta)}"', from=2-1, to=2-3]
	\arrow["{{{-} \circ \colimm^A(\delta)}}", from=1-1, to=1-3]
	\arrow["{{\pstc_{F,T}}}", from=1-3, to=2-3]
\end{tikzcd}\]
\end{lemma} 
\begin{proof}
For each $f^{\ast} : \colimm(G) \to_A T$, note that 
\begin{align*}
& \ \pstc_{F,T}(f^{\ast} \circ \pstc^{-1}_{F, \colimm(G)}(\K(\delta)))  
\\ = &  \ \mathsf{Cocone}_F(f^{\ast} \circ {-})(\pstc_{F, \colimm(G)}(\pstc^{-1}_{F, \colimm(G)}(\K(\delta)))) \tag{\cref{pstcompnatsq}}
\\ = & \ \mathsf{Cocone}_F(f^{\ast} \circ {-})(\K(\delta)).
\end{align*}
By \eqref{delteq}, it thus suffices to prove that 
$\mathsf{Cocone}_F(f^{\ast} \circ {-})(\K(\delta)) = \mathsf{Cocone}^T({-} \circ \delta)(\pstc_{G,T}(f^{\ast}))$. We leave such a proof, which is messy yet routine, to the Agda formalization~\cite[\href{https://github.com/PHart3/colimits-agda/blob/v0.5.0/Colimit-coslice/Map-Nat/CosColimitMap18.agda\#L207}{CosColim-NatSq2}]{agda-colim-TR}.
\end{proof}

\begin{corollary}[{\cite[\href{https://github.com/PHart3/colimits-agda/blob/v0.5.0/Colimit-coslice/Main-Theorem/CosColim-Adjunction.agda}{CosColim-Adjunction}]{agda-colim-TR}}]\label{colimadj}
We have an adjunction $\colimm^A \dashv \const_{\Gamma}$, where 
$\const_{\Gamma}$ denotes the constant diagram functor $A/\U \to \mathsf{Diag}(\Gamma, A/\U)$.
\end{corollary}
\begin{proof}
Combine \cref{mainequiv,pstcompnatsq,precompnatsq}. 
\end{proof}

\subsection{Pushout-coproduct construction of coslice colimits}\label{Contr2}

In this section, we apply the $3 \times 3 $ lemma to our first construction of $\colimm^A(F)$ to obtain the familiar construction of $\colimm^A(F)$ as a pushout of coproducts in $A/\U$.

To begin, consider the following grid of commuting squares:
\[ 
\hspace*{-2.7cm}
\adjustbox{scale=.81} {
\begin{tikzcd}[column sep = 38, row sep= 45, /tikz/column 1/.append style={column sep=23pt}]
	{\sum_{\left(i,j,g\right) : \sum_{\left(i,j\right) : \Gamma_0 \times \Gamma_0}\Gamma_1(i,j)}\pr_1(F_i)} && {\left( \sum_{\left(i,j,g\right) : \sum_{\left(i,j\right) : \Gamma_0 \times \Gamma_0}\Gamma_1(i,j)}\pr_1(F_i)\right) + \left( \sum_{\left(i,j,g\right) : \sum_{\left(i,j\right) : \Gamma_0 \times \Gamma_0}\Gamma_1(i,j)}\pr_1(F_i) \right)} && {\sum_{i : \Gamma_0}\pr_1(F_i)} \\
	{\left(\sum_{\left(i,j\right) : \Gamma_0 \times \Gamma_0}\Gamma_1(i,j) \right)\times A} && {\left( \left(\sum_{\left(i,j\right) : \Gamma_0 \times \Gamma_0}\Gamma_1(i,j) \right)\times A\right)  + \left( \left(\sum_{\left(i,j\right) : \Gamma_0 \times \Gamma_0}\Gamma_1(i,j) \right)\times A\right) } && {\Gamma_0 \times A} \\
	A && A && A
	\arrow["{{{{{\idd +   \idd}}}}}"', from=1-3, to=1-1]
	\arrow["{{{{{\left(i, x\right)  +  \left(j, \pr_1(F_{i,j,g})(x)\right)}}}}}", from=1-3, to=1-5]
	\arrow[""{name=0, anchor=center, inner sep=0}, "{{{{{\left(i,j,g,\pr_2(F_i)(a)\right)}}}}}"{description}, from=2-1, to=1-1]
	\arrow[""{name=1, anchor=center, inner sep=0}, "{{{{\pr_2}}}}"{description}, from=2-1, to=3-1]
	\arrow[""{name=2, anchor=center, inner sep=0}, "{{{{{\left(i,j,g,\pr_2(F_i)(a)\right)  +  \left(i,j,g,\pr_2(F_i)(a)\right) }}}}}"{description}, from=2-3, to=1-3]
	\arrow["{{{{{\idd +   \idd}}}}}"{description}, from=2-3, to=2-1]
	\arrow["{{{{{\left(i,a\right)  +   \left(j,a\right)}}}}}"{description}, from=2-3, to=2-5]
	\arrow[""{name=3, anchor=center, inner sep=0}, "{{{{\pr_2  +  \pr_2}}}}"{description}, from=2-3, to=3-3]
	\arrow[""{name=4, anchor=center, inner sep=0}, "{{{{{\left(i, \pr_2(F_i)(a)\right)}}}}}"{description}, from=2-5, to=1-5]
	\arrow[""{name=5, anchor=center, inner sep=0}, "{{{{\pr_2}}}}"{description}, from=2-5, to=3-5]
	\arrow["{{{{\idd_A}}}}", from=3-3, to=3-1]
	\arrow["{{{{\idd_A}}}}"', from=3-3, to=3-5]
	\arrow["{{{{{{\refl_{\left(i,j,g, \pr_2(F_i)(a)\right)}  +  \refl_{\left(i,j,g, \pr_2(F_i)(a)\right)}}}}}}}"{description}, draw=none, from=2, to=0]
	\arrow["{{{{{{\refl_{\left(i, \pr_2(F_i)(a)\right)}  + \ap_{\left(j, {-}\right)}(\pr_2(F_{i,j,g})(a))}}}}}}"{description, pos=0.57}, draw=none, from=2, to=4]
	\arrow["{{{{\refl_a  +  \refl_a}}}}"{description}, draw=none, from=3, to=1]
	\arrow["{{{{\refl_a  +  \refl_a}}}}"{description}, draw=none, from=3, to=5]
\end{tikzcd}}  \]
Call the pushouts of the left, middle, and right vertical spans $V_1$, $V_2$, and $V_3$, respectively. Call the pushouts of the top, middle, and bottom horizontal spans $H_1$, $H_2$, and $H_3$, respectively. We form two additional pushouts
\[\begin{tikzcd}
	{V_2} & {V_3} && {H_2} & {H_1} \\
	{V_1} & {P_V} && {H_3} & {P_H}
	\arrow["{\eta_2}"', from=1-4, to=2-4]
	\arrow["{\eta_1}", from=1-4, to=1-5]
	\arrow[from=1-5, to=2-5]
	\arrow[from=2-4, to=2-5]
	\arrow["\lrcorner"{anchor=center, pos=0.125, rotate=180}, draw=none, from=2-5, to=1-4]
	\arrow["{\delta_1}"', from=1-1, to=2-1]
	\arrow["{\delta_2}", from=1-1, to=1-2]
	\arrow[from=2-1, to=2-2]
	\arrow[from=1-2, to=2-2]
	\arrow["\lrcorner"{anchor=center, pos=0.125, rotate=180}, draw=none, from=2-2, to=1-1]
\end{tikzcd} \]
where  
\begin{itemize}
\item $\delta_1$ denotes the function induced by the middle-to-left map of spans
\item  $\delta_2$ the function induced by the middle-to-right map of spans
\item $\eta_1$ the function induced by the middle-to-top map of spans
\item $\eta_2$ the function induced by the middle-to-bottom map of spans.
\end{itemize}
Licata and Brunerie construct an equivalence $\tau_1 : P_H \xrightarrow{\simeq} P_V$ of types by double induction on pushouts~\cite[Section VII]{3x3}, which in particular satisfies
\begin{align*}
\tau_1(\inl(\inl(a))) & \  \equiv \  \inl(\inl(a))
\\ \tau_1(\inr(\inr(i,x))) & \ \equiv \ \inr(\inr(i,x)) .
\end{align*}

\begin{lemma}
We have an equivalence
\[
\xi \ : \ V_2 \xrightarrow{\simeq} \left(\bigvee_{i,j,g}\pr_1(F_i) \right) \vee \left(\bigvee_{i,j,g}\pr_1(F_i)\right)
\] 
\end{lemma}
\begin{proof}
Letting 
$W_1 \coloneqq \left( \left( \sum_{\left(i,j\right) : \Gamma_0 \times \Gamma_0}\Gamma_1(i,j) \right) \times A\right) + \left(\left( \sum_{\left(i,j\right) : \Gamma_0 \times \Gamma_0}\Gamma_1(i,j) \right) \times A\right)$ and 
$W_2 \coloneqq \left( \sum_{\left(i,j,g\right) : \sum_{\left(i,j\right) : \Gamma_0 \times \Gamma_0}\Gamma_1(i,j)}\pr_1(F_i)\right) + \left( \sum_{\left(i,j,g\right) : \sum_{\left(i,j\right) : \Gamma_0 \times \Gamma_0}\Gamma_1(i,j)}\pr_1(F_i) \right)$, define $\xi$ by pushout recursion on the square
\[\begin{tikzcd}
	{W_1} && {W_2} \\
	A && {\left( \bigvee_{i,j,g} \pr_1(F_i)\right) \vee \left( \bigvee_{i,j,g} \pr_1(F_i)\right)}
	\arrow[from=1-1, to=1-3]
	\arrow[""{name=0, anchor=center, inner sep=0}, "{{\inl \circ \inr +   \inr \circ \inr}}", from=1-3, to=2-3]
	\arrow[""{name=1, anchor=center, inner sep=0}, from=1-1, to=2-1]
	\arrow["{{a \mapsto \inl(\inl(a))}}"', from=2-1, to=2-3]
\end{tikzcd} 
 \] that commutes via the path
 \[
 \ap_{\inl}(\glue_{\bigvee_{i,j,g}\pr_1(F_i)})(i,j,g,a) \ + \  \glue_{\bigvee \vee \bigvee}(a) \cdot \ap_{\inr}(\glue_{\bigvee_{i,j,g}\pr_1(F_i)})(i,j,g,a)
 \] for all $g : \Gamma_1(i,j)$ and $a : A$.
Define an inverse $\tilde{\xi}$ of $\xi$ by recursion on $\bigvee \vee \bigvee$ with the commuting square
\[\begin{tikzcd}
	A & {\bigvee_{i,j,g}\pr_1(F_i)} \\
	{\bigvee_{i,j,g}\pr_1(F_i)} & {V_2}
	\arrow["{\epsilon_2}", from=1-2, to=2-2] 
	\arrow[""{name=0, anchor=center, inner sep=0}, from=1-1, to=1-2]
	\arrow[""{name=1, anchor=center, inner sep=0}, "{\epsilon_1}"', from=2-1, to=2-2]
	\arrow[from=1-1, to=2-1] 
	\arrow["{\refl_{\inl(a)}}"{description}, draw=none, from=0, to=1]
\end{tikzcd} \]
Here, $\epsilon_1$ and $\epsilon_2$ are defined, respectively, by pushout recursion on the commuting squares
\[\begin{tikzcd}[ /tikz/row 3/.append style={row sep=10pt}]
	{ \left(\sum_{\left(i,j\right) : \Gamma_0 \times \Gamma_0}\Gamma_1(i,j)\right) \times A} && {\sum_{\left(i,j,g\right) : \sum_{\left(i,j\right) : \Gamma_0 \times \Gamma_0}\Gamma_1(i,j)}\pr_1(F_i)} \\
	\\
	A && {V_2} \\
	{ \left( \sum_{\left(i,j\right) : \Gamma_0 \times \Gamma_0}\Gamma_1(i,j)\right) \times A} && {\sum_{\left(i,j,g\right) : \sum_{\left(i,j\right) : \Gamma_0 \times \Gamma_0}\Gamma_1(i,j)}\pr_1(F_i)} \\
	\\
	A && {V_2}
	\arrow[""{name=0, anchor=center, inner sep=0}, "{\inl}"', from=3-1, to=3-3]
	\arrow[from=1-1, to=3-1]
	\arrow[""{name=1, anchor=center, inner sep=0}, from=1-1, to=1-3]
	\arrow["{\inr \circ \inl}", from=1-3, to=3-3]
	\arrow[from=4-1, to=6-1]
	\arrow[""{name=2, anchor=center, inner sep=0}, from=4-1, to=4-3]
	\arrow[""{name=3, anchor=center, inner sep=0}, "{\inl}"', from=6-1, to=6-3]
	\arrow["{\inr \circ \inr}", from=4-3, to=6-3]
	\arrow["{\glue_{V_2}(\inl(i,j,g,a))}"{description}, draw=none, from=1, to=0]
	\arrow["{ \glue_{V_2}(\inr(i,j,g,a))}"{description}, draw=none, from=2, to=3]
\end{tikzcd}\] 

We prove that $\tilde{\xi} \circ \xi \sim \idd_{V_2}$ by \cref{funcpo}. We first see that 
\begin{align*}
& \tilde{\xi}(\xi(\inl(a)))  \ \equiv \ \tilde{\xi}(\inl(\inl(a))) \ \equiv \ \epsilon_1(\inl(a)) \ \equiv \ \inl(a)
\\ & \tilde{\xi}(\xi(\inr(\inl(i,j,g,x)))) \ \equiv \ \tilde{\xi}(\inl(\inr(i,j,g,x))) \ \equiv \ \epsilon_1(\inr(i,j,g,x)) \ \equiv \ \inr(\inl(i,j,g,x))
\\ & \tilde{\xi}(\xi(\inr(\inr(i,j,g,x)))) \ \equiv \ \tilde{\xi}(\inr(\inr(i,j,g,x))) \ \equiv \ \epsilon_2(\inr(i,j,g,x))
 \ \equiv \ \inr(\inr(i,j,g,x))
\end{align*}
Then by the path $\beta$-rules for $\tilde{\xi}$ and $\xi$, we easily have that $\ap_{\tilde{\xi}}(\ap_{\xi}(\glue(\inl(i,j,g,a)))) = \glue(\inl(i,j,g,a))$ and that
$\ap_{\tilde{\xi}}(\ap_{\xi}(\glue(\inr(i,j,g,a)))) = \glue(\inr(i,j,g,a))$ for all $g : \Gamma_1(i,j)$ and $a : A$. 

Next, we prove that $\xi \circ \tilde{\xi} \sim \idd_{\bigvee \vee \bigvee}$ by \cref{funcpo} again. We first see that
\begin{align*}
& \xi(\tilde{\xi}(\inl(\inr(i,j,g,x)))) \ \equiv  \ \xi(\inr(\inl(i,j,g,x))) \ \equiv \ \inl(\inr(i,j,g,x))
\\ & \xi(\tilde{\xi}(\inl(\inl(a)))) \ \equiv \ \xi(\inl(a)) \ \equiv \ \inl(\inl(a))
\\ & \xi(\tilde{\xi}(\inr(\inr(i,j,g,x))))  \ \equiv \ \xi(\inr(\inr(i,j,g,x))) \ \equiv \ \inr(\inr(i,j,g,x))
\\ & \xi(\tilde{\xi}(\inr(\inl(a)))) \ \equiv \ \xi(\inl(a)) \ \equiv \ \inl(\inl(a)) \ = \ \inr(\inl(a)) \tag{$\glue_{\bigvee \vee \bigvee}(a)$}
\end{align*}
On $\inl$, the path $\beta$-rules for $\epsilon_1$ and $\xi$ imply that $ \ap_{\xi}(\ap_{\tilde{\xi} \circ \inl}(\glue(i,j,g,a))) = \ap_{\inl}(\glue(i,j,g,a)) $ for all $g : \Gamma_1(i,j)$ and $a : A$. On $\inr$, the $\beta$-rules for $\epsilon_2$ and $\xi$ easily imply that $\ap_{\xi}(\ap_{\tilde{\xi} \circ \inr}(\glue(i,j,g,a))) = \glue(a) \cdot \ap_{\inr}(\glue(i,j,g,a))$. Finally, $\ap_{\xi}(\ap_{\tilde{\xi}}(\glue(a))) = \refl_{\inl(\inl(a))}$ for all $a : A$.
\end{proof}

Now, define $\sigma : \left(\bigvee_{i,j,g}\pr_1(F_i) \right) \vee \left(\bigvee_{i,j,g}\pr_1(F_i)\right) \to \bigvee_i\pr_1(F_i)$ as the cogap map for
\[\begin{tikzcd}
	A & {\bigvee_{i,j,g}\pr_1(F_i)} \\
	{\bigvee_{i,j,g}\pr_1(F_i)} & {\bigvee_i\pr_1(F_i)}
	\arrow["{\alpha_1}"', from=2-1, to=2-2]
	\arrow[""{name=0, anchor=center, inner sep=0}, "{\alpha_2}", from=1-2, to=2-2]
	\arrow[""{name=1, anchor=center, inner sep=0}, from=1-1, to=2-1]
	\arrow[from=1-1, to=1-2]
	\arrow["{\refl_{\inl(a)}}"{description}, draw=none, from=1, to=0]
\end{tikzcd}\]
Here, $\alpha_1$ and $\alpha_2$ are defined, respectively, by pushout recursion on the commuting squares
\[\begin{tikzcd}[  /tikz/row 3/.append style={row sep=10pt}]
	{\left(\sum_{\left(i,j\right) : \Gamma_0 \times \Gamma_0}\Gamma_1(i,j)\right) \times A} && {\sum_{\left(i,j,g\right) : \sum_{\left(i,j\right) : \Gamma_0 \times \Gamma_0}\Gamma_1(i,j)}\pr_1(F_i)} \\
	\\
	A && {\bigvee_i\pr_1(F_i)} \\
	{\left(\sum_{\left(i,j\right) : \Gamma_0 \times \Gamma_0}\Gamma_1(i,j)\right) \times A} && {\sum_{\left(i,j,g\right) : \sum_{\left(i,j\right) : \Gamma_0 \times \Gamma_0}\Gamma_1(i,j)}\pr_1(F_i)} \\
	\\
	A && {\bigvee_i\pr_1(F_i)}
	\arrow[""{name=0, anchor=center, inner sep=0}, from=1-1, to=3-1]
	\arrow[from=1-1, to=1-3]
	\arrow[""{name=1, anchor=center, inner sep=0}, "{\inr(i,x)}", from=1-3, to=3-3]
	\arrow["{\inl}"', from=3-1, to=3-3]
	\arrow[from=4-1, to=4-3]
	\arrow[""{name=2, anchor=center, inner sep=0}, from=4-1, to=6-1]
	\arrow["{\inl}"', from=6-1, to=6-3]
	\arrow[""{name=3, anchor=center, inner sep=0}, "{\inr(j, \pr_1(F_{i,j,g})(x))}", from=4-3, to=6-3]
	\arrow["{\glue_{\bigvee_i\pr_1(F_i)}(i,a)}"{description}, draw=none, from=0, to=1]
	\arrow["{\glue_{\bigvee_i\pr_1(F_i)}(j,a) \cdot \ap_{\inr(j,{-})}(\pr_2(F_{i,j,g})(a))^{-1}}"{description}, draw=none, from=2, to=3]
\end{tikzcd} \] 
 We have a map of spans:
\[\begin{tikzcd}
	{V_1} && {V_2} && {V_3} \\
	{\bigvee_{i,j,g}\pr_1(F_i)} && {\left(\bigvee_{i,j,g}\pr_1(F_i) \right) \vee \left(\bigvee_{i,j,g}\pr_1(F_i)\right)} && {\bigvee_i\pr_1(F_i)}
	\arrow["\idd"', from=1-1, to=2-1]
	\arrow["{{\delta_1}}"', from=1-3, to=1-1]
	\arrow["{{\delta_2}}", from=1-3, to=1-5]
	\arrow["\xi", from=1-3, to=2-3]
	\arrow["\simeq"', from=1-3, to=2-3]
	\arrow["\idd", from=1-5, to=2-5]
	\arrow["{{\idd \vee \idd}}", from=2-3, to=2-1]
	\arrow["\sigma"', from=2-3, to=2-5]
\end{tikzcd}\] 
\begin{notation}
Denote the pushout of the lower span by $\PW(F)$ (for ``pushout of wedges'').
\end{notation}
\noindent Indeed, we use \cref{funcpo} on $V_2$ to get two homotopies making these two subsquares commute. For the left sub-square, we first see that
\begin{align*}
& \left(\idd \vee \idd\right)(\xi(\inl(a))) \ \equiv \ \inl(a) \ \equiv \ \delta_1(\inl(a))
\\ & \left(\idd \vee \idd\right)(\xi(\inr(\inl(i,j,g,x)))) \ \equiv \  \inr(i,j,g,x) \  \equiv \  \delta_1(\inr(\inl(i,j,g,x)))
\\ & \left(\idd \vee \idd\right)(\xi(\inr(\inr(i,j,g,x)))) \ \equiv \ \inr(i,j,g,x) \ \equiv \ \delta_1(\inr(\inr(i,j,g,x)))
\end{align*}
To complete the homotopy, we easily check that $\ap_{\idd \vee \idd}(\ap_{\xi}(\glue(\inl(i,j,g,a)))) = \ap_{\delta_1}(\glue(\inl(i,j,g,a)))$ and that $\ap_{\idd \vee \idd}(\ap_{\xi}(\glue(\inr(i,j,g,a)))) = \ap_{\delta_1}(\glue(\inr(i,j,g,a)))$.
For the right sub-square, we define the homotopy in the same way.
We now have an isomorphism of spans, which induces an equivalence of pushouts $\tau_2 : P_V \xrightarrow{\simeq} \PW(F)$.

\begin{lemma}
We have an equivalence between $\psi$ and $\eta_1$:
\[ \label{cofmap}
\begin{tikzcd}
	{\colimm{A}} & {\colimm(\F(F))} \\
	{H_2} & {H_1}
	\arrow["\psi", from=1-1, to=1-2]
	\arrow["\simeq", from=1-1, to=2-1]
	\arrow["{{{w_0}}}"', from=1-1, to=2-1]
	\arrow["\simeq"', from=1-2, to=2-2]
	\arrow["{{{w_1}}}", from=1-2, to=2-2]
	\arrow["{\eta_1}"', from=2-1, to=2-2]
\end{tikzcd} \tag{$\texttt{$\psi$-$\eta_1$-sq}$}
\]
\end{lemma}
\begin{proof}
Define $w_0$ and $w_1$ by the following cocones under $A$ and $\F(F)$, respectively:
\[
\begin{tikzcd}[/tikz/row 1/.append style={row sep=-2pt}, /tikz/row 2/.append style={row sep=-2pt},
/tikz/row 4/.append style={row sep=-2pt}, /tikz/row 5/.append style={row sep=-2pt}]
	A && A \\
	&&& {\left(\lambda{a}.\glue_{H_2}(\inr(i,j,g,a))^{-1} \cdot \glue_{H_2}(\inl(i,j,g,a))\right)} \\
	& {H_2} \\
	{\pr_1(F_i)} && {\pr_1(F_j)} \\
	&&& {\left(\lambda{x}.\glue_{H_1}(\inr(i,j,g,x))^{-1} \cdot \glue_{H_1}(\inl(i,j,g,x))\right)} \\
	& {H_1}
	\arrow["{{{\idd_A}}}", from=1-1, to=1-3]
	\arrow["{{{\inr(i,{-})}}}"', from=1-1, to=3-2]
	\arrow["{{{\inr(j,{-})}}}", from=1-3, to=3-2]
	\arrow["{{{\pr_1(F_{i,j,g})}}}", from=4-1, to=4-3]
	\arrow["{{{\inr(i,{-})}}}"', from=4-1, to=6-2]
	\arrow["{{{\inr(j,{-})}}}", from=4-3, to=6-2]
\end{tikzcd}\]
We prove that $\eqref{cofmap}$ commutes by \cref{colimmapeq}: We have that $\eta_1(w_0(\iota_i(a)))  \equiv   \eta_1(\inr(i,a)) \equiv  \inr(i,\pr_2(F_i)(a)) \equiv   w_1(\iota_i(\pr_2(F_i)(a))) \equiv  w_1(\psi(\iota_i(a)))$. We also see that $\ap_{\eta_1}(\ap_{w_0}(\kappa_{i,j,g}(a))) = \ap_{w_1}(\ap_{\psi}(\kappa_{i,j,g}(a)))$ by the relevant path $\beta$-rules.

To see that $w_0$ and $w_1$ are equivalences, define inverses $y_0$ and $y_1$ of $w_0$ and $w_1$, respectively, by recursion on puhsouts:
\[
\begin{tikzcd}[column sep = large]
	{W_1} && {\Gamma_0 \times A} & \\
	{\left( \sum_{\left(i,j\right) : \Gamma_0 \times \Gamma_0}\Gamma_1(i,j) \right) \times A} && {H_2} \\
	&&& {\colimm{A}} \\
	{W_2} && {\sum_{i : \Gamma_0}\pr_1(F_i)} \\
	{\sum_{\left(i,j,g\right) : \sum_{\left(i,j\right) : \Gamma_0 \times \Gamma_0}\Gamma_1(i,j)}\pr_1(F_i)} && {H_1} \\
	&&& {\colimm(\F(F))}
	\arrow[from=1-1, to=1-3]
	\arrow[""{name=0, anchor=center, inner sep=0}, from=1-1, to=2-1]
	\arrow[""{name=1, anchor=center, inner sep=0}, "{{{\left(i,a\right) \mapsto \iota_i(a)}}}", curve={height=-12pt}, from=1-3, to=3-4]
	\arrow["{{{\left(i,j,g,a\right) \mapsto \iota_j(a)}}}"', curve={height=12pt}, from=2-1, to=3-4]
	\arrow["{{{y_0}}}", dashed, from=2-3, to=3-4]
	\arrow[from=4-1, to=4-3]
	\arrow[""{name=2, anchor=center, inner sep=0}, from=4-1, to=5-1]
	\arrow[""{name=3, anchor=center, inner sep=0}, "{{{\left(i,x\right) \mapsto \iota_i(x)}}}", curve={height=-12pt}, from=4-3, to=6-4]
	\arrow["{{{\left(i,j,g,x\right) \mapsto \iota_j(\pr_1(F_{i,j,g})(x))}}}"', curve={height=12pt}, from=5-1, to=6-4]
	\arrow["{{{y_1}}}", dashed, from=5-3, to=6-4]
	\arrow["{{{\kappa_{i,j,g}(a) \ + \ \refl_{\iota_j(a)}}}}"{description}, draw=none, from=0, to=1]
	\arrow["{{{\kappa_{i,j,g}(x) \ + \ \refl_{\iota_j(F_{i,j,g}(x))}}}}"{description}, draw=none, from=2, to=3]
\end{tikzcd} \]
By \cref{funcpo,colimmapeq}, it is routine to check  $y_0$ and $y_1$ are inverses to $w_0$ and $w_1$, respectively.
\end{proof}

Note that \eqref{cofmap} fits into an isomorphism of spans
\[
\begin{tikzcd}[column sep = large]
	A &[6mm] {\colimm{A}} & {\colimm(\F(F))} \\
	{H_3} &[6mm] {H_2} & {H_1}
	\arrow["\psi", from=1-2, to=1-3]
	\arrow["{{{\left[\idd_A\right]}}}"', from=1-2, to=1-1]
	\arrow["\simeq", from=1-1, to=2-1]
	\arrow[""{name=0, anchor=center, inner sep=0}, "{{{\inl}}}"', from=1-1, to=2-1]
	\arrow["\simeq", from=1-2, to=2-2]
	\arrow[""{name=1, anchor=center, inner sep=0}, "{{{{{w_0}}}}}"', from=1-2, to=2-2]
	\arrow["\simeq"', from=1-3, to=2-3]
	\arrow["{{{{{w_1}}}}}", from=1-3, to=2-3]
	\arrow["{{{\eta_1}}}"', from=2-2, to=2-3]
	\arrow["{{{\eta_2}}}", from=2-2, to=2-1]
	\arrow["{\mathsf{lhs}}"{description}, draw=none, from=0, to=1]
\end{tikzcd}  \]  Here, the homotopy $\mathsf{lhs}$ is defined by \cref{colimmapeq}: We see that $\eta_2(w_0(\iota_i(a))) \equiv  \eta_2(\inr(i,a))  \equiv   \inr(a)  =  \inl(a)  \equiv  \inl(\left[\idd_A\right](\iota_i(a)))$ by $\glue_{H_3}(a)^{-1}$, and it's easy to check that $ \ap_{\eta_2}(\ap_{w_0}(\kappa_{i,j,g}(a))) \cdot \glue(a)^{-1} = \glue(a)^{-1} \cdot \ap_{\inl}(\ap_{\left[\idd_A\right]}(\kappa_{i,j,g}(a)))$.
The pushout of the upper span here is exactly $\P_A(F)$, so the pushout action on maps yields an equivalence $\tau_0 : \colimm^A(F) \xrightarrow{\simeq} P_H$.

\begin{corollary}\label{coeqcop}
We have an equivalence $T_F : \colimm^A(F) \xrightarrow{\simeq} \PW(F)$ such that 
$T_F(\inl(a))  \equiv \inl(\inl(a))$ and
$T_F(\inr(\iota_i(x))) \equiv \inr(\inr(i,x))$
\end{corollary}
\begin{proof}
Define $T_F \coloneqq \tau_2 \circ \tau_1 \circ \tau_0$.
\end{proof}

\section{Universality of colimits}\label{App1}

Let $\U$ be a universe and $A$ be a type. Let $\Gamma$ be a graph and $F$ be an $A$-diagram over $\Gamma$. We say that $\colimm^A(F)$ is \emph{universal}, or \emph{pullback-stable}, if 
for every pullback square
\[ \label{pbs}
\begin{tikzcd}
	{\colimm^A(F) \times_V Y} & Y \\
	{\colimm^A(F)} & V
	\arrow["{\pi_2}", from=1-1, to=1-2]
	\arrow["{\pi_1}"', from=1-1, to=2-1]
	\arrow["\lrcorner"{anchor=center, pos=0.125}, draw=none, from=1-1, to=2-2]
	\arrow["h", from=1-2, to=2-2]
	\arrow["f"', from=2-1, to=2-2]
\end{tikzcd} \tag{$\texttt{pb}$}
\] in $A/\U$, the cogap map $\sigma_{f,h} : \colimm^A(F \times_V Y) \to_A  \colimm^A(F) \times_V Y$
is an isomorphism (as a cocone morphism).\footnote{The cocone under $F \times_V Y$ that induces this cogap map is actually nontrivial to construct and relies on the bicategorical structure of $A/\U$. See \cite[\href{https://github.com/PHart3/colimits-agda/blob/v0.5.0/Pullback-stability/Stability-cos-coc.agda}{Stability-cos-coc}]{agda-colim-TR} for a mechanized construction of the map.} (See below \cref{pbcoslice} for an explicit construction of pullback squares in $A/\U$ from those in $\U$.)

\begin{lemma} \label{crlim}
 The forgetful functor $\F : A/\U \to \U$ preserves limits.
\end{lemma}
\begin{proof}
Consider the free functor $\mleft(\mleft({-}\mright) + A\mright): \U \to A/\U$. This is left adjoint to $\F$. Let $F$ be an $A$-diagram over a graph $\Gamma$. Let $\left(C, r, K\right)$ be a limiting cone over $F$ (the dual notion to colimiting cocone). We must show that the function $ \left(\F(C, r, K) \circ {-}\right)$ is an equivalence. But for each $X : \U$, it is easy to check that this equals the composite of equivalences
\begin{align*}
& \ X \to \pr_1(C)
\\ \simeq & \ \left(X + A\right) \to_A C
\\ \simeq & \ \limm(\left(X + A\right) \to_A F)
\\ \simeq & \ \limm(X \to \F(F)) \qedhere
\end{align*}
\end{proof}

\begin{theorem}[{\cite[\href{https://github.com/PHart3/colimits-agda/blob/v0.5.0/Pullback-stability/Stability-ord.agda}{Stability-ord}]{agda-colim-TR}}] \label{LCC}
All colimits in $\U$ are universal.
\end{theorem}

\begin{corollary} \label{univcolim}
For each tree $\Gamma$ and each $A$-diagram $F$ over $\Gamma$, the colimit $\colimm^A(F)$ is universal.
\end{corollary}
\begin{proof}
Suppose that $\Gamma$ is a tree and consider the pullback square \eqref{pbs}. By \cref{create,crlim}, we have a cocone isomorphism under the $\U$-valued diagram $\F(F \times_V Y)$:
\[\begin{tikzcd}[row sep = small]
	& {\F(F \times_V Y)} \\
	{\pr_1(\colimm^A(F) \times_V Y)} && {\pr_1(\colimm^A(F)) \times_{\pr_1(V)} \pr_1(Y)}
	\arrow[from=1-2, to=2-1]
	\arrow[from=1-2, to=2-3]
	\arrow["\simeq"', from=2-1, to=2-3]
\end{tikzcd}\]
 By \cref{coluniq:p2,LCC}, it follows that $\pr_1(\colimm^A(F) \times_V Y)$ is a colimit of $\F(F \times_V Y)$. By \cref{create} again, $\colimm^A(F) \times_V Y$ is a colimit of $F \times_V Y$. Finally, by \cref{coluniq:p1}, $\sigma_{f,h}$ is an isomorphism.
\end{proof}

\begin{note}[{\cite[\href{https://github.com/PHart3/colimits-agda/blob/v0.5.0/Pullback-stability/Cos-pullback.agda}{Cos-pullback}]{agda-colim-TR}}]  \label{pbcoslice}
We can construct pullbacks in $A/\U$ as follows. Consider a cospan $\S \coloneqq X \xrightarrow{f} Z \xleftarrow{g} Y$ in $A/\U$ and form the standard pullback of $\F(\S)$ in $\U$~\cite[Definition 4.1.1]{Avi}:
\[
\Phi(\pr_1(f),\pr_1(g)) \ \coloneqq \ \sum_{x:\pr_1(X)}\sum_{y:\pr_1(Y)}\pr_1(f)(x) = \pr_1(g)(y)
\] Define $\mu_{f,g} : A \to \Phi(\pr_1(f),\pr_1(g))$ by $ \mu_{f,g}(a) \coloneqq  \left(\pr_2(X)(a), \pr_2(Y)(a), \pr_2(f)(a) \cdot \pr_2(g)(a)^{-1} \right)$. Now we have a cone over $\S$:
\[
\label{spb}
\begin{tikzcd}[row sep = 20, column sep = large]
	{\mleft(\Phi(\pr_1(f),\pr_1(g)), \mu_{f,g}\mright)} & Y \\
	X & Z
	\arrow[""{name=0, anchor=center, inner sep=0}, "{\left(\pi_y, \refl_x\right)}", from=1-1, to=1-2]
	\arrow["{\left(\pi_x, \refl_y\right)}"', from=1-1, to=2-1]
	\arrow["g", from=1-2, to=2-2]
	\arrow[""{name=1, anchor=center, inner sep=0}, "f"', from=2-1, to=2-2]
	\arrow["{\left\langle{\left(x,y,p\right) \mapsto p , H_p }\right\rangle}"{description, pos=0.6}, shift right=4, draw=none, from=0, to=1]
\end{tikzcd} \tag{$\texttt{sq}$}
\] Here, $H_p(a)$ denotes the evident path $\left(\pr_2(f)(a) \cdot \pr_2(g)(a)^{-1}\right)^{-1} \cdot \pr_2(f)(a) = \pr_2(g)(a)$ for each $a : A$, and $\pi$ denotes field projection for a $\Sigma$-type. We claim that \eqref{spb} is a pullback square, i.e., the function
$\left( \mathsf{sq} \circ {-} \right) : \left(\left(T, f_T\right) \to_A \left(\Phi(\pr_1(f),\pr_1(g)), \mu_{f,g}\right)\right) \to \mathsf{Cone}(\left(T, f_T\right); \S)$ is an equivalence for each $\left(T, f_T\right) : A/\U$. Indeed, for all cones $K \coloneqq \left(k_1, k_2, \left\langle{q,Q}\right\rangle\right) : \mathsf{Cone}(\left(T, f_T\right); \S)$, over $\S$, the fiber $\fib_{\left( \mathsf{sq} \circ {-} \right)}(K)$ is equivalent to the type of tuples
\begin{align*}
\begin{split}
d & \ : \  T \to \Phi(\pr_1(f),\pr_1(g))
\\ h_1 & \ : \ \pi_x \circ d \sim \pr_1(k_1)
\\ h_2 & \ : \ \pi_y \circ d \sim \pr_1(k_2)
\\ \tau & \ : \ \prod_{t : T}\ap_{\pr_1(f)}(h_1(t)) \cdot q(t) = \pi_p(d(t)) \cdot \ap_{\pr_1(g)}(h_2(t)) 
\end{split}
\begin{split}
d_p & \ : \ d \circ f_T \sim \mu_{f,g}
\\ H_1 & \ : \ \prod_{a :A} \ap_{\pi_x}(d_p(a)) = h_1(f_T(a)) \cdot \pr_2(k_1)(a)
\\ H_2 & \ : \ \prod_{a :A} \ap_{\pi_y}(d_p(a)) = h_2(f_T(a)) \cdot \pr_2(k_2)(a)
\\ \nu & \ : \ \prod_{a : A} \Lambda(\tau, H_1, H_2,  d_p, a) = Q(a)
\end{split}
\end{align*} 
where $\Lambda(\tau, H_1, H_2, d_p, a)$ is the path $\ap_{\pr_1(f)}(\pr_2(k_1)(a)) \cdot \pr_2(f)(a) = q(f_T(a)) \cdot \ap_{\pr_1(g)}(\pr_2(k_2)(a)) \cdot \pr_2(g)(a)$ obtained via $\tau(f_T(a))$, $H_1(a)$, $H_2(a)$, and $d_p(a)$. The four left-hand fields make up the fiber of $\mleft( \F(\mathsf{sq}) \circ {-} \mright)$ over $\F(K)$, which is contractible since $\F(\mathsf{sq})$ is the pullback of $\F(\S)$. As dependent sums preserve truncation level, it suffices to show that the four right-hand fields are contractible for each choice $\mleft(d, h_1, h_2, \tau\mright)$ of left-hand fields. We have two cone morphisms $\F(\mathsf{sq}) \circ \mu_{f,g} \to \F(\mathsf{sq})$ as follows. (A \emph{cone morphism $\K_1 \to \K_2$} consists of a function $\textsf{m-tip} : \tip(\K_1) \to \tip(\K_2)$, commuting triangles $\textsf{sq-left}$ and $\textsf{sq-right}$ witnessing that $\textsf{m-tip}$ commutes with the cones' left legs and right legs, respectively, and a path $\textsf{sq-coh}(x)$ for each $x : \tip(\K_1)$ witnessing that $\textsf{sq-left}(x)$ and $\textsf{sq-right}(x)$ fit into a commuting square with the cones' square homotopies.) The function $\mu_{f,g}$ immediately induces such a cone morphism. In addition, the function $d \circ f_T$ has the following cone morphism structure:
\[\begin{tikzcd}[column sep = small, /tikz/column 1/.append style={column sep=60pt},  /tikz/column 4/.append style={column sep=60pt}]
	A && \Phi(\pr_1(f),\pr_1(g)) & A && \Phi(\pr_1(f),\pr_1(g)) \\
	& {\pr_1(X)} &&& {\pr_1(Y)}
	\arrow[""{name=0, anchor=center, inner sep=0}, "{d \circ f_T}", from=1-1, to=1-3]
	\arrow["{\pr_2(X)}"', from=1-1, to=2-2]
	\arrow["{\pi_x}", from=1-3, to=2-2]
	\arrow[""{name=1, anchor=center, inner sep=0}, "{d \circ f_T}", from=1-4, to=1-6]
	\arrow["{\pr_2(Y)}"', from=1-4, to=2-5]
	\arrow["{\pi_y}", from=1-6, to=2-5]
	\arrow["{h_1(f_T(a)) \cdot \pr_2(k_1)(a)}"{description}, shift left=3, draw=none, from=0, to=2-2]
	\arrow["{h_2(f_T(a)) \cdot \pr_2(k_2)(a)}"{description}, shift left=3, draw=none, from=1, to=2-5]
\end{tikzcd}\] Here, the coherence $\textsf{sq-coh}$ is obtained easily from the data of $K$ along with $\pr_2(f)(a)$, $h_1(f_T(a))$, $h_2(f_T(a))$, and $\tau(f_T(a))$.
Now, by \cref{SIP}, the right-hand fields are collectively equivalent to paths between these two cone morphisms, and the type of cone morphisms into $\F(\mathsf{sq})$ is contractible by the definition of pullback. It follows that the right-hand fields are contractible, as desired.
\end{note}

\begin{remark}
The wild category $A/\U$ is usually \emph{not} LCC. Indeed, it is not LCC whenever $A$ is connected. In this case, suppose, for example, that $\Gamma$ is the discrete graph on $\2$ and define the $A$-diagram $F$ over $\Gamma$ by $F_i \coloneqq  \left(A, \idd_A\right)$ for each $i : \2$. By a direct calculation using \cref{coscoprodeqv}, we have that
\[
 \pr_1(\colimm^A(F) \times_{\1} \2) \simeq A + A  \not\simeq A + A + A \simeq \pr_1(\colimm^A(F \times_{\1} \2)) 
\] where $A \to \2$ is defined by, say, $a \mapsto 0$.
 
By the classical adjoint functor theorem, a locally presentable $\infty$-category is LCC if and only if all its colimits are universal. In this light, \cref{univcolim} may be seen as a lower bound on how close $A/\U$ is to being LCC.
\end{remark}

\section{Coslice colimits preserve connected maps} \label{Presv}

 The central result of this section is that $\colimm^A$ preserves the connected maps of an OFS on $\U$.

Let $\Gamma$ be a graph. Consider the wild category $\D_{\Gamma}$ of diagrams over $\Gamma$ valued in $\U$. The object type is $\sum_{F : \Gamma_0 \to \U}\prod_{i,j : \Gamma_0}\Gamma_1(i,j) \to F_i \to F_j$, and the morphisms from $F$ to $G$ are the natural transformations $F \Rightarrow G$, i.e., functions $\alpha : \prod_{i : \Gamma_0}F_i \to G_i$ equipped with a homotopy $G_{i,j,g} \circ \alpha_i \sim \alpha_j \circ F_{i,j,g}$ for each edge $g : \Gamma_1(i,j)$. The identity natural transformation is
$ \idd_F \coloneqq \left(\lambda{i}.\idd_{F_i}, \lambda{i}\lambda{j}\lambda{g}\lambda{x}.\refl_{F_{i,j,g}(x)} \right)$,
and the composition of natural transformations is defined as
\begin{align*}
 & \circ \ : \ \left(G \Rightarrow H\right) \to \left(F \Rightarrow G\right)  \to \left(F \Rightarrow H\right)
\\  & \left(\rho, q\right) \circ  \left(\alpha, p\right)  \  \coloneqq \  \mleft(\lambda{i}.\rho_i \circ \alpha_i, \left(q \ast p\right)(i,j,g,x)\mright)
\\ & \quad \textit{where $\left(q \ast p\right)(i,j,g,x) \coloneqq  q_{i,j,g}(\alpha(x)) \cdot \ap_{\rho_j}(p_{i,j,g}(x))$}
\end{align*}

\begin{lemma} \label{inddiag}
For all $\left(\alpha,p\right), \left(\rho, q\right) : F \Rightarrow G$, the canonical function $\happly_{\Gamma} : \mleft(\left(\alpha,p\right) = \left(\rho, q\right)\mright) \to \mleft(\left(\alpha,p\right) \sim_d \left(\rho,q\right)\mright)$ is an equivalence. Here, the latter type denotes the type of \emph{homotopies} between $\left(\alpha,p\right)$ and $\left(\rho,q\right)$, i.e., functions $W :\prod_{i : \Gamma_0}\alpha_i \sim \rho_i$ equipped with a commuting square
\[\begin{tikzcd}[column sep = huge]
	{G_{i,j,g}(\alpha_i(x))} & {G_{i,j,g}(\rho_i(x))} \\
	{\alpha_j(F_{i,j,g}(x))} & {\rho_j(F_{i,j,g}(x))}
	\arrow["{\ap_{G_{i,j,g}}(W_i(x))}", equals, from=1-1, to=1-2]
	\arrow["{p_{i,j,g}(x)}"', equals, from=1-1, to=2-1]
	\arrow["{q_{i,j,g}(x)}", equals, from=1-2, to=2-2]
	\arrow["{W_j(F_{i,j,g}(x))}"', equals, from=2-1, to=2-2]
\end{tikzcd}\]
for all $g : \Gamma_1(i,j)$ and $x : F_i$. 
\end{lemma} 
\begin{proof}
By \cref{SIP}.
\end{proof}

With this notion of homotopy between natural transformations, the identity and associativity laws for $\D_{\Gamma}$ are easily defined with path induction.

\begin{notation} $ $
\begin{itemize}
\item Define $\left\langle{W,C}\right\rangle \coloneqq \happly_{\Gamma}^{-1}(W,C)$.
\item Variables of the form $\alpha^* : F \Rightarrow G$ are abbreviations of pairs $\left(\alpha, \alpha_p\right)$.
\end{itemize}
\end{notation}

\begin{lemma}\label{SIPnattr}
Let $\left(\alpha,p\right), \left(\rho, q\right) : F \Rightarrow G$. For all $\left(W_1, C_1\right), \left(W_2, C_2\right) : \left(a, p\right) \sim_d \left(\rho, q\right)$, the type $\left(W_1, C_1\right) = \left(W_2, C_2\right)$ is equivalent to the type of  $H : W_1 \sim W_2$ equipped with a commuting triangle
\[\begin{tikzcd}[column sep = 45, row sep = large]
	& {\ap_{G_{i,j,g}}(W_2(i,x)) \cdot q_{i,j,g}(x) \cdot W_2(j,F_{i,j,g}(x))^{-1}} \\
	{p_{i,j,g}(x)} & {\ap_{G_{i,j,g}}(W_1(i,x)) \cdot q_{i,j,g}(x) \cdot W_1(j,F_{i,j,g}(x))^{-1}}
	\arrow["{\textit{via $H(i,x)$ and $H(j, F_{i,j,g}(x))$}}", Rightarrow, no head, from=1-2, to=2-2]
	\arrow["{C_2(i,j,g,x)}", Rightarrow, no head, from=2-1, to=1-2]
	\arrow["{C_1(i,j,g,x)}"', Rightarrow, no head, from=2-1, to=2-2]
\end{tikzcd}\] for all $g : \Gamma_1(i,j)$ and $x : F_i$.
\end{lemma}
\begin{proof}
By \cref{SIP}.
\end{proof}

\Cref{inddiag} lets us ``path induct'' on homotopies of transformations (see \cref{idsys}), thereby making \cref{precompdg,postcomp2,concat2} (stated next) provable with simple path algebra.

\begin{lemma}[Left whiskering of $\sim_d$] \label{precompdg}
Let $\zeta^{\ast} : F  \Rightarrow G$. Let $\alpha^{\ast},\rho^{\ast} : G \Rightarrow H$. For every $\left(W,C\right) : \alpha^{\ast} \sim_d \rho^{\ast}$, we have a path
$\ap_{{-} \circ \zeta^{\ast}}\mleft(\left\langle{W,C}\right\rangle\mright)  =  \left\langle{\lambda{i}.W_i(\zeta_i({-})),  \tau_{W, C}}\right\rangle$
between elements of the identity type
\[\begin{tikzcd}[column sep = large, /tikz/column 3/.append style={column sep=15pt}, row sep = large]
	{F_i} && {F_j} && {F_i} && {F_j} \\
	{H_i} && {H_j} && {H_i} && {H_j}
	\arrow[""{name=0, anchor=center, inner sep=0}, "{F_{i,j,g}}", from=1-1, to=1-3]
	\arrow[""{name=1, anchor=center, inner sep=0}, "{\alpha_j \circ \zeta_j}"{description}, from=1-3, to=2-3]
	\arrow["{\alpha_i \circ \zeta_i}"{description}, from=1-1, to=2-1]
	\arrow[""{name=2, anchor=center, inner sep=0}, "{H_{i,j,g}}"', from=2-1, to=2-3]
	\arrow[""{name=3, anchor=center, inner sep=0}, "{\rho_i \circ \zeta_i}"{description}, from=1-5, to=2-5]
	\arrow[""{name=4, anchor=center, inner sep=0}, "{H_{i,j,g}}"', from=2-5, to=2-7]
	\arrow["{\rho_j \circ \zeta_j}"{description}, from=1-7, to=2-7]
	\arrow[""{name=5, anchor=center, inner sep=0}, "{F_{i,j,g}}", from=1-5, to=1-7]
	\arrow["{\alpha_p \ast \zeta_p}"{description}, draw=none, from=0, to=2]
	\arrow[shorten <=26pt, shorten >=26pt, Rightarrow, no head, from=1, to=3]
	\arrow["{\rho_p \ast \zeta_p}"{description}, draw=none, from=5, to=4]
\end{tikzcd}\]
where the path $\tau_{W,C}(i,j,g,x)$ is obtained via $C_{i,j,g}(\zeta_i(x))$ and $\zeta_p(i,j,g,x)$.
\end{lemma}

\begin{lemma}[Right whiskering of $\sim_d$] \label{postcomp2}
Let $\zeta^{\ast} : G \Rightarrow H$. Let $\alpha^{\ast}, \rho^{\ast} : F \Rightarrow G$. For every $\left(W, C\right) : \alpha^{\ast} \sim_d \rho^{\ast}$, we have a path
$\ap_{ \zeta^{\ast} \circ {-}}\mleft(\left\langle{W, C}\right\rangle\mright)  =   \left\langle{\lambda{i}.\ap_{\zeta_i}(W_i({-})), \tau_{W,C}}\right\rangle$
between elements of the identity type
\[\begin{tikzcd}[row sep = large, column sep = large, /tikz/column 3/.append style={column sep=15pt}]
	{F_i} && {F_j} && {F_i} && {F_j} \\
	{H_i} && {H_j} && {H_i} && {H_j}
	\arrow[""{name=0, anchor=center, inner sep=0}, "{F_{i,j,g}}", from=1-1, to=1-3]
	\arrow[""{name=1, anchor=center, inner sep=0}, "{\zeta_j \circ \alpha_j}"{description}, from=1-3, to=2-3]
	\arrow["{\zeta_i \circ \alpha_i}"{description}, from=1-1, to=2-1]
	\arrow[""{name=2, anchor=center, inner sep=0}, "{H_{i,j,g}}"', from=2-1, to=2-3]
	\arrow[""{name=3, anchor=center, inner sep=0}, "{\zeta_i \circ \rho_i}"{description}, from=1-5, to=2-5]
	\arrow[""{name=4, anchor=center, inner sep=0}, "{H_{i,j,g}}"', from=2-5, to=2-7]
	\arrow["{\zeta_j \circ \rho_j}"{description}, from=1-7, to=2-7]
	\arrow[""{name=5, anchor=center, inner sep=0}, "{F_{i,j,g}}", from=1-5, to=1-7]
	\arrow["{\zeta_p \ast \alpha_p}"{description}, draw=none, from=0, to=2]
	\arrow[shorten <=26pt, shorten >=26pt, Rightarrow, no head, from=1, to=3]
	\arrow["{\zeta_p \ast \rho_p}"{description}, draw=none, from=5, to=4]
\end{tikzcd}\]
where the path $\tau_{W,C}(i,j,g,x)$ is obtained via $W_i(x)$ and $C_{i,j,g}(x)$.
\end{lemma} 

\begin{lemma}[Composition of $\sim_d$] \label{concat2}
Let $\alpha^{\ast},\rho^{\ast}, \epsilon^{\ast} : F \Rightarrow G$. Let $\left(W, C\right): \alpha^{\ast} \sim_d \rho^{\ast}$ and $\left(Y, D \right) : \rho^{\ast} \sim_d \epsilon^{\ast}$. We have a path
\[
 \left\langle{W,C}\right\rangle \cdot  \left\langle{Y,D}\right\rangle \ =_{\alpha^{\ast} = \epsilon^{\ast}} \  \left\langle{\lambda{i}.W_i({-}) \cdot Y_i({-}),  \tau_{C,D}  }\right\rangle
\]
where the path $\tau_{C,D}(i,j,g,x)$ is obtained via $C_{i,j,g}(x)$ and $D_{i,j,g}(x)$.
\end{lemma}

\begin{lemma} \label{diagbic}
The wild category $\D_{\Gamma}$ is a bicategory.
\end{lemma}
\begin{proof}
The unit and associativity laws of maps hold definitionally in $\U$. Thus, by \cref{precompdg,postcomp2,concat2} combined with \cref{SIPnattr}, verifying that $\D_{\Gamma}$ is a bicategory reduces to simple path algebra, which we omit here.
\end{proof}

\begin{lemma}\label{diaguniv}
Assuming the univalence axiom, $\D_{\Gamma}$ is univalent.
\end{lemma}
\begin{proof}
Let $F, G : \D_{\Gamma}$ and $F \xRightarrow{\sim} G$ denote the type of natural transformations that are levelwise equivalences in $\U$. By univalence, \cref{SIP} implies that the canonical function $F = G  \to F \xRightarrow{\sim} G$ is an equivalence. Further, by the associated induction principle for $\xRightarrow{\sim}$, we see that every levelwise equivalence is an equivalence in $\D_{\Gamma}$, and the converse implication is clear. This biimplication is between propositions, so we have an equivalence $F \xRightarrow{\sim} G \to F \simeq_{\D_{\Gamma}} G$ that sends the identity to the identity. The composite equivalence $\mleft(F = G \mright) \simeq \mleft(F \simeq_{\D_{\Gamma}} G\mright)$ witnesses that $\D_{\Gamma}$ is univalent.
\end{proof}

\subsubsection*{Lifting an OFS to diagrams}

Let $\left(\L, \RI\right)$ be an OFS on $\U$. We lift this OFS to one on $\D_{\Gamma}$ as follows.

\begin{theorem} \label{unfact}
For all $F, G : \D_{\Gamma}$ and $\left(h, \alpha\right): F \Rightarrow G$, define the predicates
$\widehat{\L}(h, \alpha)  \coloneqq  \prod_{i : \Gamma_0}\L(h_i)$ and $\widehat{\RI}(h, \alpha) \coloneqq \prod_{i : \Gamma_0}\RI(h_i)$.
 Let $\left(h, \alpha\right) : F \Rightarrow G$. The following type is contractible: 
\[
\fact_{\widehat{\L},\widehat{\RI}}(h, \alpha) \ \coloneqq \ \sum_{A : \D_{\Gamma}}\sum_{S : F \Rightarrow A}\sum_{T : A \Rightarrow G}\left(T \circ S \sim_d \left(h, \alpha\right)\right) \times \widehat{\L}(S) \times \widehat{\RI}(T)
\] 
\end{theorem}
\begin{proof}
By its definition, $\fact_{\widehat{\L},\widehat{\RI}}(h, \alpha)$ is equivalent to the type of tuples
\begin{align*}
\begin{split}
A_0 & \ : \ \Gamma_0 \to \U 
\\   S_0 & \ : \ \prod_{i : \Gamma_0}F_i \to A_0(i) 
\\ T_0 & \ : \ \prod_{i : \Gamma_0}A_0(i) \to G_i
\\ P & \ : \ \prod_{i : \Gamma_0} T_0(i) \circ S_0(i) \sim h_i
\\ L & \ : \ \prod_{i : \Gamma_0}\L(S_0(i))
\\ R & \ : \ \prod_{i : \Gamma_0}\RI(T_0(i))
\end{split}
\qquad
\begin{split}
 A_1 & \ : \ \prod_{i,j : \Gamma_0}\prod_{g : \Gamma_1(i,j)}A_0(i) \to  A_0(j)
\\ S_1 & \ : \ \prod_{i,j,g} A_1(i,j,g) \circ S_0(i) \sim S_0(j) \circ F_{i,j,g}
\\ T_1 & \ : \ \prod_{i,j,g} G_{i,j,g} \circ T_0(i) \sim T_0(j) \circ A_1(i,j,g)
\\ C & \ : \ \prod_{i,j,g}\prod_{x : F_i} \left(T_1 \ast S_1\right)(i,j,g,x) \cdot P_j(F_{i,j,g}(x)) = \ap_{G_{i,j,g}}(P_i(x)) \cdot \alpha_{i,j,g}(x)
\end{split}
\end{align*}
We can contract the six left-hand fields because $\fact_{\L, \RI}(h_i)$ is contractible for each $i : \Gamma_0$. Let $\left(A_0, S, T, P, L, R\right)$ be the unique tuple of the first six fields and consider the type of the last four fields $\mathsf{coher}_{\L, \RI}(A_0, S, T, P, L, R)$. We want to prove that $\mathsf{coher}_{\L, \RI}(A_0, S, T, P, L, R)$ is contractible. Since $\Pi$-types distribute over $\Sigma$-types, it is equivalent to the type of dependent functions taking an edge $g : \Gamma_1(i,j)$ to a diagonal filler $f : A_i \to A_j$ equipped with a pair of commuting subtriangles like so
\[\begin{tikzcd}[row sep = large, column sep = large]
	{F_i} && {A_j} \\
	{A_i} && {G_j}
	\arrow[""{name=0, anchor=center, inner sep=0}, "{{{S_j \circ F_{i,j,g}}}}", from=1-1, to=1-3]
	\arrow["{{{S_i}}}"', from=1-1, to=2-1]
	\arrow["{{{T_j}}}", from=1-3, to=2-3]
	\arrow[""{name=1, anchor=center, inner sep=0}, "{{{f}}}"{description}, from=2-1, to=1-3]
	\arrow[""{name=2, anchor=center, inner sep=0}, "{{{G_{i,j,g} \circ T_i}}}"', from=2-1, to=2-3]
	\arrow["{{{s}}}"{description}, shift right=5, draw=none, from=0, to=1]
	\arrow["{{ t}}"{description}, shift right=5, draw=none, from=2, to=1]
\end{tikzcd}\]
as well as a path $c(x) : t(S_i(x)) \cdot \ap_{T_j}(s(x)) \cdot P_j(F_{i,j,g}(x))  =  \ap_{G_{i,j,g}}( P_i(x)) \cdot \alpha_{i,j,g}(x)$ for each $x : F_i$. For each $g : \Gamma_1(i,j)$, letting $W(x) \coloneqq \ap_{G_{i,j,g}}(P_i(x)) \cdot \alpha_{i,j,g}(x) \cdot P_j(F_{i,j,g}(x))^{-1}$ for all $x : F_i$, this type is clearly equivalent to the type of diagonal fillers of the square
\[\begin{tikzcd}
	{F_i} && {A_j} \\
	{A_i} && {G_j}
	\arrow["{{{{S_j \circ F_{i,j,g}}}}}", from=1-1, to=1-3]
	\arrow[""{name=0, anchor=center, inner sep=0}, "{{{{S_i}}}}"', from=1-1, to=2-1]
	\arrow[""{name=1, anchor=center, inner sep=0}, "{{{{T_j}}}}", from=1-3, to=2-3]
	\arrow["{{{{G_{i,j,g} \circ T_i}}}}"', from=2-1, to=2-3]
	\arrow["{W}"{description}, draw=none, from=0, to=1]
\end{tikzcd}\] which is contractible by \cref{ULP}.
\end{proof}

\begin{corollary}[{\cite[\href{https://github.com/PHart3/colimits-agda/blob/v0.5.0/HoTT-Agda/core/lib/wild-cats/Diag-ty-OFS.agda}{Diag-ty-OFS}]{agda-colim-TR}}]
Every OFS on $\U$ lifts levelwise to $\D_{\Gamma}$.
\end{corollary}

The wild adjunction $\colimm({-}) \dashv \const_{\Gamma} : \D_{\Gamma} \rightleftarrows \U$ satisfies the coherence condition \eqref{natscohadj}~\cite[{\href{https://github.com/PHart3/colimits-agda/blob/v0.5.0/HoTT-Agda/theorems/homotopy/ColimAdjoint-hex.agda}{ColimAdjoint-hex}}]{agda-colim-TR}. Further, $\const_{\Gamma} : \U \to \D_{\Gamma}$ clearly takes $\RI$ to $\widehat{\RI}$. It follows that $\colimm({-})$ takes $\widehat{\L}$ to $\L$ by \cref{fspres} (which applies here thanks to \cref{diagbic,diaguniv}). 

Let $A : \U$. For all $X, Y : A/\U$, consider the predicate $\L_A(f,p) \coloneqq \L(f)$ on $X \to_A Y$. Then the wild functor $\colimm^A$ takes $\widehat{\L}_A$ (the class of maps levelwise in $\L_A$) to $\L_A$~\cite[\href{https://github.com/PHart3/colimits-agda/blob/v0.5.0/Colimit-coslice/OFS-Preserve/CosColim-lftclass.agda}{CosColim-lftclass}]{agda-colim-TR}. Indeed, for each map $\delta : \A \Rightarrow \B$ of $A$-diagrams, the underlying function of $\colimm^A(\delta)$ is induced by the span map
\[\begin{tikzcd}
	A & {\colimm{A}} & {\colimm(\F(\A))} \\
	A & {\colimm{A}} & {\colimm(\F(\B))}
	\arrow["\idd"', from=1-1, to=2-1]
	\arrow["\idd", from=1-2, to=2-2]
	\arrow[from=1-2, to=1-1]
	\arrow[from=2-2, to=2-1]
	\arrow[from=1-2, to=1-3]
	\arrow[from=2-2, to=2-3] 
	\arrow["{\hat{\delta}}", from=1-3, to=2-3]
\end{tikzcd}\]  The left and middle legs here belong to $\L$ because $\L$ contains all identities. Since $\colimm({-})$ takes $\widehat{\L}$ to $\L$, the right leg is also in $\L$ when $\delta$ is in $\widehat{\L}_A$. Therefore, the commuting square \eqref{pocolcsq} implies that $\colimm^A(\delta)$ belongs to $\L_A$ when $\delta$ is in $\widehat{\L}_A$.

In particular, if $F$ is a $\U^{\ast}$-valued (or \emph{pointed}) diagram over $\Gamma$ such that each $\pr_1(F_i)$ is $\left(\L,\RI\right)$-connected, then the type $\colimm^{\ast}(F)$ is also $\left(\L,\RI\right)$-connected. (A type $X :\U$ is \emph{$\left(\L, \RI\right)$-connected} if the function $X \to \1$ belongs to $\L$.) Indeed, since $\colimm^{\ast}{\1}$ is contractible, $\colimm^{\ast}$ takes the unique map $F \Rightarrow_{\ast} \1$ of pointed diagrams to the terminal map $\left(c, c_p\right) : \colimm^{\ast}(F) \to_{\ast} \1$ such that $c \in \L$.

\begin{example}\label{npres} $ $
\begin{exmpenum}
\item \label{npres:p1} 
For each truncation level $n$, if each $\pr_1(F_i)$ is $n$-connected, then so is $\pr_1(\colimm^{\ast}(F))$. (In fact, if $F$ is an $A$-diagram with each $\pr_1(F_i)$ $n$-connected and $A$ is $n$-connected, then \cref{coeqcop} shows that the underlying type of $\colimm^A(F)$ is also $n$-connected.)
\item \label{npres:p2} 
Let $\Gamma$ be the graph with a single point $\ast$ and a single edge from $\ast$ to $\ast$. Define the diagram $F$ over $\Gamma$ by $F(\ast) \coloneqq \1$ and $F_{\ast, \ast, \ast} \coloneqq \idd_{\1}$. Then $\colimm(F) = S^1 $, which proves that $\colimm$ does not preserve $n$-connectedness when $n \geq 1$, unlike $\colimm^{\ast}$.
\end{exmpenum}
\end{example}

\subsection{Colimits of higher groups}\label{cocomp}

Consider truncation levels ${-2} \leq n \leq \infty$ and ${-1} \leq k < \infty$. Recall from \cite{Higher} the wild category ${\left(n,k\right)\mathsf{GType}}$ of \emph{$k$-tuply groupal $n$-groupoids}---an object of which, called a \emph{higher group}, is a pointed $n$-type equipped with a $k$-fold delooping. This category is isomorphic to the full subcategory $\U^{\ast}_{ \geq {k-1}, \leq n+k}$ of $\U^{\ast}$ on $\left(k-1\right)$-connected, $\left(n+k\right)$-truncated pointed types. Consider the full subcategory $\U^{\ast}_{\geq {k-1}}$ of $\U^{\ast}$ on those objects whose underlying types are $\left(k-1\right)$-connected. By \cref{npres:p1}, this subcategory inherits colimits from $\U$. For each truncation level $m$, note that truncations preserve $m$-connectedness and that the function $A \to \lVert{A}\rVert_m$ is an equivalence when $A$ is $m$-truncated. By \cref{modpres}, we now see that ${\left(n,k\right)\mathsf{GType}}$ (the full subcategory of $\U^{\ast}_{\geq {k-1}}$ on $\left(n+k\right)$-truncated types) is reflective in $\U^{\ast}_{\geq {k-1}}$, in the sense of \cref{reflin}. This gives us a way to build colimits in ${\left(n,k\right)\mathsf{GType}}$ from those in $\U^{\ast}$.

Let $\left(\L, \RI\right)$ be an OFS on $\U$ and $A: \U$. We know from \cref{coprodcos} that coproducts in $A/\U$ are definable as ordinary colimits over trees. This means that our pushout-coproduct construction of coslice colimits (\cref{coeqcop}) forms $\colimm^A$ from ordinary colimits over trees, which preserve $\left(\L, \RI\right)$-connectedness by \cref{treecolim}. Thus, if $A$ is $\left(\L, \RI\right)$-connected, the full subcategory of $A/\U$ on $\left(\L,\RI\right)$-connected types has colimits. This closure property is crucial for our next application, which builds colimits of \emph{higher pointed abelian groups}.

Let $Q : \U \to \mathsf{Prop}$. Consider types $A$ and $B$ in the subuniverse $\U_Q$ and a function $\varphi : A \to B$. Define the \emph{coslice-coslice} wild category $\left(B, \varphi\right)/\left(A/\U_Q\right)$ as follows. The objects are commuting triangles of the form $\left(\left(Z, \underline{\hspace{2mm}}\right), g_Z, k, \alpha\right)$
\[\begin{tikzcd}[ampersand replacement=\&]
	\& A \\
	B \&\& Z
	\arrow["\varphi"', from=1-2, to=2-1]
	\arrow["{g_Z}", from=1-2, to=2-3]
	\arrow[""{name=0, anchor=center, inner sep=0}, "k"', from=2-1, to=2-3]
	\arrow["\alpha"{description}, draw=none, from=1-2, to=0]
\end{tikzcd}\]
and the morphisms $\left(Z_1, g_{Z_1}, k_1, \alpha_1\right) \to_{\varphi} \left(Z_2, g_{Z_2}, k_2, \alpha_2\right)$ are tuples
\begin{align*}
f & \ : \ Z_1 \to Z_2
\\ p & \ : \ f \circ g_{Z_1} \sim g_{Z_2}
\\ H & \ : \ f \circ k_1 \sim k_2
\\ K & \ : \ \prod_{a : A}   \ap_f(\alpha_1(a)) \cdot p(a) = H(\varphi(a)) \cdot \alpha_2(a) 
\end{align*} 
We have evident identity morphisms, and composition $\circ$ of morphisms is defined by
\begin{align*}
& \left(f_2, p_2, H_2, K_2\right) \circ \left(f_1, p_1, H_1, K_1\right) \ : \  \left(Z_1, g_{Z_1}, k_1, \alpha_1\right) \to_{\varphi} \left(Z_3, g_{Z_3}, k_3, \alpha_3\right)
\\ & \left(f_2, p_2, H_2, K_2\right) \circ \left(f_1, p_1, H_1, K_1\right)  \ \coloneqq \ \left(f_2 \circ f_1 , \lambda{a}.\ap_{f_2}(p_1(a)) \cdot p_2(a), \lambda{b}.\ap_{f_2}(H_1(b)) \cdot H_2(b), \sigma(K_2, K_1) \right) 
\end{align*}
where $\sigma(K_2, K_1,a)$ denotes the following chain of paths for each $a : A$:
\[\begin{tikzcd}
	{\ap_{f_2 \circ f_1}(\alpha_1(a)) \cdot \ap_{f_2}(p_1(a)) \cdot p_2(a)} \\
	{ \ap_{f_2}(H_1(\varphi(a)) \cdot \alpha_2(a)) \cdot p_2(a)} \\
	{ \left(\ap_{f_2}(H_1(\varphi(a))) \cdot H_2(\varphi(a))\right) \cdot \alpha_3(a)}
	\arrow["{{{\emph{via $K_1(a)$}}}}", equals, from=1-1, to=2-1]
	\arrow["{{{\emph{via $K_2(a)$}}}}", equals, from=2-1, to=3-1]
\end{tikzcd}\]

Next, we define  $0$-functors 
\[
\begin{tikzcd}
	{\left(B, \varphi\right)/\left(A/\U_Q\right)} & {B/\U_Q}
	\arrow["\gamma", shift left=3, from=1-1, to=1-2]
	\arrow["\xi", shift left=3, from=1-2, to=1-1]
\end{tikzcd}
\] as follows. Define $\gamma_0 : \obb(\left(B, \varphi\right)/\left(A/\U_Q\right)) \to \obb(B/\U_Q)$ by
$\gamma_0(Z, g_Z, k, \alpha) \coloneqq \left(Z, k\right)$. Conversely, define $\xi_0 : \obb(B/\U_Q) \to \obb(\left(B, \varphi\right)/\left(A/\U_Q\right))$ by $\xi_0(Z,k) \coloneqq \left(Z, k \circ \varphi, k, \refl_{k(\varphi({-}))}\right)$.
Next, define
\begin{align*}
& \gamma_1 \ : \  \homm_{\left(B, \varphi\right)/\left(A/\U_Q\right)}(\left(Z_1, g_{Z_1}, k_1, \alpha_1\right),\left(Z_2, g_{Z_2}, k_2, \alpha_2\right)) \to  \homm_{B/\U_Q}(\left(Z_1, k_1\right), \left(Z_2, k_2\right))
\\ & \gamma_1(f, p, H, k) \ \coloneqq \ \left(f, H\right)
\\ & \xi_1 \ : \ \homm_{B/\U_Q}(\left(Z_1, k_1\right), \left(Z_2, k_2\right)) \to \homm_{\left(B, \varphi\right)/\left(A/\U_Q\right)}\mleft(\left(Z_1, k_1 \circ \varphi, k_1, \refl_{k_1(\varphi({-}))}\right), \left(Z_2, k_2 \circ \varphi, k_2, \refl_{k_2(\varphi({-}))}\right)\mright)
\\ & \xi_1(f, H) \ \coloneqq \  \left(f, H \circ \varphi, H, \rid(H(\varphi({-}))) \right)
\end{align*} 
Note that $\xi$ preserves composition as follows:
\begin{align*}
& \ \xi_1(g \circ f , \ap_g(H_1) \cdot H_2)  
\\ \equiv & \ \left(g \circ f , \mleft(\ap_g(H_1) \cdot H_2 \mright) \circ \varphi , \ap_g(H_1) \cdot H_2 , \rid(\ap_g(H_1(\varphi({-}))) \cdot H_2(\varphi({-}))) \right)
\\ = & \ \left(g \circ f , \mleft(\ap_g(H_1) \cdot H_2 \mright) \circ \varphi , \ap_g(H_1) \cdot H_2 , \sigma(\rid(H_2(\varphi({-}))), \rid(H_1(\varphi({-})))) \right) \tag{$\textit{in the final component: for each $a : A$, $\PI(H_1(\varphi(a)), H_2(\varphi(a)))$}$}
\\ \equiv & \  \left(g, H_2 \circ \varphi, H_2, \rid(H_2(\varphi({-}))) \right) \circ \left(f, H_1 \circ \varphi, H_1, \rid(H_1(\varphi({-}))) \right)
\\ \equiv & \ \xi_1(g , H_2) \circ \xi_1(f , H_1)
\end{align*}

\begin{lemma}\label{coscosadj}
We have that $\xi$ is a $2$-coherent left adjoint to $\gamma$.
\end{lemma}
\begin{proof}
Let $\left(Z_1, g_{Z_1}, k_1, \alpha\right) : \obb(\left(B, \varphi\right)/\left(A/\U_Q\right))$ and $\left(Z_2, k_2\right) : \ob(B/\U_Q)$.
Define 
\begin{align*}
& \mu \ : \ \homm_{\left(B, \varphi\right)/\left(A/\U_Q\right)}\mleft(\left(Z_2, k_2 \circ \varphi, k_2, \refl_{k_2(\varphi({-}))} \right), \left(Z_1, g_{Z_1}, k_1, \alpha\right)\mright) \to \homm_{B/\U_Q}(\left(Z_2, k_2\right), \left(Z_1, k_1\right))
\\ & \mu(f, p, H, K)  \ \coloneqq \ \left(f, H\right)
\end{align*}
We claim that $\mu$ is contractible, hence an equivlanece. Indeed, for each $\left(g, I\right) : \homm_{B/\U_Q}(\left(Z_2, k_2\right), \left(Z_1, k_1\right))$,
\begin{align*}
& \fib_{\mu}(g, I)
\\ \simeq \ \ &  \sum_{f : Z_2 \to Z_1}\sum_{p :  f \circ k_2 \circ \varphi \sim g_{Z_1}}\sum_{H : f \circ k_2 \sim k_1}\sum_{K : \prod_{a :A}p(a) = H(\varphi(a)) \cdot \alpha(a)}\sum_{U : f \sim g}\prod_{b : B} H(b) = U(k_2(b)) \cdot I(b) 
\\ \simeq \ \ & \1 
\end{align*}
Now, $\mu$ is trivially natural in both variables, so that $2$-coherence in this case is easy to check.
\end{proof}

\begin{lemma}
The wild category $\left(B, \varphi\right)/\left(A/\U_Q\right)$ is reflective in $B/\U_Q$, i.e., it admits a $2$-coherent left adjoint from $B/\U_Q$ whose counit is an isomorphism.
\end{lemma}
\begin{proof}
The counit $\epsilon : \xi \circ \gamma \to \idd_{\left(B, \varphi\right)/\left(A/\U_Q\right)}$ of the adjunction of \cref{coscosadj} is an equivalence in each component: 
\begin{align*}
&  \epsilon_{Z^{\ast}} \ : \ \left(Z, k \circ \varphi, k, \refl_{k(\varphi({-}))}\right) \xrightarrow{\simeq}_{\varphi} \left(Z, g_Z, k, \alpha\right)
 \\ & \epsilon_{Z^{\ast}} \ \coloneqq \ \left(\idd_Z,\alpha,\refl_{k({-})}, \refl_{\alpha({-})} \right)
\end{align*}
This means that $\epsilon$ is an isomorphism.
\end{proof}

\begin{corollary}\label{pthicol}
Consider truncation levels ${-2} \leq n \leq \infty$ and ${-1} \leq k < \infty$. For each pointed type $G$ in ${\left(n,k\right)\mathsf{GType}}$, the coslice $G/{\left(n,k\right)\mathsf{GType}}$ is cocomplete.
\end{corollary}
\begin{proof}
As a wild category, $G/\U^{\ast}_{\geq {k-1}, \leq n+k}$ is reflective in $\pr_1(G)/\U_{\geq {k-1}, \leq n+k}$, which has colimits.
\end{proof}

For example, let $n: \N$ with $n >0$ and $m < n$. The Eilenberg-MacLane space $K(\Z, n+m)$ is the free $\left(n,m\right)$-group on one generator in the category ${\left(n,m\right)\mathsf{GType}}$, for which we view $\Omega^{n+m} : \left(n,m\right)\mathsf{GType} \to \mathbf{Set}$ as the forgetful functor. Indeed, letting $d \coloneqq n +m$, for all $\left(Y,y\right) : \U^{\ast}_{\geq {m-1},\leq d}$, we have the composite equivalence
\begin{align*}
& K(\Z,d) \to_{\ast} \left(Y,y\right)
\\  \equiv \ & \lN{\Sigma^{d-1}(K(\Z, 1))}\rN_d \to_{\ast} \left(Y,y\right)
\\ \simeq \ & S^d \to_{\ast} \left(Y,y\right)
\\ \simeq \ & \Omega^d(Y,y)
\end{align*}
Thus, when $m >0$, $K(\Z, d)/\U^{\ast}_{\geq {m-1},\leq d}$ is a higher version of the category of \emph{pointed abelian groups}~\cite{pab}. By \cref{pthicol}, we know how to build colimits of such higher pointed abelian groups.

\section{Weak continuity of cohomology} \label{contcohom}

For this section, we need the notion of \emph{finite graph}. 

\begin{definition} $ $
\begin{itemize}
\item We say that a type $X$ is \textit{finite} if it is merely equivalent to a standard finite type. (The word ``merely'' here refers to propositional truncation.)
\item We say that a graph $\Gamma$ is \textit{finite} if $\Gamma_0$ is finite and for all $i, j : \Gamma_0$, $\Gamma_1(i,j)$ is finite.
\end{itemize}
\end{definition}

\begin{lemma}\label{finprop}
If $\Gamma$ is a finite graph, then the type $\sum_{i,j : \Gamma_0}\Gamma_1(i,j)$ is finite.
\end{lemma}
\begin{proof}
By \cite[\href{https://unimath.github.io/agda-unimath/univalent-combinatorics.dependent-pair-types.html\#a-dependent-sum-of-finite-types-indexed-by-a-finite-type-is-finite}{A dependent sum of finite types indexed by a finite type is finite}]{agda-unimath}.
\end{proof}

Let $\Gamma$ be a finite graph. We claim that every Eilenberg-Steenrod cohomology theory $H : \left(\U^{\ast}\right)^{\op} \to \mathbf{Ab}$ takes pointed colimits over $\Gamma$ to weak limits in $\mathbf{Set}$, in the sense that the universal map from the limit is an epi in $\mathbf{Ab}$. If $H$ is additive (e.g., induced by an $\Omega$-spectrum), then this holds when $\Gamma$ is a \emph{projective} graph, i.e., both $\Gamma_0$ and $\Gamma_1(i,j)$ satisfy the set-level axiom of choice~\cite[Definition 6.1]{cohom}.

\subsection{Eilenberg-Steenrod cohomology}
Let $H$ be a $\Z$-indexed family of functors $\left(\U^{\ast}\right)^{\op} \to \mathbf{Ab}$. We say that $H$ is an \textit{(Eilenberg-Steenrod) cohomology theory} if it satisfies the following two axioms.

\begin{itemize}
\item For all $n : \Z$, we have a natural isomorphism $H^{n+1}(\Sigma{-}) \xrightarrow{\sigma_n} H^n({-})$ of functors $\U^{\ast} \to \mathbf{Ab}$.
\item For all maps $f : X \to_{\ast} Y$, the following sequence is exact (where $\sfrac{Y}{X}$ is the cofiber of $f$):
\[
H^n(\sfrac{Y}{X}) \xrightarrow{H^n(\cofiber(f))} H^n(Y) \xrightarrow{H^n(f)} H^n(X)
\]
\end{itemize}

\begin{example}
Suppose that $E : \Z \to \U^{\ast}$ is a prespectrum, with structure maps $\epsilon_n : E_n \to_{\ast} \Omega{E_{n+1}} $. For each $n : \Z$, we have a sequence
\[ 
\lN{X \to_{\ast} \Omega^k{E_{n+k}}}\rN_0 \xrightarrow{\lN{\Omega^k(\epsilon_{n+k}) \circ {-}}\rN_0} \lN{X \to_{\ast} \Omega^{k+1}{E_{n+\left(k +1\right)}}}\rN_0 \tag{$\mathtt{seq}_E$} \label{seqE}
\]
of abelian groups. (We assume that addition $+ : \Z \times \N \to \Z$ is defined by pattern matching on the second argument.) For each $n : \Z$, define
\begin{align*}
& \widetilde{E}^n  \ : \ \U^{\ast} \to \mathbf{Ab}
\\ & \widetilde{E}^n(X) \ \coloneqq \  \colimm_{k : \mathbb{N}}\lN{X \to_{\ast} \Omega^k{E_{n+k}}}\rN_0
\end{align*}
where the colimit $\colimm(G)$ of a sequence of abelian groups has underlying set $\colimm_{k:\N}(\pr_1(G_k))$ and has abelian group structure defined by induction on sequential colimits.
This is a cohomology theory acting on maps via the contravariant hom-action. The suspension axiom is easy to verify by using the identity $\left(n+1\right) + k = n + \left(k + 1\right)$. We now turn to verifying the exactness axiom.
\begin{definition}
Let $\left(A, a \right)$ be a $\U$-valued sequential diagram.
Let $n : \mathbb{N}$ and $x : A_n$. Define the \emph{lifting} function $x^{\left({-}\right)} : \prod_{m : \mathbb{N}} A_{n + m}$ by $x^0  \coloneqq x$ and $x^{m+1} \coloneqq a_{n+m}(x^m)$,
where $+ : \N \to \N \to \N$ is defined by pattern matching on the second argument.
\end{definition}
\begin{lemma}
Consider a levelwise exact sequence 
\[
\left(A, a\right) \xrightarrow{\left(m_1, M_1\right)} \left(B,b\right) \xrightarrow{\left(m_2, M_2\right)} \left(C,c\right)
\] of sequential diagrams valued in $\mathbf{Ab}$. Then the following sequence of abelian groups is exact:
\[
\colimm{A} \xrightarrow{\colimm(m_1)} \colimm(B) \xrightarrow{\colimm(m_2)} \colimm(C)
\]
\end{lemma}
\begin{proof} 
For each $k : \N$ and $x : A_k$,
$\colimm(m_2 \circ m_1 )(\iota_k(x))  \equiv  \iota_k(m_2(k,(m_1(k,x))))  =  0$ because $m_2(k) \circ m_1(k)$ is the zero map by levelwise exactness.

Next, let $k : \N$ and $x : B_k$. Suppose that $\colimm(m_2)(\iota_k(x)) = 0$ (where $\colimm(m_2)(\iota_k(x)) \equiv \iota_k(m_2(k,x))$). We want to show that the fiber of $\colimm(m_1)$ over $\iota_k(x)$ is inhabited. 
By \cite[Theorem 7.4]{SDR}, we have an equivalence
\[
\left(\iota_k(0)  =_{\colimm(C)} \iota_k(m_2(k,x))\right) \  \simeq \  \colimm_{n : \mathbb{N}}(0^{+n} =_{C_{k+n}} m_2(k,x)^{+n})
\] As $\iota_k(0) = 0$, it thus suffices to prove that the fiber is inhabited given an element of $\colimm_n(0^{+n} = m_2(k,x)^{+n})$. We proceed by induction on sequential colimits. Let $n : \N$ and $p : 0^{+n} = m_2(k,x)^{+n}$. By naturality of $m_2$, we see that $m_2(k,x)^{+n} = m_2(k+n, x^{+n})$. As $0^{+n} = 0$, it follows that $x^{+n}$ belongs to the kernel of $m_2(k+n)$. By levelwise exactness, this gives us an element $\left(d, q\right) : \fib_{m_1(k+n)}(x^{+n})$. We have that
\[
\colimm(m_1)(\iota_{k+n}(d)) \ \equiv \ \iota_{k+n}(m_1(k+n,d)) \ = \ \iota_{k+n}(x^{+n}) \ = \ \iota_k(x).
\] This proves that the fiber over $\iota_k(x)$ is inhabited.
\end{proof}
For exactness of $\widetilde{E}$, it now suffices to observe that when $k \geq 1$, the sequence
\[  
\lN{\sfrac{Y}{X} \to_{\ast} \Omega^k{E_{n+k}}}\rN_0 \xrightarrow{\lN{{-} \circ \cofiber(f)}\rN_0} \lN{Y \to_{\ast} \Omega^k{E_{n+k}}}\rN_0  \xrightarrow{\lN{{-} \circ f}\rN_0} \lN{X \to_{\ast} \Omega^k{E_{n+k}}}\rN_0
\] is exact for every $f : X \to_{\ast} Y$ (see \cite[Section 3.2.2]{Cav}).
\end{example}

\begin{remark}
A classical result of Adams states that for each $n : \Z$, the functor $\widetilde{E}^n({-})$ computes on finite CW-complexes the $\text{-}n$-th degree $\left[\Sigma^{\infty}({-}), E\right]_{-n}$  of the graded hom-group in the homotopy category of prespectra~\cite[Proposition 2.8]{BLUE}, where $\Sigma^{\infty}(X)$ denotes the suspension prespectrum of a pointed type $X$. For example, if $E$ is the sphere spectrum, then $\widetilde{E}^n({-})$ is precisely the $n$-th stable cohomotopy functor $\pi^n_S({-})$ on finite CW-complexes.
\end{remark}

\smallskip

\noindent We say that a cohomology theory $H$ is \emph{ordinary} if it satisfies
$H^n(S^0) \cong \1$ for all $n \ne 0$.
We say that it is \emph{additive} if the map
$\prod_{i : I}H^n(\inr \circ \left(i, {-}\right)) : H^n\mleft(\bigvee_{i : I}F_i\mright) \to \prod_{i : I}H^n(F_i)$
is an isomorphism for every set $I$ satisfying the set-level axiom of choice and every family $F: I \to \U^{\ast}$ of pointed types.

\begin{example}
Any $\Omega$-spectrum $E$ induces an additive cohomology theory $\widetilde{E}$. Indeed, the triangle
\[\begin{tikzcd}
	{\Omega^k{E_{n+k}}} && {\Omega^{k+1}{E_{n+\left(k +1\right)}}} \\
	& {E_n}
	\arrow["{\Omega^k(\epsilon_{n+k})}", from=1-1, to=1-3]
	\arrow["\simeq"', from=1-1, to=2-2]
	\arrow["\simeq", from=1-3, to=2-2]
\end{tikzcd}\] commutes by induction on $k$ like so: 
\[\begin{tikzcd}[row sep = 25]
	{\Omega^{k}(\Omega(E_{n+\left(k+1\right)}))} && {\Omega^{k+1}(\Omega(E_{n+\left(k +2\right)}))} \\
	{\Omega^k(E_{n+k})} && {\Omega^{k+1}(E_{n+\left(k+1\right)})} \\
	& {E_n}
	\arrow["{\Omega^{k}(\Omega(\epsilon_{n+\left(k+1\right)}))}", from=1-1, to=1-3]
	\arrow["{\Omega^k(\epsilon_{n+k})^{-1}}"', from=1-1, to=2-1]
	\arrow["{\Omega^{k}(\Omega(\epsilon_{n+\left(k+1\right)}))^{-1}}", from=1-3, to=2-3]
	\arrow["{\Omega^{k}(\epsilon_{n+k})}"{description}, from=2-1, to=2-3]
	\arrow[from=2-1, to=3-2]
	\arrow[from=2-3, to=3-2]
\end{tikzcd}\]
Thus, $\eqref{seqE}$ becomes the constant diagram at $\lN{X \to_{\ast} E_n}\rN_0$.  The induced theory is ordinary when $E$ is an Eilenberg-MacLane spectrum. 
\end{example}

\subsection{Cohomology sends finite colimits to weak limits}\label{cohomwklim}
Suppose that $H^{\ast}$ is a cohomology theory. Consider a pushout of pointed types and pointed maps:
\[\begin{tikzcd}
	Z & Y \\
	X & P
	\arrow["f"', from=1-1, to=2-1]
	\arrow["g", from=1-1, to=1-2]
	\arrow[from=2-1, to=2-2]
	\arrow[from=1-2, to=2-2]
	\arrow["\lrcorner"{anchor=center, pos=0.125, rotate=180}, draw=none, from=2-2, to=1-1]
\end{tikzcd}\]
Cavallo has constructed, within HoTT, the Mayer-Vietoris long exact sequence for cohomology~\cite[Section 4.5]{Cav}:
\[ 
\begin{tikzcd}[column sep =52,  /tikz/column 1/.append style={column sep=6pt},/tikz/column 5/.append style={column sep=6pt}] 
	{\cdots } & {H^{n-1}(Z)} & {H^n(P)} &[5 mm] {H^n(X) \times H^n(Y)} &[2 mm] {H^n(Z)} & \cdots
	\arrow[from=1-1, to=1-2]
	\arrow["{\partial}", from=1-2, to=1-3]
	\arrow["{\left(H^n(\inl) , H^n(\inr) \right)}", from=1-3, to=1-4]
	\arrow["{H^n(f) - H^n(g)}", from=1-4, to=1-5]
	\arrow[from=1-5, to=1-6]
\end{tikzcd} \]
Note that the type $\ker(H^n(f) - H^n(g))$ is by definition the pullback
\[\begin{tikzcd}
	{H^n(X) \times_{H^n(Z)}H^n(Y)} & {H^n(Y)} \\
	{H^n(X)} & {H^n(Z)}
	\arrow[from=1-1, to=1-2]
	\arrow[from=1-1, to=2-1]
	\arrow["\lrcorner"{anchor=center, pos=0.125}, draw=none, from=1-1, to=2-2]
	\arrow["{H^n(g)}", from=1-2, to=2-2]
	\arrow["{H^n(f)}"', from=2-1, to=2-2]
\end{tikzcd}\] Exactness states that the canonical map $ H^n(P) \to H^n(X) \times_{H^n(Z)} H^n(Y)$ induced by $\left(H^n(\inl) , H^n(\inr) \right)$ is surjective. Now, suppose that $\Gamma$ is a finite graph. It is known that cohomology preserves finite coproducts inside HoTT~\cite[Section 4.2]{Cav}.
By \cref{coeqcop,finprop}, Cavallo's long exact sequence induces an exact sequence
\[\begin{tikzcd}[column sep =30]
	{H^n(\colimm^{\ast}(F))} & {\prod_{i,j,g}H^n(F_i) \times \prod_{i}H^n(F_i)} & {\prod_{i,j,g}H^n(F_i) \times \prod_{i,j,g}H^n(F_i) }
	\arrow["{\zeta_n}", from=1-1, to=1-2]
	\arrow["{\mu_n -\nu_n}", from=1-2, to=1-3]
\end{tikzcd} \label{eq:ext} \tag{$\mathtt{ES}$} \]
for all $n : \N$. If $H^{\ast}$ is additive, this holds when $\Gamma$ is just a projective graph. (Projective types are also closed under $\Sigma$-types.) Here, $\zeta_n$ is the composite
\[\begin{tikzcd}[row sep = 25]
	{H^n(\colimm^{\ast}(F))} & {\prod_{i,j,g}H^n(F_i) \times \prod_{i}H^n(F_i)} \\
	{H^n\mleft(\bigvee_{i,j,g}F_i\mright) \times H^n\mleft(\bigvee_iF_i\mright)}
	\arrow["{{\zeta_n}}", dashed, from=1-1, to=1-2]
	\arrow["{\left(H^n(\inl) , H^n(\inr) \right)}"', from=1-1, to=2-1]
	\arrow["{{\cong \times \cong}}"', from=2-1, to=1-2]
\end{tikzcd}  \] and $\mu_n$ and $\nu_n$ are defined as
\begin{align*}
& \mu_n  \ : \ \prod_{i,j,g}H^n(F_i) \to  \prod_{i,j,g}H^n(F_i) \times \prod_{i,j,g}H^n(F_i)
\\ & \mu_n(h) \ \coloneqq \ \left(h   ,  h \right)
\\
& \nu_n   \ : \ \prod_{i}H^n(F_i) \to  \prod_{i,j,g}H^n(F_i) \times \prod_{i,j,g}H^n(F_i)
\\ & \nu_n(h)  \  \coloneqq \ \left(\lambda{i}\lambda{j}\lambda{g}.h_i   , \lambda{i}\lambda{j}\lambda{g}.H^n(F_{i,j,g})(h_j)  \right)
\end{align*} We have a cone $\mathcal{M}_{F, H, n}$
\[\begin{tikzcd}
	& {H^n(\colimm^{\ast}_{\Gamma}(F))} \\
	{H^n(F_j)} && {H^n(F_i)}
	\arrow["{H^n(F_{i,j,g})}"', from=2-1, to=2-3]
	\arrow["{H^n(\iota_j)}"', from=1-2, to=2-1]
	\arrow["{H^n(\iota_i)}", from=1-2, to=2-3]
\end{tikzcd}\]
over $H^n(F)$ and thus a commuting diagram
\[\begin{tikzcd}
	{H^n(\colimm^{\ast}_{\Gamma}(F))} && {\limm(H^n(F))} \\
	& {H^n(F_i)}
	\arrow["{\Delta^n_F}", dashed, from=1-1, to=1-3]
	\arrow["{H^n(\iota_i)}"', from=1-1, to=2-2]
	\arrow["{\pr_i}", from=1-3, to=2-2]
\end{tikzcd}\]
induced by the universal property of limits in $\mathbf{Ab}$. In fact, $\Delta^n_F$ is induced by the cone
\[\begin{tikzcd}
	{H^n(\colimm^{\ast}(F))} \\
	& {\limm(H^n(F))} & {\prod_iH^n(F_i)} \\
	\\
	& {\prod_{i,j,g}H^n(F_i)} & {\prod_{i,j,g}H^n(F_i) \times \prod_{i,j,g}H^n(F_i)}
	\arrow[dashed, from=1-1, to=2-2]
	\arrow["{\pr_2 \circ \zeta_n}", curve={height=-12pt}, from=1-1, to=2-3]
	\arrow["{\pr_1 \circ \zeta_n}"', curve={height=12pt}, from=1-1, to=4-2]
	\arrow[from=2-2, to=2-3]
	\arrow[from=2-2, to=4-2]
	\arrow["\lrcorner"{anchor=center, pos=0.125}, draw=none, from=2-2, to=4-3]
	\arrow["{\nu_n}", from=2-3, to=4-3]
	\arrow["{\mu_n}"', from=4-2, to=4-3]
\end{tikzcd}\]
The exactness of \eqref{eq:ext} states that $\Delta^n_F$ is surjective (equivalently, an epi in $\mathbf{Set}$). Classically, this implies that $\Delta^n_F$ has a section, so that $H^n(\colimm^{\ast}_{\Gamma}(F))$ is a \emph{weak limit} of $H^n(F)$ in $\mathbf{Set}$. If we assume the axiom of choice inside HoTT~\cite[Section 3.8]{Uni13}, then $\Delta^n_F$ \emph{merely} has a section (i.e., up to propositional truncation). In this case, we conclude that $H^n(\colimm^{\ast}_{\Gamma}(F))$ is merely a weak limit in $\mathbf{Set}$. In other words, the following function is surjective (not necessarily split) for each set $X$:
\[
\left( \mathcal{M}_{F, H, n} \circ {-}\right) \ : \  \left(X \to H^n(\colimm^{\ast}_{\Gamma}(F))\right) \to \mathsf{Cone}_{H^n(F)}(X)
\]

\appendix

\section{Structure identity principle}

We record our main tool for characterizing path spaces of structured types.

\begin{definition} \label{idsysdef}
Let $\left(A,a\right)$ be a pointed type. Consider a type family $B$ over $A$ and an element $b : B(a)$. We say that $\left(B, b\right)$ is an \textit{identity system} on $\left(A,a\right)$ if the total space $\sum_{x : A}B(x)$ is contractible.
\end{definition}

\begin{theorem}[{\cite[Theorem 11.2.2]{FTID}}] \label{idsys}
The following are logically equivalent.
\begin{itemize}
\item The family $B$ is an identity system on $\left(A,a\right)$.
\item The family $f : \prod_{x :A}\left(a = x\right) \to B(x)$ defined by $f(a, \refl_a) \coloneqq b$ is a family of equivalences.
\item For each family of types $P : \prod_{a :A}B(x) \to \U$, the following function has a section:
\[
h \mapsto h(a,b) \ : \ \left( \prod_{x :A}\prod_{y : B(x)}P(x,y)\right) \to P(a,b)
\] 
\end{itemize}
\end{theorem}

\begin{theorem}[{\cite[\href{https://unimath.github.io/agda-unimath/foundation.structure-identity-principle.html\#the-structure-identity-principle-1}{The structure identity principle}]{agda-unimath}}] \label{SIP}
Let $\left(A, a\right)$ be a pointed type, $\left(B,b\right)$ a pointed type family over $A$, and $\left(C,c\right)$ an identity system on $\left(A,a\right)$. Let $D  : \prod_{x :A}B(x) \to C(x) \to \U$ and $d  : D(a,b,c)$. If $\sum_{y : B(a)}D(a,y,c)$ is contractible, then the type family
\[
\left(x,y\right) \ \mapsto \ \sum_{z : C(x)}D(x,y,z)
\]
is an identity system on $\left(\sum_{x:A}B(x), \left(a,b\right)\right)$.
\end{theorem}

\section{Wild left adjoints and colimits} 

We briefly review the main results of \cite{2coherUF}, in particular that \emph{$2$-coherent} left adjoints between wild categories preserve colimits.

\begin{definition} \label{2coher}
Let $L : \C \to \D$ be a functor of wild categories and let $\left(\alpha, V_1, V_2\right) : L \dashv R$ be an adjunction (\cref{adjdef}). We say that $L$ is \emph{2-coherent} if for all $h_1 : \homm_{\D}(L(X), Y)$, $h_2 : \homm_{\C}(Z, X)$, and $h_3 : \homm_{\C}(W, Z)$, the following diagram commutes:
\[\begin{tikzcd}[column sep = huge]
	{\left( \alpha(h_1) \circ h_2\right) \circ h_3} && {\alpha(h_1) \circ \left(h_2 \circ h_3\right)} \\
	{\alpha(h_1 \circ L(h_2)) \circ h_3} && {\alpha(h_1 \circ L(h_2 \circ h_3))} \\
	{\alpha(\left(h_1 \circ L(h_2)\right) \circ L(h_3))} && {\alpha(h_1 \circ \left(L(h_2) \circ L(h_3) \right))}
	\arrow["{{\assoc(\alpha(h_1), h_2, h_3)}}", Rightarrow, no head, from=1-1, to=1-3]
	\arrow["{{V_2(h_2 \circ h_3, h_1)}}", Rightarrow, no head, from=1-3, to=2-3]
	\arrow["{{\ap_{{-} \circ h_3}(V_2(h_2,h_1))}}", Rightarrow, no head, from=2-1, to=1-1]
	\arrow["{{\ap_{\alpha}(\ap_{h_1 \circ {-}}(L_{\circ}(h_2,h_3)))}}", Rightarrow, no head, from=2-3, to=3-3]
	\arrow["{{V_2(h_3,h_1 \circ L(h_2))}}", Rightarrow, no head, from=3-1, to=2-1]
	\arrow["{{\ap_{\alpha}(\assoc(h_1, L(h_2), L(h_3)))}}"', Rightarrow, no head, from=3-1, to=3-3]
\end{tikzcd}\]
\end{definition}

Let $\C$ be a wild category and $\Gamma$ be a graph. Let $F$ be a $\Gamma$-shaped diagram in $\C$.

\begin{definition} \label{colimwildc}
A cocone $\left(C, r, K\right)$ under $F$ is \emph{colimiting} if for all $X : \obb(\C)$, the following function is an equivalence:
\begin{align*}
& \postcomp(C,r, K,X) \ : \ \homm_{\C}(C, X) \to \limm_{i : \Gamma^{\op}}(\homm_{\C}(F_i, X))
\\ & \postcomp(C,r,K,X,f) \ \coloneqq \ \left(\lambda{i}.f \circ r_i, \lambda{j}\lambda{i}\lambda{g}.\assoc(f, r_j, F_{i,j,g}) \cdot \ap_{f \circ {-}}(K_{i,j,g})  \right)
\end{align*} 
\end{definition}

Let $\D$ be a wild category, Let $L : \C \to \D$ and $R : \D \to \C$ be wild functors. Suppose that $\left(\alpha, V_1, V_2\right) : L\dashv R$ and that $\left(C,r,K\right)$ is a colimiting cocone under $F$. We have an induced cocone
\[\begin{tikzcd}[row sep = large]
	{L_0(F_i)} && {L_0(F_j)} \\
	& {L_0(C)}
	\arrow[""{name=0, anchor=center, inner sep=0}, "{L_1(F_{i,j,g})}", from=1-1, to=1-3]
	\arrow["{L_1(r_i)}"', from=1-1, to=2-2]
	\arrow["{L_1(r_j)}", from=1-3, to=2-2]
	\arrow["{L(K_{i,j,g})}"{description}, draw=none, from=0, to=2-2]
\end{tikzcd}\] under $L(F)$. Here, $L(K_{i,j,g}) \coloneqq L_{\circ}(r_j, F_{i,j,g})^{-1} \cdot \ap_{L_1}(K_{i,j,g})$. 

\begin{theorem}\label{LAPC}
If $L$ is $2$-coherent, then the cocone $\left(L_0(C), L_1(r), L(K)\right)$ under $L(F)$ is colimiting. 
\end{theorem}

\begin{theorem} 
The suspension $\Sigma : \U^{\ast} \to \U^{\ast}$ is a $2$-coherent left adjoint to the loop space functor.
\end{theorem}

\begin{corollary}\label{suspcol}
Suspension preserves colimits.
\end{corollary}

\bibliographystyle{plainurl}
\bibliography{colimits-papers}

\begin{thebibliography}{10}

\bibitem{BLUE}
J.~F. Adams.
\newblock {\em Stable homotopy and generalised homology}.
\newblock University of Chicago Press, Chicago, Ill., 1974.
\newblock Chicago Lectures in Mathematics.

\bibitem{Avi}
Jeremy Avigad, Krzysztof Kapulkin, and Peter~LeFanu Lumsdaine.
\newblock Homotopy limits in type theory.
\newblock {\em Mathematical Structures in Computer Science}, 25(5):1040–1070,
  2015.
\newblock \href {https://doi.org/10.1017/S0960129514000498}
  {\path{doi:10.1017/S0960129514000498}}.

\bibitem{acyc}
Ulrik Buchholtz, Tom de~Jong, and Egbert Rijke.
\newblock {Epimorphisms and Acyclic Types in Univalent Foundations}.
\newblock {\em The Journal of Symbolic Logic}, page 1–36, 2025.
\newblock \href {https://doi.org/10.1017/jsl.2024.76}
  {\path{doi:10.1017/jsl.2024.76}}.

\bibitem{cohom}
Ulrik Buchholtz and Kuen-Bang Hou~(Favonia).
\newblock {Cellular Cohomology in Homotopy Type Theory}.
\newblock {\em {Logical Methods in Computer Science}}, {Volume 16, Issue 2},
  2020.
\newblock \href {https://doi.org/10.23638/LMCS-16(2:7)2020}
  {\path{doi:10.23638/LMCS-16(2:7)2020}}.

\bibitem{Higher}
Ulrik Buchholtz, Floris van Doorn, and Egbert Rijke.
\newblock {Higher Groups in Homotopy Type Theory}.
\newblock In {\em Proceedings of the 33rd Annual ACM/IEEE Symposium on Logic in
  Computer Science}, LICS '18, page 205–214, New York, NY, USA, 2018.
  Association for Computing Machinery.
\newblock \href {https://doi.org/10.1145/3209108.3209150}
  {\path{doi:10.1145/3209108.3209150}}.

\bibitem{Cav}
Evan Cavallo.
\newblock {Synthetic Cohomology in Homotopy Type Theory}.
\newblock Master's thesis, Carnegie Mellon University, 2015.
\newblock URL: \url{https://ecavallo.net/works/thesis15.pdf}.

\bibitem{agda-colim-TR}
Perry Hart.
\newblock Formal proofs related to coslice colimits and 2-coherent left
  adjoints.
\newblock \url{https://github.com/PHart3/colimits-agda/tree/v0.5.0}, 2026.

\bibitem{2coherUF}
Perry Hart.
\newblock {On Left Adjoints Preserving Colimits in HoTT}.
\newblock In Stefano Guerrini and Barbara K\"{o}nig, editors, {\em 34th EACSL
  Annual Conference on Computer Science Logic (CSL 2026)}, volume 363 of {\em
  Leibniz International Proceedings in Informatics (LIPIcs)}, pages
  20:1--20:17, Dagstuhl, Germany, 2026. Schloss Dagstuhl -- Leibniz-Zentrum
  f{\"u}r Informatik.
\newblock URL:
  \url{https://drops.dagstuhl.de/entities/document/10.4230/LIPIcs.CSL.2026.20},
  \href {https://doi.org/10.4230/LIPIcs.CSL.2026.20}
  {\path{doi:10.4230/LIPIcs.CSL.2026.20}}.

\bibitem{3x3}
Daniel~R. Licata and Guillaume Brunerie.
\newblock {A Cubical Approach to Synthetic Homotopy Theory}.
\newblock In {\em 2015 30th Annual ACM/IEEE Symposium on Logic in Computer
  Science}, pages 92--103, 2015.
\newblock \href {https://doi.org/10.1109/LICS.2015.19}
  {\path{doi:10.1109/LICS.2015.19}}.

\bibitem{pab}
{nLab authors}.
\newblock pointed abelian group.
\newblock \url{https://ncatlab.org/nlab/show/pointed+abelian+group}, November
  2024.
\newblock
  \href{https://ncatlab.org/nlab/revision/pointed+abelian+group/3}{Revision 3}.

\bibitem{FTID}
Egbert Rijke.
\newblock {Introduction to Homotopy Type Theory}, 2022.
\newblock \href {https://arxiv.org/abs/2212.11082} {\path{arXiv:2212.11082}}.

\bibitem{RS}
Egbert Rijke, Michael Shulman, and Bas Spitters.
\newblock {Modalities in homotopy type theory}.
\newblock {\em {Logical Methods in Computer Science}}, {Volume 16, Issue 1},
  January 2020.
\newblock \href {https://doi.org/10.23638/LMCS-16(1:2)2020}
  {\path{doi:10.23638/LMCS-16(1:2)2020}}.

\bibitem{agda-unimath}
Egbert Rijke, Elisabeth Stenholm, Jonathan Prieto-Cubides, Fredrik Bakke, and
  {others}.
\newblock {The agda-unimath library}.
\newblock URL: \url{https://github.com/UniMath/agda-unimath/}.

\bibitem{SDR}
Kristina Sojakova, Floris~van Doorn, and Egbert Rijke.
\newblock {Sequential Colimits in Homotopy Type Theory}.
\newblock In {\em Proceedings of the 35th Annual ACM/IEEE Symposium on Logic in
  Computer Science}, LICS '20, page 845–858, 2020.
\newblock \href {https://doi.org/10.1145/3373718.3394801}
  {\path{doi:10.1145/3373718.3394801}}.

\bibitem{Uni13}
The {Univalent Foundations Program}.
\newblock {\em {Homotopy Type Theory: Univalent Foundations of Mathematics}}.
\newblock \url{https://homotopytypetheory.org/book}, Institute for Advanced
  Study, 2013.

\bibitem{DRB}
Floris van Doorn, Jakob von Raumer, and Ulrik Buchholtz.
\newblock {Homotopy Type Theory in Lean}.
\newblock In {\em Interactive Theorem Proving}, pages 479--495. Springer
  International Publishing, 2017.
\newblock \href {https://doi.org/10.1007/978-3-319-66107-0_30}
  {\path{doi:10.1007/978-3-319-66107-0_30}}.

\end{thebibliography}

\end{document}